\documentclass[11pt,a4paper]{article}

\usepackage[utf8]{inputenc}
\usepackage[T1]{fontenc}
\usepackage[margin=1in]{geometry}
\usepackage{microtype}
\usepackage{amsmath,amsthm,amsfonts}		% Math features
\usepackage{graphicx}
\usepackage{mathtools}
\usepackage{dsfont}
\usepackage{bbold}
\usepackage[dvipsnames]{xcolor}
\usepackage{hyperref}
\hypersetup{
	hidelinks,
    hypertexnames = false,
    colorlinks = true,
    linkcolor = Red,
    citecolor = Red
}
\usepackage[ruled,vlined]{algorithm2e}
\usepackage{float}
\usepackage{makecell} % to break lines in tables
\usepackage{comment}
\usepackage[nameinlink]{cleveref}
\newcommand{\alginput}[1]{{\bf Input}: #1 \\ \vspace{3mm}}

\newtheorem{theorem}{Theorem}
\newtheorem{lemma}[theorem]{Lemma}
\newtheorem*{lemma*}{Lemma}
\newtheorem{claim}[theorem]{Claim}

\theoremstyle{definition}
\newtheorem{definition}[theorem]{Definition}

\theoremstyle{remark}

\crefname{claim}{Claim}{Claims}
\crefname{lemma}{Lemma}{Lemmas}
\crefname{observation}{Observation}{Observations}
\crefname{equation}{Eq.}{Eqs.}
\crefname{algocf}{Algorithm}{Algorithms}
\Crefname{algocf}{Algorithm}{Algorithms}

\newcommand{\pa}[1]{\left( #1 \right)}
\DeclareMathOperator{\bbN}{\mathbb{N}}
\DeclareMathOperator{\bbR}{\mathbb{R}}
\DeclareMathOperator{\bbE}{\mathbb{E}}

\global\long\def\pull{\mathcal{PULL}}
\global\long\def\push{\mathcal{PUSH}}
\renewcommand{\Pr}{\mathbb{P}}

\newcommand{\binomial}{\mathrm{Binomial}}

\newcommand{\argmax}{\mathrm{argmax}}

\newcommand{\Y}{\tilde{Y}}

\newcommand{\R}{\mathcal{R}}

\definecolor{myOrange}{rgb}{0.7,0.3,0}

\newcommand{\nnote}[1]{{\color{magenta} ({\bf Niccolò:} #1)}}

\usepackage{xspace}
\newcommand{\SF}{\hyperref[sec:SF]{SF}\xspace}
\newcommand{\USF}{\hyperref[sec:USF]{USF}\xspace}
\newcommand{\SSF}{\hyperref[sec:SSF]{SSF}\xspace}

\newcommand{\func}[1]{\left[ \begin{array}{c} #1 \end{array} \right]}
\newcommand{\Rad}{\mathrm{Rad}}

\title{\texorpdfstring{Fast and Robust Information Spreading in the Noisy $\pull$ Model}{Fast and Robust Information Spreading in the Noisy PULL Model}}

\author{
Niccolò D'Archivio\thanks{Université Côte d'Azur, Inria, CNRS, i3s, Sophia-Antipolis, France. \href{mailto:niccolo.darchivio@inria.fr}{niccolo.darchivio@inria.fr}. ORCID: \href{https://orcid.org/0009-0005-9491-2928}{0009-0005-9491-2928}.}
\and
Amos Korman\thanks{Department of Computer Science, University of Haifa, Haifa, Israel. \href{mailto:amos.korman@gmail.com}{amos.korman@gmail.com}. ORCID: \href{https://orcid.org/0000-0001-8652-9228}{0000-0001-8652-9228}.}
\and
Emanuele Natale\thanks{Université Côte d'Azur, CNRS, Inria, i3s, Sophia-Antipolis, France. \href{mailto:emanuele.natale@inria.fr}{emanuele.natale@inria.fr}. ORCID: \href{https://orcid.org/0000-0002-8755-3892}{0000-0002-8755-3892}.}
\and
Robin Vacus\thanks{Sorbonne Université, CNRS, LIP6, Paris, France. \href{mailto:robin.vacus@gmail.com}{robin.vacus@gmail.com}. ORCID: \href{https://orcid.org/0000-0002-7368-4912}{0000-0002-7368-4912}.}
}

\date{}

\begin{document}

\maketitle
%TODO mandatory: add short abstract of the document
\begin{abstract}
%Understanding how information can efficiently spread in distributed systems under stochastic and faulty conditions is a fundamental question in multi-agent systems. We investigate the tradeoff between communication capacity and robustness in well-mixed populations.

Efficient information spreading in stochastic multi agent systems is a core challenge when communication is noisy, bandwidth limited, and agents lack global coordination. Yet biological systems, including ant colonies and fish schools, routinely overcome these constraints. A small number of informed individuals can reliably guide large, uncoordinated populations using minimal, noisy signals. Motivated by these observations, we study how reliable information dissemination can be achieved in such bio inspired settings.

A population of $n$ agents, each with a binary preference, includes a designated subset of source agents, and the goal is to converge to the majority preference among the sources. In the noisy $\pull(h)$ model, each agent observes noisy messages from $h$ randomly sampled peers in every round. Prior work shows that convergence requires $\Omega(n/h)$ rounds even under favorable conditions. We ask how far we can push simplicity, with no synchronization at the start time and minimal message size, without compromising convergence speed.

We present a quasi self-stabilizing protocol using only 2-bit messages that converges from arbitrary initial states despite severe noise and initial asynchrony. It achieves optimal convergence time $O((n/h)\log n)$ with high probability, and in particular $O(\log n)$ time in the snapshot regime $h=\Theta(n)$. A key subroutine is an even simpler 1-bit protocol assuming simultaneous start, based on a natural two phase listen then amplify mechanism.

Together, our results show that simple, biologically inspired protocols can achieve optimal and robust information dissemination even in highly unreliable and uncoordinated systems.
\end{abstract}


\section{Introduction}
Designing artificial multi-agent systems is challenging when agents must achieve distributed coordination and consensus under local, asynchronous, and noisy communication constraints. Yet, nature offers inspiration: for example, in ant colonies and fish schools, a few informed individuals can lead an entire uncoordinated group using minimal, noisy signals \cite{gelblum2015ant,razin,gelblum2016emergent,couzin2005effective}. Further details on the biological motivation are provided in the \Cref{house-hunting}. %Supplementary Material (SM), Section 8.
This motivates the study of simple, fault-tolerant algorithms for information spreading in settings with unreliable communication and severe bandwidth limits \cite{DBLP:journals/dc/FraigniaudN19,boczkowski2018limits}.

In the spirit of {\em opinion dynamics}, designed to model information propagation and consensus formation in multi-agent systems~\cite{becchetti2020consensus,Cruciani_AAMAS18,Cruciani_AAAI19,lesfariBiasedMajorityOpinion2022,Gauy_AAMAS25}, we consider a population of $n$ agents, each holding a binary opinion ($0$ or $1$). We assume that a subset of agents, called {\em sources}, hold initial preferences. We define the correct opinion as the one supported by the majority of sources. The goal is for the entire population to converge on this correct opinion. This formulation generalizes the standard rumor spreading problem (where sources agree on the correct opinion) to include majority consensus (where all agents are sources). We refer to this broader problem as {\em majority rumor spreading} for simplicity.
%information dissemination.

\paragraph*{Noise in communication}
To capture random, unpredictable, and noisy communication, we adopt the noisy PULL($h$) model \cite{Breathe,DBLP:journals/dc/FraigniaudN19,clementiConsensusVsBroadcast2020,boczkowski2018limits}. In this model, time proceeds in synchronous rounds, and in each round every agent pulls messages from $h$ randomly chosen peers (neighbors) in the network. Each message is drawn from a fixed alphabet $\Sigma$. Crucially, the channel is noisy: the received symbol may be corrupted randomly. We assume the noise is biased, i.e.
%it retains the original message with probability slightly greater than what a completely random guess would yield \cite{boczkowski2018limits}.
%\rnote{I suggest: ``it retains the original message with probability slightly greater than $1/|\Sigma|$.'' I think it's slightly clearer than the current text.}
it retains the original message with probability slightly greater than $1/|\Sigma|$.
Despite this slight reliability, without further algorithmic design the information can easily get lost in the noise. This raises a central question: How can the agents quickly and reliably agree on the correct opinion under such extreme noise, using small $|\Sigma|$?

One difficulty that arises in the noisy PULL($h$) model is that sources cannot safely signal the fact that they are sources using a designated bit, since this bit itself would be noisy.
%Moreover, when sources are few, most received samples would falsely appear to come from  sources, even though they originate from non-sources.
%\rnote{This last sentence seems wrong. Consider replacing with ``As a consequence, the majority of messages perceived as originating from a source would in fact be artifacts of noise, rather than genuine source messages.''}
As a consequence, the majority of messages perceived as originating from a source would in fact be artifacts of noise, rather than genuine source messages.
Conversely, without the ability to identify sources, non-sources must somehow extract information from multiple samples, most of which are unreliable.

Notably, prior work established a fundamental time barrier: even under relatively favorable conditions, any protocol requires on the order of $\Omega(n/h)$ rounds to spread the information with constant success probability \cite{boczkowski2018limits}. In other words, if each agent communicates with only a constant number of others per round ($h=O(1)$), it takes linear time $\Omega(n)$ in the population size to reliably propagate a rumor; larger $h$ may potentially linearly speed this up, but no algorithm can circumvent this bound.
%Our goal is to approach this optimal convergence speed using as little communication per round as possible (both in terms of $h$ and the message size), and to do so under minimal assumptions about coordination among agents.
%\rnote{In this paper $h$ is a parameter of the setting. I find it weird to say that we try to minimize it. Why not ``Our goal is to approach this optimal convergence speed using the smallest alphabet $\Sigma$, and to do so under minimal assumptions about coordination among agents.''}
Our goal is to approach this optimal convergence speed using the smallest alphabet $\Sigma$, and to do so under minimal assumptions about coordination among agents.

\paragraph*{Quasi self-stabilization}
To model the lack of global coordination, we consider a variant of the self-stabilization framework \cite{altisen_introduction_2022}.
A self-stabilizing protocol ensures convergence from any initial state—potentially set by an adversary—without assuming synchronized clocks or a common start time. However, the traditional model assumes complete adversarial control over all initial internal states, which may be unnecessarily strong for modeling real systems. We therefore introduce a relaxed version -- quasi self-stabilization -- that allows reliable access to certain information, such as the size of the population~$n$, the distribution of the noise, or whether an agent is a source (see more details in Section \ref{sec:problem}).

When agents start the execution at different times, several technical challenges arise. First, without synchronous  counters, it is difficult to implement phase-based strategies. Early transmissions from sources may occur before some agents have even begun executing the protocol, rendering those messages effectively lost. Moreover, if agents reset their state upon “wake-up,” they risk disrupting ongoing consensus dynamics, since different subsets of the population may be at different logical times. This asynchrony complicates both correctness and efficiency guarantees: protocols must ensure convergence despite arbitrary misalignment of agents’ clocks, while simultaneously tolerating noisy communication.

Previous studies have explored self-stabilizing rumor spreading in the $\pull$ model, focusing on minimizing message size \cite{boczkowskiMinimizingMessageSize2017,bastide2021self,korman2022early,becchettiRoleMemoryRobust2023}, but without addressing noisy communication. Conversely, Boczkowski et al.\ \cite{boczkowski2018limits} analyzed the noisy $\pull(h)$ model and showed that with any constant-size alphabet, convergence still requires $\tilde{\Omega}(n/h)$ rounds, even assuming synchronized start. In this work, we address both dimensions of unreliability: noisy messages and arbitrary, uncoordinated initial configurations.

\paragraph*{Main contribution}
We answer the above challenge by designing protocols for noisy majority rumor spreading, emphasizing a tiny message alphabet and robustness to communication noise and/or arbitrary start conditions.

%\paragraph{1-Bit Source Filter (SF) Protocol.} We introduce a simple dissemination protocol that uses only a binary alphabet ($|\Sigma|=2$) – effectively 1-bit messages. Assuming all agents begin in a synchronized start (simultaneous wake-up) and have rough knowledge of the network size (an approximation of $\log n$), SF guarantees that all agents converge to the correct opinion in near-optimal time. In general, the convergence time is $O\left(\frac{n}{h}\log n\right)$ rounds with high probability (w.h.p.).

%\paragraph{Self-Stabilizing 4-Symbol Protocol (SSF).}
Our main result is a quasi self-stabilizing majority rumor spreading algorithm, called SSF, that uses a 4-symbol alphabet (only 2 bits of information per message) and, under the reasonable source regime in which a constant fraction of agents is non-sources, operates correctly from arbitrary initial states: even if agents wake up at different times or start with arbitrary opinions, the system will converge to the truth. It achieves convergence in $O\left(\frac{n}{h}\log n\right)$ rounds w.h.p. In the extreme well-connected case of $h = n$, it achieves $O(\log n)$ rounds, an exponential improvement over naive gossip. For general $h$, the upper bound matches the $\Omega(n/h)$ lower bound up to a logarithmic factor, which could be attributed to high probability guarantee of SSF, versus the constant probability guarantee of the lower bound.

%Notably, our protocol works when there is only one source with the correct bit and all other agents initially holding the wrong opinion (bias $s=1$) – whereas many previous opinion dynamics protocols involving communication noise required a sizeable initial majority bias to succeed. We also show how if could be adapted to scenarios where agents lack exact knowledge of $n$ (knowing only an estimate of $\log n$ and the noise parameters), incurring only an extra $O(\log n)$ factor in running time (Theorem 5).
%When communication is complete (e.g. $h = n$), it again runs in $O(\log n)$ rounds. Notably, this guarantee holds without requiring a large initial bias – the protocol works even if the correct opinion initially has only one supporter.

To our knowledge, SSF is the first protocol to solve noisy majority rumor spreading with such a small message alphabet in a quasi self-stabilizing manner, advancing the state of the art in fault-tolerant consensus in multi-agent systems subject to strong noise.

Two points carry most of the technical work. First, the technical master lemma (\Cref{lemma:first_phase}) gives, across the full range of noise levels~$\delta$, a precise $\delta$-dependent sample requirement for producing the small bias that majority boosting subsequently amplifies. This estimate is not a routine application of standard concentration or normal-approximation tools such as the Berry--Esseen theorem: both the frequency of informative samples and their conditional bias change between the source-dominated and noise-dominated regimes. Second, for \SSF we prove that the agents' weak opinions remain mutually independent even though non-sources transmit their weak opinions. This independence is what makes the subsequent majority amplification valid.

\subsection{Overview of key algorithmic ideas}

SSF (which is formally described in Sec. \ref{sec:SSF} and Alg. \ref{alg:SSF}) builds on a simple 1-bit protocol, SF, which assumes a simultaneous start and achieves fast majority rumor spreading through a two-phase process. SSF extends this method to the quasi self-stabilizing setting by embedding lightweight synchronization directly into the protocol, eliminating the need for a global clock.

The two-phase process of SF may be of independent interest. In Phase 1, agents collect $h$ samples per round to form a weak opinion—a guess correct with probability slightly above $1/2$. This is achieved by having sources continuously broadcast their bit, while non-sources contribute minimal interference. Weak opinions are computed using only first-hand information—direct observations of other agents—to maintain independence and avoid bias from relayed noise. This is informally further explained in \Cref{sec:intuition} and formalized in \Cref{lemma:first_phase,lemma:first_phase_SF,lemma:correctness_weak}.

In Phase 2, agents share their weak opinions and update by majority rule, amplifying the bias to achieve consensus with high probability. This boosting step resembles mechanisms used in prior works on majority dynamics under noise \cite{Breathe,DBLP:journals/dc/FraigniaudN19}.

For such a procedure to work, it is desirable that all weak-opinions are mutually independent, so that the majority of weak-opinions would be correct with high probability.
However, having non-sources relay information between themselves can cause dependencies. Therefore, in our algorithms, weak-opinions are computed in a way that is oblivious of the information obtained (and the randomness involved) by other non-sources. Hence, in a sense, the weak-opinions are based on ``first-hand'' information, that is, information that comes directly from the sources, despite the fact that sources are not easily detectable.

%The self-stabilizing algorithm SSF repeats such phases in a cyclic fashion and uses the extra symbols (beyond the binary opinions) to coordinate phase transitions and recover from any inconsistent states. Essentially, the 4-symbol alphabet encodes both an opinion bit and minimal phase signals that allow agents to stay in sync even without a global clock.

The quasi self-stabilizing protocol SSF implicitly cycles through these phases  and uses the additional symbols (beyond binary opinions) to coordinate phase transitions and recover from inconsistent states. Instead of relying on an external synchronization subroutine, as in \cite{boczkowskiMinimizingMessageSize2017}, which may be difficult to implement in the presence of noise, SSF integrates coordination directly into the message structure. The first bit indicates whether the sender is a source (1) or not (0), while the second bit conveys the weak opinion. Of course, upon observation, these bits may be flipped due to noise.
Agents compute weak opinions by applying a majority vote over the second bits of messages marked as coming from sources, and update their actual opinion, which is kept internally, via a majority over all received second bits, regardless of the value of their first bit. This design allows SSF to converge  in near-optimal time, despite both noise and lack of coordination, using only 2-bit messages.
\begin{figure}[h]
    \centering
    \includegraphics[width=0.5\textwidth]{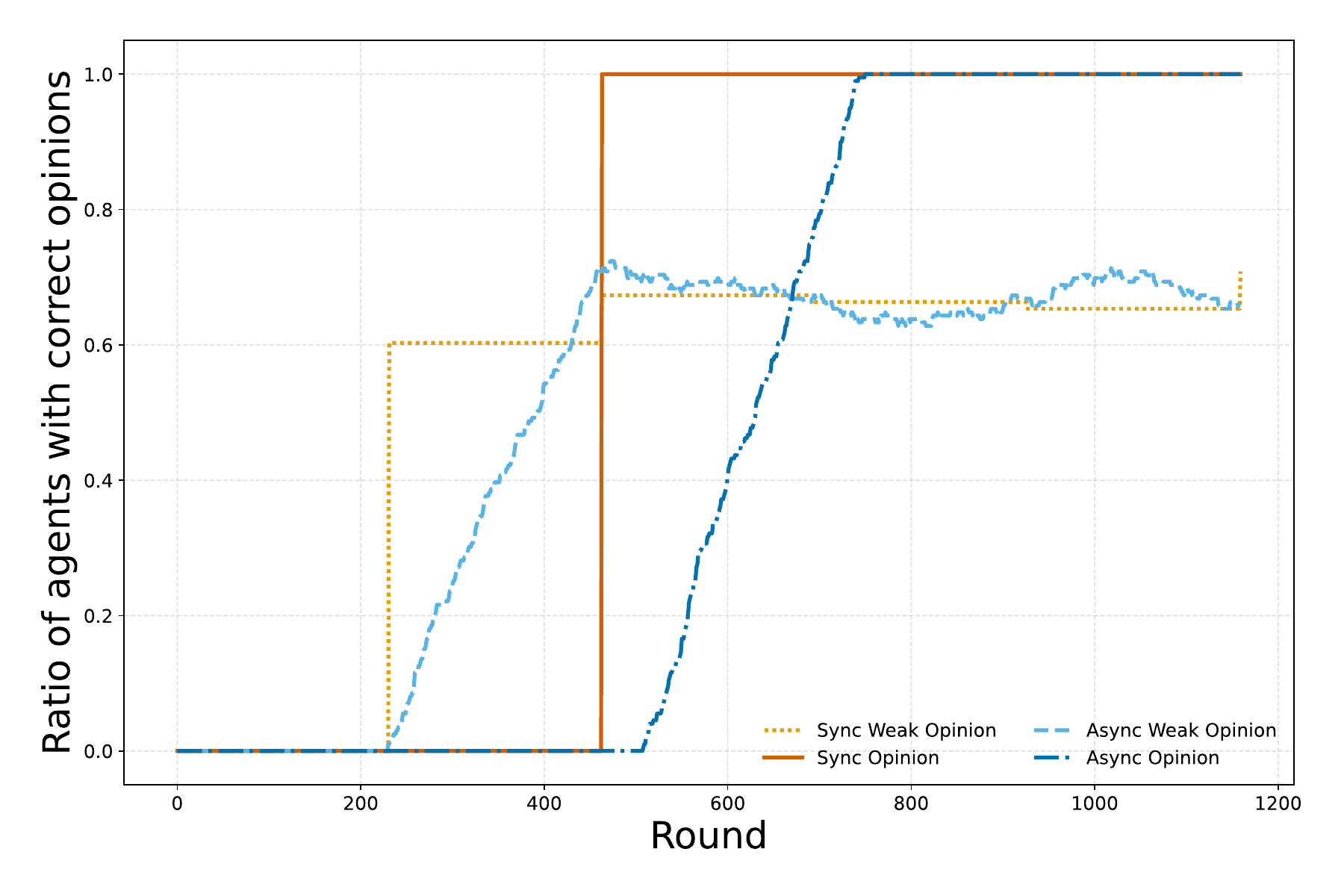}
    \caption{Dynamics of weak and actual opinions ($n=200$, $s_1=s=1$, $h=10$, and $\delta = 0.2$): In the synchronous case (warm colors), agents start with empty memory, and execute SSF; in the asynchronous case (cold colors), agents begin with a random number of worst-case messages (incorrect opinions labeled as from a source) stored in their memory, and execute SSF.}
    \label{fig:selfstab_plots}
\end{figure}
%\nnote{Move forward this part?} \rnote{I don't feel like this is the best place. I would move in Section 1.1 where we explain the main algorithmic ideas, as an illustration. Alternatively, it could be moved to Section 2.}
Fig.~\ref{fig:selfstab_plots} depicts the dynamics of weak and actual opinions under synchronous and asynchronous starts, corresponding to aligned or misaligned execution phases of \SSF.

\subsection{Problem Definition}\label{sec:problem}

Consider a set $I=\{1, \dots, n\}$ of agents, where each agent~$i$ holds a binary opinion~$Y^{(i)} \in \{0,1\}$.
The population contains a subset called {\em source agents}. Each source agent is initially given a {\em preference} in $\{0,1\}$ (which may differ from its opinion) and will never change it.
We let~$s_1$ (resp.~$s_0$) denote the number of sources with preference~$1$ (resp.~$0$). The {\em bias} is defined as~$s := |s_1 - s_0|$, and we require that $s\geq 1$.
The preference supported by the strict majority of sources is called {\em correct}.
The goal is for all agents, including sources, to eventually adopt the correct preference as their opinion as fast as possible.

Observe that the challenge of retrieving the correct opinion given a certain bias ($s=|s_0-s_1|$) can be very different depending on  the number of sources ($s_0+s_1$). To illustrate this point, consider two cases: (1) $s_0=0$, $s_1=1$, and (2) $s_0=\lfloor n/2 \rfloor-1$, $s_1=n-s_0$. In Case 1 the system attempts to retrieve a bit from a single source agent, while in Case 2 all agents are sources, making the notion of sources meaningless, and leading to a majority consensus problem.

\paragraph*{The Noisy \texorpdfstring{$\pull(h)$}{PULL(h)} model} Agents exchange information over the complete graph $K_n$, using an alphabet of messages, termed $\Sigma$, which is assumed to be of constant size. %It might contain messages representing opinions~$0$ and~$1$.
Time proceeds in discrete rounds, where in every round the following happens.
\begin{enumerate}
    \item Each agent (source or non-source) chooses a message $\sigma\in \Sigma$ to display.

    \item Every agent samples~$h$ agents in~$I$, uniformly at random with replacement. In particular, Agent~$i$ may sample twice the same agent, and may also sample itself.

    \item
          Every agent receives a noisy version of the messages displayed by its sampled peers. Formally, we model the noise using a stochastic matrix ~$N \in \bbR^{|\Sigma| \times |\Sigma|}$, called {\em noise matrix}, where each row sums to 1 and all entries are non-negative. If Agent~$i$ samples Agent~$j$, displaying a message~$\sigma \in \Sigma$, then Agent~$i$ observes message~$\sigma'$ with probability~$N_{\sigma, \sigma'}$.

          % Every agent receives noisy versions of the messages displayed by the sampled agents. Specifically, we consider a stochastic\footnote{A matrix is stochastic if all coefficients are non-negative and each row sums to 1.} matrix~$N \in \bbR^{|\Sigma| \times |\Sigma|}$, called {\em noise matrix}, so that if Agent~$i$ samples Agent~$j$, displaying a message~$\sigma \in \Sigma$, then Agent~$i$ observes message~$\sigma'$ with probability~$N_{\sigma, \sigma'}$.
    \item Every agent $i$ may update its opinion $Y^{(i)}$, and its internal state (that may include counters, and other variables).
\end{enumerate}
The decisions of an agent may depend on all the messages that it received so far, and possibly on the agent's coin tosses (in the case that the agent's protocol is probabilistic).
We assume that the random events involved in sampling, noise, and agents' coin tosses are all independent (across both rounds and agents).
As a standard convention, we say that an event~$A$ happens with high probability (w.h.p.) if $\Pr(A) = 1 - \mathcal{O} (1/n^2)$.

\begin{definition}
    \label{def:noisetypes}
    For~$\delta \in [0,1/|\Sigma|]$ and a (stochastic) noise matrix~$N$, we say that
    \begin{itemize}
        \item $N$ is~$\delta$-lower bounded if~$N_{\sigma,\sigma'} \geq \delta$ for every~$\sigma,\sigma' \in \Sigma$.

        \item $N$ is~$\delta$-upper bounded if
              \begin{equation} \label{eq:N_def}
                  N_{\sigma,\sigma'} \leq \delta, \quad \text{for every } \sigma \neq \sigma' \in \Sigma,
              \end{equation}
              which, in particular, implies that $N_{\sigma,\sigma} \geq 1-(|\Sigma|-1) \delta$. %, for every $\sigma \in \Sigma$.
        \item $N$ is~$\delta$-uniform if there is equality in \Cref{eq:N_def}, i.e.,
              \begin{equation*}
                  N_{\sigma,\sigma'} = \delta, \quad \text{for every } \sigma \neq \sigma' \in \Sigma,
              \end{equation*}
              which, in particular, implies $N_{\sigma,\sigma} = 1-(|\Sigma|-1) \delta$. %, for every $\sigma \in \Sigma$.
    \end{itemize}
\end{definition}

\paragraph*{Quasi self-stabilizing setting}

%In order to capture the lack of built-in synchronized clocks, and in particular, the inability to know when the process starts, we consider a somewhat relaxed version of the classical notion of self-stabilization in distributed computing \cite{DBLP:journals/cacm/Dijkstra74,schneider1993self}.
We assume that at time~$0$ (unknown to the agents) an adversary may choose arbitrary values for every mutable local variable. In \SSF, these variables are the opinion~$Y^{(i)}$, the weak opinion~$\Y^{(i)}$, and the contents and size of the message multiset~$M^{(i)}$; in particular, the multiset may initially contain fake samples purportedly gathered before time~$0$. More generally, mutable counters, clocks, and phase variables, if present, may also be corrupted independently. Thus, agents need not share a logical starting time.

The problem instance and the read-only protocol parameters are not mutable state. The source set and source preferences may be chosen arbitrarily as part of the instance, but thereafter each agent reliably knows whether it is a source and, if so, its own fixed preference. Agents also have hard-wired access to a constant approximation of~$n$ and to the noise matrix~$N$. While precise knowledge of~$N$ is required in the general case, for symmetric noise an upper bound on the noise parameter~$\delta$ suffices\footnote{Symmetric noise can arise, for example, due to symmetric environmental conditions, as appears to be the case in cooperative transport by ants (see more details in Section \ref{sec:conclusions}). In general, when considering biological contexts, knowledge assumptions (such as an upper bound on~$\delta$) may plausibly reflect information encoded through evolutionary adaptation.}.

We call the setting quasi self-stabilizing because only the mutable state is adversarially corrupted. The protected source role and preference are natural immutable inputs: they define the problem instance and hence the correct preference, rather than representing state generated by the protocol. Similarly, the estimates of~$n$ and~$N$ are fixed environmental parameters used to configure the protocol. For any such instance, the preference that is more frequent among sources is called the {\em correct} preference, and correctness requires every agent to adopt it.

\paragraph*{Convergence}
%[Convergence]    \label{def:convergence}
We say that a protocol solves the {\em noisy majority rumor spreading} problem in time~$\mathcal{T}$ w.h.p~if all agents adopt the correct preference as their opinion in at most~$\mathcal{T}$ rounds w.h.p. Note that with this definition, even the sources whose preference is incorrect must adopt as their opinion the correct preference. We additionally say that a protocol solves the problem in time~$\mathcal{T}$ in a quasi self-stabilizing manner if, starting with any initial configuration of the mutable state (in the sense described above), all agents adopt the correct preference as their opinion in at most~$\mathcal{T}$ rounds and retain it for a polynomial number of rounds, say $n^3$ rounds, w.h.p.

%\begin{remark}
%  It is crucial to emphasize that the challenge of retrieving the correct opinion given a certain bias ($s=|s_0-s_1|$) can be very different depending on  the number of sources ($s_0+s_1$). To illustrate this point, consider two cases: (1) $s_0=0$, $s_1=1$, and (2) $s_0=\lfloor n/2 \rfloor$, $s_1=s_0+1$ (assuming  $n$ is odd). In Case 1 the system attempts to retrieve a bit from a single source agent, while in Case 2 all agents are sources, making the notion of sources meaningless, and leading to a majority consensus problem.
%\end{remark}

%

\subsection{Lower bound on convergence time}

% We consider the task of broadcasting a binary value from a set of (possibly conflicting) {\em source} agents to the rest of the population, under the noisy $\pull(h)$ model of communication, where the goal is to converge to the plurality opinion of the sources.

%We provide upper bounds for the noisy information spreading time for every $h$, which match the following bound proved in \cite{boczkowski2018limits}, up to a $\mathcal O(\log n)$ factor.
Let us first state the lower bound theorem in~\cite{boczkowski2018limits}. That paper assumed that all sources have the same (and therefore correct) preference, which means, in particular, that the bias $s$ is the same as the number of sources.
%however, their lower bound holds also in the case of conflicting sources, since this situation is only more difficult.
The main focus therein was the case of pairwise interactions, that is, $h=1$, however their proof, that can be found in~\cite{boczkowski_supplementary_2018}, applies for general $h$. Note that in the case that~$\delta$ is bounded away from~$1/|\Sigma|$, the lower bound becomes informative only when~$s \leq \sqrt{n}$.

%which holds even when agents are fully synchronized, so that they know when the execution starts, and hence can share a global notion of time.
\begin{theorem}[Theorem 4 in~\cite{boczkowski_supplementary_2018}.]
    \label{thm:boczkowski2018limits}
    %Assume that either~$s_0 = 0$ or~$s_1 = 0$, or in other words, no source has the wrong opinion
    %Assume that no source agent has the wrong opinion, and that the number of sources is less than~$\sqrt{n}$.
    Fix a non-source agent~$u$ and an integer $h$. Any rumor spreading protocol in the noisy $\pull(h)$ model with alphabet $\Sigma$ and $\delta$-lower bounded noise requires $\Omega\! \left(\frac{n\delta}{s^2 h(1 - |\Sigma| \, \delta)^2}\right)$ rounds
    in order to guarantee that the opinion of $u$ is correct with probability at least $2/3$.
\end{theorem}

This lower bound holds even assuming that all parameters, including the number of agents $n$, the noise matrix~$N$, and the number of sources $s$, are known to the protocol's designer.

In particular, this theorem implies that for a single source (i.e., $s=1)$, constant sample size $h$, and constant noise $\delta>0$, the rumor spreading time is $\Omega(n)$, even if correctness is guaranteed with only a constant probability per agent.
\cite{clementiConsensusVsBroadcast2020}
showed that for constant parameters as above, if convergence is required w.h.p then the techniques in \cite{boczkowski2018limits} can be modified to yield an $\Omega(n\log n)$ lower bound. The authors in \cite{clementiConsensusVsBroadcast2020} also
present a protocol that is unlikely to be found in nature but achieves rumor spreading in $O(n\log n)$ rounds w.h.p., matching the lower bound for the case $h=1$.

We note that, if a protocol guarantees correct convergence after $\mathcal{T}/h$ rounds with a sample size $h$, then, given enough memory, agents could simulate  it with a smaller sample size $h'$, say $h'=1$, allowing them to adopt the correct opinion after $\mathcal{T}$ rounds (essentially by  collecting messages for $h$ rounds before changing their state). Therefore, a lower bound of $\mathcal{T}$ in the $\pull(1)$ model, can be extended to a lower bound $\mathcal{T}/h$ in the $\pull(h)$ model. However, it is unclear whether an upper bound in the $\pull(1)$ model would directly translate to a speed-up of~$h$ in the $\pull(h)$ model.
%$1/h$ \rnote{Shouldn't it be just~$h$?}

\subsection{Our results}
%The correctness of the protocols in \Cref{thm:upper,thm:upper_recursive} depends on the assumption that agents begin execution simultaneously, allowing them to transition through the phases in sync. At the cost of increasing the minimal alphabet size to 4, and giving up for the acceleration caused by having a large bias $s$, we also present an efficient protocol that removes the SMmultaneous wake-up assumption to become self-stabilizing.

Before stating our results, we note that our modeling choices (see Section \ref{sec:problem}) are directly motivated by the lower bound in~\cite{boczkowski_supplementary_2018}. As mentioned, that paper established a linear-time barrier of $\Omega(n)$ for rumor propagation under random pairwise interactions, concluding that fast information spreading is impossible when communication is noisy, its pattern is completely random and limited to pairwise exchanges. However, it remained unclear which assumptions were fundamentally responsible for this limitation.
Our results clarify this point by isolating the role of sampling: we show that increasing the number of samples per round (while maintaining the other constraints of the model) is sufficient to overcome the linear barrier.
\paragraph*{Quasi self-stabilizing protocol} Our main result is a highly efficient quasi self-stabilizing rumor spreading algorithm called SSF, that uses 2-bits per message.
\begin{theorem}\label{thm:self}
    Let $\varepsilon>0$ be an arbitrary constant. Consider $n$ agents with the number of sources $s_0+s_1 \leq (1-\varepsilon)n$ and with bias $s := |s_1 - s_0|\geq 1$ interacting according to the noisy $\pull(h)$ model with~$\delta$-upper bounded noise, over any alphabet $\Sigma$ s.t.~$4\leq |\Sigma|=\mathcal{O}(1)$.
    The noisy majority rumor spreading problem can be solved in a quasi self-stabilizing manner, w.h.p in number of rounds:
    \begin{equation*}
        \mathcal{T}  := \mathcal{O} \left(\frac {\delta n\log n}{ h(1-4\delta)^2}+\frac{n}{h}\right),
    \end{equation*}
    while using~$\mathcal O\! \pa{\log \mathcal{T} + \log h}$ bits of memory per agent.
    %Consider $n$ agents with bias $s := |s_1 - s_0|\geq 1$ interacting according to the uniform $\pull(h)$ model with $\delta$-upper bounded noise, and let $m := \frac {\delta n\log n}{ (1-4\delta)^2}+n$. Then, the noisy rumor spreading problem can be solved w.h.p in $O\left(\frac {m}{h}\right)$ number of rounds, in a self-stabilizing way, while using $O(\log m)$ bits of memory.
\end{theorem}

%Recall from our definition of self-stabilization that agents know an approximation of the number of agents $n$ and the noise matrix~$N$. Moreover, each agent always knows whether it is a source, as determined by the initial configuration. In contrast, agents do not need to know the bias $s$.

Note that assuming the noise is~$\delta$-uniform and constant and that the bias $s$ is constant, the theorem matches the lower bound in  \Cref{thm:boczkowski2018limits} up to a logarithmic factor.
%   \item
This extra log factor is likely due to the fact that \Cref{thm:boczkowski2018limits} only requires agents to be correct with constant probability, while our notion of convergence demands high probability.
Indeed, for the case where $s=1$ and $\delta$ is constant, the authors in~\cite{clementiConsensusVsBroadcast2020} showed that the extra $\log n$ factor is necessary to achieve convergence w.h.p.~when~$h=1$ (Theorem 7 in~\cite{clementiConsensusVsBroadcast2020}), but in fact, the result can be extended to general $h$, by the same argument presented immediately after \Cref{thm:boczkowski2018limits}.
%In contrast to  \Cref{thm:upper} (the non-self-stabilizing setting), \Cref{thm:self}, as \Cref{thm:upper_recursive}, does not require agents to know the bias $s$.

\Cref{thm:self} demonstrates how a larger sample size can linearly accelerate the spreading time. %As already mentioned in \Cref{subsec:background}, this question remained open in \Cref{thm:boczkowski2018limits}.
In particular, under constant noise, in the extreme snapshot case where each agent observes all other agents in each communication round ($h=n$), majority rumor spreading can be reliably achieved w.h.p in $\mathcal{O}(\log n)$ time.

Note that  correctness is guaranteed in \Cref{thm:self} even when $s=1$. This contrasts  multiple works in population protocols that guarantee convergence to the majority opinion only when the bias towards it is $\Omega(\sqrt{n\log n})$, see e.g., \cite{DBLP:journals/dc/AngluinAE08,condon2017simplifying,DBLP:conf/podc/AmirABBHKL23}.

Our protocol is presented assuming an alphabet of size 4 (2-bit messages). In
%Section 3 of SM
\Cref{sec:reduction alphabet}
we show that for any larger alphabet, it is possible to adapt systematically any protocol (even self-stabilizing) to work with a smaller alphabet, without changing the asymptotic convergence time.
\paragraph*{A simple protocol assuming simultaneous wake-up}

Our quasi self-stabilizing protocol SSF builds on a simpler protocol, SF, which assumes simultaneous wake-up, uses only 1-bit messages, and in some settings converges slightly faster. Owing to its simplicity, SF may be of independent interest.

%Similarly to \Cref{thm:boczkowski2018limits}, the theorem assumes that the designer of the protocol knows the number of agents $n$, the number of sources $s_0+s_1$, the bias $s=|s_0-s_1|$, and the noise matrix $N$.
To present its convergence time, let us first define the parameter  $m(n,s_0,s_1)$ by the following equation, omitting its arguments in what follows. Let $\rho=\min\{s^2,n\}$, then:
\begin{equation}
    \label{eq:m for SF}
    m = \Theta\pa{\pa{\frac{n \delta}{\rho(1-|\Sigma|\delta)^2} + \frac{\sqrt{n}}{s} + \frac{(s_0+s_1)}{s^2} + h}\log n}.
\end{equation}

\begin{theorem}
    \label{thm:upper}
    Consider $n$ agents interacting according to the noisy $\pull(h)$ model with $\delta$-upper bounded noise over an alphabet $\Sigma$ s.t. $2\leq |\Sigma|= \mathcal{O}(1)$. Assume that the initial configuration has bias $s := |s_1 - s_0|\geq 1$, and that
    the designer of the protocol knows the noise matrix $N$ and has access to constant approximations of $n$, $s_0+s_1$, $s$.
    %(to simplify notation, with denote such approximations with the original symbols).
    The noisy majority rumor spreading problem can be solved w.h.p in
    $ \mathcal{T}  := \mathcal{O}\! \pa{ \frac{m}{h}}$
    rounds, while using~$\mathcal O\! \pa{\log \mathcal{T} + \log h}$ bits of memory per agent.
\end{theorem}
%\begin{remark}~
%\begin{itemize}
% \item

%In practical contexts, this shift from linear to logarithmic time can be the difference between impracticality and feasibility.
%This suggests that an increased sample size can effectively compensate for the lack of structure in noisy environments (see further discussion in \Cref{sec:conclusions}).

%\item
When~$\delta > s/({4\sqrt{n}})$, and the number of sources is bounded as~$s_0,s_1 \leq \sqrt{n}$, the upper bound in \Cref{thm:upper} becomes
%\begin{equation*}
$\mathcal{O}\! (  \frac{n \delta}{s^2 h (1-2\delta)^2}
    \cdot \log n + \log n)$.
%\end{equation*}
Together with assuming the noise to be~$\delta$-uniform, the theorem matches the lower bound in  \Cref{thm:boczkowski2018limits} up to a logarithmic factor, which, as mentioned, can be attributed to the high probability guarantee.
%   \item
%    This extra log factor is likely due to the fact that \Cref{thm:boczkowski2018limits} only requires agents to be correct with constant probability, while our notion of convergence demands high probability.   Indeed, for the case where $s=1$ and $\delta$ is constant, the authors in~\cite{clementiConsensusVsBroadcast2020} showed that the extra $\log n$ factor is necessary to achieve convergence w.h.p.~when~$h=1$ (Theorem 7 in~\cite{clementiConsensusVsBroadcast2020}), but in fact, the result can be extended to general $h$, by the same argument presented immediately after \Cref{thm:boczkowski2018limits}.

%\item
Observe that our general definition of convergence time does not require the agents to stop the protocol after reaching consensus. However,  \Cref{thm:upper} assumes that agents can calculate $\mathcal{T}$, and hence they can terminate after this time, while maintaining a high probability correctness guarantee.

In \Cref{thm:upper}, the protocol assumes knowledge of the number of agents, sources, bias, and noise matrix. While restrictive, these assumptions can be relaxed,  %agents know only the noise matrix and a constant estimate of $\log n$—
at the cost of an additional $\mathcal{O}\! (\log n)$ factor in the runtime, see %SM, Section 4.
\Cref{sec:USF}.

To test the effectiveness of our protocols against their theoretical asymptotic behavior, we ran several simulations. The results show that the one behind the asymptotic bound of \Cref{thm:self} is 1.89 and the constant behind \Cref{thm:upper} is 3.06. See \Cref{sec:simulations}.%Section 6 of the SM for more details.

%\begin{theorem}
%    \label{thm:upper_recursive}
%    Consider $n$ agents with bias $s := |s_1 - s_0|\geq 1$ interacting according to the noisy $\pull(h)$ model with~$\delta$-upper bounded noise, over any alphabet $\Sigma$ s.t.~$2\leq |\Sigma|=\mathcal{O}(1)$.
%   Assume that the protocol designer knows the noise matrix $N$ and a constant approximation of $\log n$.
%  The noisy rumor spreading problem  can be solved w.h.p in $\mathcal{T}:= \mathcal{O}\! \pa{ \frac{m}{h} \log n} $
% rounds,  using~$\mathcal O\! \pa{\log \mathcal{T} + \log h}$ bits of memory per agent.
%\end{theorem}
%Since the knowledge assumptions about the systems' parameters in \Cref{thm:upper_recursive} are weak, it is not clear how agents could terminate the protocol after $\mathcal{T}$ rounds, as they do in \Cref{thm:upper}. Instead, using their knowledge of $\log n$, they could terminate the protocol after a polynomial time in $n$.

%The correctness of the protocols in \Cref{thm:upper,thm:upper_recursive} depends on the assumption that agents begin execution simultaneously, allowing them to transition through the phases in sync.
%At the cost of increasing the minimal alphabet size to 4, and giving up for the acceleration caused by having a large bias $s$, we also present an efficient protocol that removes the simultaneous wake-up assumption to become self-stabilizing.

%Overall, our results show that a large sample size can offset the lack of structure in noisy environments and still enable reliable, efficient dissemination under extremely restricted communication bandwidth.

We present our protocols and the intuition behind them in \Cref{sec:BBS}. Both are simple, and mainly rely on basic operations, namely, averaging, counting, and taking majority.  %All proofs and technical details that are omitted here due to space constraints can be found in the Supplementary Material.

\subsection{More related works}
\label{sec:relatedworks}
The problem of propagating information from a few sources to the entire population has been extensively studied under various names such as rumor spreading,
%information dissemination,
bit-dissemination, epidemics, gossip, and broadcast, see, e.g.,~\cite{demersEpidemicAlgorithmsReplicated1987,karp2000randomized,chierichetti2018rumor,kempe2003gossip}.
%\rnote{Also here I would mention ``bit-dissemination'' (if we agree to add it in the first place).}
The problem of converging to the most frequent opinion among sources is also known as ``zealot consensus", ``majority bit dissemination" and ``majority consensus"~\cite{mobiliaRoleZealotryVoter2007,Mossel_AAMAS14,boczkowskiMinimizingMessageSize2017,damorePhaseTransitionNonlinear2022,berenbrinkUndecidedStateDynamics2024}.
%These tasks become particularly challenging when communication is limited and the system is vulnerable to faults.
%Note that while majority-like consensus protocols~\cite{Mossel_AAMAS14,damorePhaseTransitionNonlinear2022} are often robust to noise, they fail in this setting because they disregard the opinions of the sources.

%\nnote{This problem has been studied in the areas of {\em population protocols}  \cite{angluin2006computation,aspnes2009introduction,DBLP:conf/stoc/BerenbrinkGK20} and {\em opinion dynamics} \cite{becchetti2020consensus,DBLP:conf/podc/BerenbrinkCGMMR23}, in which, more in general, it is often studied the consensus problem.   }

The terms  $\pull(h)$ and  $\push(h)$ denote random meeting patterns, in which in each round, each agent $u$ samples $h$ agents uniformly at random, and may either extract information from them ($\pull$), or inform them with a message ($\push$)~\cite{demersEpidemicAlgorithmsReplicated1987,karp2000randomized}.
These models capture ``well-mixed'' scenarios, in which agents have little  control over who they interact with, and are reminiscent of random meeting patterns studied in the area of {\em population protocols}  \cite{angluin2006computation,aspnes2009introduction,DBLP:conf/stoc/BerenbrinkGK20} and {\em opinion dynamics} \cite{becchetti2020consensus,DBLP:conf/podc/BerenbrinkCGMMR23,Cruciani_AAMAS18,Cruciani_AAAI19,Gauy_AAMAS25}.

\section{Protocols and Analytical Intuition}
\label{sec:BBS}

Here, we present two of our protocols and the intuition behind their construction. A last protocol is deferred to
%Section 4 in SM,
\Cref{sec:USF}
as it does not introduce novel ideas.

In
%Section 1 of SM,
\Cref{sec:non-uniform}
we show how to reduce the problem to the case that the noise is~$\delta$-uniform using results of \cite{draft_symm_noise}, in order to facilitate several aspects of our algorithms and their analysis.
For this purpose,
upon receiving a message~$\sigma \in \Sigma$, we let each agent apply an ``artificial'' noise to~$\sigma$ to obtain a new random message~$\sigma'$, and then run the original protocol on~$\sigma'$.
The combination of the noise matrix $N$ and the artificial noise ends up being a noise matrix $T$ which is $\delta$-uniform for some suitable~$\delta > 0$.
A formal definition of this notion of reduction is given in
%Definition 1 in SM
\Cref{def:simulation}
, while the fact that it can be achieved in our scenario is proved in
%Theorem 3 of SM%
\Cref{thm:reduction}
.
When~$|\Sigma|=2$, this reduction can be easily achieved; however, for a larger alphabet, its feasibility is less obvious.
Specifically, showing that uniform noise can be achieved for a general~$\Sigma$ under a noise matrix~$N$ involves showing that~$N$ can be inverted, computing a suitable~$\delta > 0$, and proving that the product of~$N^{-1}$ with a $\delta$-uniform matrix is stochastic, so that the strategy can be implemented.
%We establish the non-singularity of~$N$ in \Cref{sec:linear_algebra} (\Cref{cor:N_is_invertible}), and the other statements in \Cref{prop:non-uniform_case}.

%In Section 3 of SM
%\Cref{sec:reduction alphabet} we show how our protocols which are designed for a specific alphabet $\Sigma'$ of cardinality either 2 or 4, can be adapted to work with arbitrary alphabets $\Sigma$, without changing the asymptotic of the convergence time of that protocol, under the only assumption that $|\Sigma'|\leq |\Sigma|$.

In
%Section 2 of SM
\Cref{sec:upper_bound}
we consider the setting with uniform noise, presenting two related protocols, Source Filter (SF) and Self-stabilizing Source Filter (SSF).
The former protocol requires agents to start the execution simultaneously at the same time, while the latter is self-stabilizing, at the cost of reduced efficiency and an increased
minimal communication alphabet.
Since both protocols rely on similar ideas, we present their analysis in a unified manner.

\subsection{Source Filter (SF)} \label{sec:SF}

\newcommand{\counter}{\texttt{Counter}}

Here, we describe our fastest protocol, whose pseudo-code is given in
%Section 5 of SM.
\Cref{alg:SF}
The protocol uses a communication alphabet that corresponds to the opinion's set, namely,~$\Sigma := \{ 0,1 \}$. We refer to
%Section 3 of SM
\Cref{sec:reduction alphabet}
to show how the protocol can be adapted to work with a larger alphabet.
The execution is divided into three phases (where the first two phases are symmetric).
The protocol relies on a simultaneous wake-up assumption, which is used to synchronize the clocks of agents.
Each of the first two phases lasts for~$\lceil m/h \rceil$ rounds.
Since every agent observes~$h$ samples per round, it collects at least~$m$ messages in each of the first two phases.
%These phases, detailed below, are symmetric with respect to the opinions $0$ and $1$.
%
\begin{itemize}
    \item {\bf Phase 0.} All non-sources display 0, while sources display their preference. In addition, all agents count the number of 1-opinions that they observe during the phase.
          %keep a counter $V_t^{(i)}$ to

    \item {\bf Phase 1.} All non-sources display 1, while sources display their preference. All agents count the number of 0-opinions they observe during the phase.
\end{itemize}
Keeping Phases~0 and~1 separate is precisely what allows \SF to use a 1-bit alphabet. Across the two phases, each non-source contributes once to the 0-background and once to the 1-background, so these contributions cancel in the comparison, whereas every source keeps displaying its preference. A single phase would not provide this cancellation and would require another mechanism to distinguish source information, such as the extra source-flag bit used by \SSF.
At the end of  Phase 1, each agent~$i$ compares the number of~$1$-messages observed in Phase 0, with the number of~$0$-messages observed in Phase~$1$, and then sets its {\em weak-opinion}~$\Y^{(i)}$ to be the opinion with the greater number (breaking ties randomly), i.e., $
    \Y^{(i)} := \mathds{1}\{ \counter_1 > \counter_0 \} + \mathds{1}\{ \counter_1 = \counter_0 \} \cdot \text{Bernoulli}(\frac 12)$.
Note that the $0$-messages observed in Phase 0 and the $1$-messages observed in Phase 1 are ignored.

The third phase is devoted to boosting the bias.
\begin{itemize}
    \item {\bf Majority Boosting phase.} In every round~$t$ during this phase, each agent displays its opinion~$Y_t^{(i)}$, which is initially equal to its weak-opinion (computed as above), and then updated from time to time by sampling the opinions of several other agents and taking the majority.
          More specifically, the Majority Boosting phase is divided into $10 \, \log n + 1$ sub-phases. At the end of each sub-phase, every agent~$i$ sets its opinion to be the majority over all messages it received during this sub-phase. All sub-phases except the last one have the same duration, which is chosen in order to allow each agent to gather at least~$w = \frac{100}{(1-2\delta)^2}$ messages per sub-phase. Finally, the last sub-phase lasts long enough for agents to gather at least~$m$ messages.
\end{itemize}
By design, the Majority Boosting phase takes  $\lceil w/h \rceil \cdot 10 \, \log n + \lceil m / h \rceil$ rounds, which
is at most
the duration of the first two phases
(\Cref{claim:duration_boosting}).
Therefore, the running time of  \SF is~$\mathcal{O}(m/h)$ rounds.

%
%Notably, it is interesting that a small constant number of phases is sufficient for near-optimal convergence, without requiring non-sources to express varying levels of confidence.

By the symmetric construction of the algorithm and since the noise is uniform, we expect the number of messages wrongly counted as 1s in Phase~$0$, to be the same as the number of messages wrongly counted as 0s in Phase~$1$.
Since there are more sources with the correct preference, we also expect that they observe more correct messages originating from sources.
Together, these two observations imply that $\Y^{(i)}$ is correct with a probability larger than~$1/2$.
By setting~$m$ to be large enough, we can get this probability to be at least~$1/2+\sqrt{\log(n) / n}$ (see more details in \Cref{sec:intuition}), which is sufficiently high to allow the Majority Boosting phase to be efficient and reliable.

In addition, we show that the value of~$\Y^{(i)}$ depends only on the random samples received by Agent~$i$, the noise affecting these samples, and the coin tosses used by the agent to break ties. Since all these random variables are mutually independent among agents, so are the weak-opinions.

%\begin{remark}
%    In the first two phases of Algorithm \SF, each non-source agent presents the message 0 consecutively for many rounds and then switches to presenting 1 for the same number of rounds. Perhaps a more natural algorithm would allow each agent to first flip a fair coin to determine the message it will present on the first round, and then, over the following rounds, deterministically alternate between 0 and 1. While it is plausible that such a scheme would work as well, it does add some complexity to the analysis. For the sake of simplicity, we therefore decided to focus on Algorithm \SF.
%\end{remark}

\subsection{Self-stabilizing Source Filter (SSF)} \label{sec:SSF}

%Algorithm \SF has an optimal convergence time for a wide range of values of~$\delta$, and moreover, it operates with any communication alphabet of size at least 2. However, it requires agents to share a global notion of time, so that they can agree on when to start each phase.
%In Section 2.4 of SM
%\Cref{sec:SSF_analysis}
%, we show that, at the cost of increasing the minimal size of the alphabet and slightly increasing the running time, such a synchronization assumption is not essential.

We next present the quasi self-stabilizing protocol SSF, whose pseudo-code is given in
%Section 5 of SM
\Cref{alg:SSF}.
The protocol uses a communication alphabet of size 4, namely,~$\Sigma := \{ 0,1 \}^2$, meaning that each message is simply composed of 2 bits. We refer to
%Section 3 of SM
\Cref{sec:reduction alphabet}
to show how the protocol can be adapted to work with a larger alphabet.
In every round~$t$, each agent~$i$ maintains a weak-opinion~$\Y_t^{(i)} \in \{0,1\}$ as well as an opinion $Y_t^{(i)} \in \{0,1\}$ that is held internally. The first bit of each message indicates whether or not the sender is a source. For sources, the second bit corresponds to their preference, whereas for non-sources, it indicates their weak-opinion.
In addition, each agent~$i$ holds a multi-set in her memory, whose content in round~$t$ is denoted by~$M_t^{(i)}$. We design $M_t^{(i)}$ such that it can store $m'$ messages received by the agent, with
\[
    m' = \frac{\delta n \log n}{(1-4\delta)^2} + n.
\]
In every round, Agent~$i$ adds all messages received in this round to~$M_t^{(i)}$. If this results in~$|M_t^{(i)}|$ exceeding~$m'$ (which happens every~$\lceil m'/h \rceil$ rounds), the agent updates its opinion and weak-opinion based on the content of the multi-set, and then empties it. Specifically,
on such a round, which we refer to as an {\em update round} for Agent~$i$, the following occurs:
\begin{itemize}
    \item the new weak-opinion of Agent~$i$, namely, $\Y_{t+1}^{(i)}$, is computed as the majority over the second bits of all messages in~$M_{t}^{(i)}$ whose first bit is 1 (breaking ties uniformly at random).
          Recall that in absence of noise, such a bit would correspond to the preference of a source agent.

    \item the new opinion of Agent~$i$, namely $Y_{t+1}^{(i)}$, is computed as the majority over the second bits of all messages in~$M_{t}^{(i)}$ (breaking ties uniformly at random).
\end{itemize}
%When evaluating majorities, agents break ties uniformly at random.
Note that in a self-stabilizing framework, the multi-sets~$\{ M_t^{(i)} \}_{i \in I}$ are initialized arbitrarily. In particular, they may start with different sizes, in which case the update rounds of different agents may not occur simultaneously. This behavior is illustrated in \Cref{fig:selfstab_plots}, where we compare executions of \SSF initialized with homogeneous versus heterogeneous sizes (over $i \in I$) of the multisets~$\{ M_t^{(i)} \}_{i \in I}$.

Importantly, the larger number of sources with the correct preference (say, $1$), together with the uniform noise, imply that each agent is more likely to receive the message~$(1,1)$ than~$(1,0)$, even after noise is applied. In turn, this implies that each weak-opinion~$\Y_t^{(i)}$ is correct with a probability larger than~$1/2$.
This is formally proved in
\Cref{lemma:correctness_weak}.

We remark that the fact that for every round~$t$, all weak-opinions~$\{\Y_t^{(i)}\}_{i \in I}$ are mutually independent is less obvious than for Algorithm \SF.
As mentioned, the weak opinion of an agent is updated according to the second bits of messages that were received with a first bit that equals 1. If such a message~$\sigma$ originated from a non-source agent, then it must have been corrupted by noise. Therefore, because of the noise being uniform, the distribution of the second bit of~$\sigma$ is independent of the second bit of the original message (which is the weak opinion of the sampled non-source agent).
This implies that, eventually, the weak-opinion of each agent~$i$ is independent from the behavior of all other non-sources.
Finally, notice that, after~$2\lceil m/h \rceil$ rounds, each agent has updated its opinion and weak-opinion twice. For every agent~$i$, the set~$M_t^{(i)}$ is emptied after its first update, and from now on, it will only contain messages that were actually sampled (and not included in the arbitrary initial configuration), which allows the protocol to be self-stabilizing.

\begin{algorithm}[ht]
    Communication alphabet~$\Sigma := \{ 0,1 \}^2$ \;
    %\alginput{Any set of $m$ messages received in the last $2m/h$ rounds.}
    %
    \alginput{Sample size~$h$, Parameter~$m$}
    %any deterministic choice of $m$ different messages received within the last $2m/h$ rounds \;
    %
    S $\leftarrow$ observe the messages of $h$ random agents \;
    $M_{t+1}^{(i)} \leftarrow M_t^{(i)} + S$ \;

    \If{$|M_{t+1}^{(i)}| \geq m$} {
    $\tilde{Y}_{t+1}^{(i)} \leftarrow \underset{x \in \{0,1\}}{\argmax} ~ \{$ Number of occurrences of $(1,x)$ in $ M_{t+1}^{(i)} \}$ \quad (Breaking ties randomly) \; %\tcp*{Breaking ties randomly}
    $Y_{t+1}^{(i)} \leftarrow \underset{x \in \{0,1\}}{\argmax} ~ \{$ Number of occurrences of $(\ast, x)$ in $ M_{t+1}^{(i)} \}$ \quad (Breaking ties randomly) \; %\tcp*{Breaking ties randomly}
    $M_{t+1}^{(i)} \leftarrow \emptyset$ \;
    } \Else{
        $\tilde{Y}_{t+1}^{(i)} \leftarrow \tilde{Y}_{t}^{(i)}$ ~and~
        $Y_{t+1}^{(i)} \leftarrow Y_{t}^{(i)}$ \;
    }
    \BlankLine
    \lIf{Agent~$i$ is a source}{
    Display~$(1,\text{preference})$ {\bf else} Display~$\pa{0,\tilde{Y}_{t+1}^{(i)}}$
    }

    %   \hrulefill

    % {\em Variables are indexed by~$i$ and $t$ for analysis purposes, but agents do not need to know them.}
    %
    % \BlankLine
    \caption{Self-stabilizing Source Filter (\SSF) for agent~$i$ in round~$t$. (Agents need not know $i$ or $t$.)}
    \label{alg:SSF}
\end{algorithm}

%\begin{remark}
%    One way to make a phase-based protocol self-stabilizing is to implement a subroutine that synchronizes the agents' clocks, see e.g., \cite{boczkowskiMinimizingMessageSize2017}.
%This technique has been used, for example, in \cite{boczkowskiMinimizingMessageSize2017}, which tackles the same rumor spreading problem in the absence of noise.
%However, we are not aware of any known self stabilizing clock synchronization protocol in the presence of noise, and such designs would likely be too complex for the biological motivation.
%\end{remark}

\subsection{Weak Opinions: Correctness Intuition} \label{sec:intuition}

Here, we give an overview of the computations ensuring that the weak-opinions are correct, in both Algorithms \SF and \SSF. Formal statements are collected in \Cref{sec:upper_bound}, and the corresponding proofs are deferred to \Cref{sec:upper_bound_proofs}.

%- Explain how to define the $X_k$
Fix an agent~$i$.
For both protocols, for the sake of analysis, we map the messages used to compute the weak-opinions of Agent~$i$, to a set of i.i.d.~random variables~$\{X_k\}_{k=1,\ldots,m}$ taking values in~$\{-1,0,+1\}$.
%Although these variables are defined differently for each protocol,
The general idea is that $X_k = +1$ (resp. $-1$) if the corresponding messages suggest that the correct preference is~$1$ (resp. $0$), whereas~$X_k = 0$ if the corresponding messages do not carry relevant information.
For instance, in  Algorithm \SSF, we define one~$X_k$ for each message in the memory of the agent. When a message has a first bit equal to~$1$ (i.e., it is tagged as coming from a source), the corresponding $X_k$ is assigned value~$+1$ or~$-1$ depending on the second bit of the message, and otherwise, it is assigned value~$0$.
Instead, in the case of Algorithm \SF, each~$X_k$ corresponds to a pair of messages, one received in Phase 0 and one in Phase 1.
When these messages are either both equal to~$1$, or both equal to~$0$, then the corresponding~$X_k$ is assigned value~$+1$ or~$-1$ respectively. Conversely, when these messages are different, $X_k$ is assigned value~$0$.
Next,  consider the sum $X = \sum_{k} X_k$.
The protocols and the $X_k$s are defined so that the weak-opinion of Agent~$i$ is ~$1$ iff~$X > 0$ (leaving the tie-breaking mechanism aside).
W.l.o.g, we will assume that the correct preference is~$1$, and focus on showing that~$X>0$ with probability at least~$1/2 + \Omega(\sqrt{\log n / n})$.
This requires a fine-grained analysis of the distribution of~$X$, that classical concentration inequalities do not provide directly.

%- Explain that we reduce to Rademacher with these parameters: ...
Recall that a Rademacher random variable takes values in \\$\{-1,+1\}$. Although the~$\{X_k\}$ variables are not Rademacher, since the number of non-zero~$X_k$s is~$\ell \approx m \cdot \Pr\pa{X_k \neq 0}$, we show that~$X$ is distributed roughly as a sum of~$\ell$ i.i.d~Rademacher random variables.
    Moreover, each variable in this sum equals~$1$ w.p.~$p = \Pr\pa{X_k = 1 \mid X_k \neq 0}$.
%- State Lemma 22
We can then apply the following lower bound, derived directly from existing literature and proved in \Cref{lemma:majority_boosting}:
\begin{equation*}
    \Pr\pa{ X>0 } - \Pr\pa{ X<0 } \geq \sqrt{\frac{2}{\pi e}} \cdot \min\{(p-\tfrac{1}{2}) \sqrt{\ell},1\}.
\end{equation*}
%- Thus, given that we need the advantage to be at least $sqrt{log n/n}$, $m$ must be ...
Since we are excluding possible ties (for now), it holds that \\$\Pr\pa{X<0}=1-\Pr\pa{X>0}$, hence a sufficient condition for having~$X>0$ with probability at least~$1/2 + \Omega(\sqrt{\log n / n})$, is (up to a constant factor):
    \begin{equation}
        \label{eq:sufficient_cond_for_bias}
        \pa{p-\frac{1}{2}} \cdot \sqrt{\ell} \geq \sqrt{\frac{\log n}{n}}.
    \end{equation}
    %- What is Pr(X_k != 0), from SF and SSF?
    %- What is p exactly, from SF and SSF?
    %-> They are computed in the analysis of each protocol separately, but both satisfy...
    %-> both depends on the noise regime: pure gold vs garbage
    This condition reflects the need to gather sufficient information from the variables~$X_k$. Specifically, $\ell$ represents the number of~$X_k$s that provide any useful information, while~$p-\frac{1}{2}$ quantifies the extent to which each non-zero~$X_k$ contributes in correctly inferring the correct preference.
    In order to estimate~$\ell = m \cdot \Pr\pa{X_k \neq 0}$, we show that for both protocols,
    \begin{equation*}
        \Pr\pa{X_k \neq 0} \geq \delta + \frac{s_0 + s_1}{2n}(1-|\Sigma|\delta).
    \end{equation*}
    The first term corresponds to corrupted messages, while the second corresponds to uncorrupted messages from the sources (in \SSF, there is another factor of $(1-|\Sigma|\delta)$ that is due to the more complex pairing of $X_k$-messages; but we ignore it here for the sake of clarity).

    We distinguish two regimes in computing~$p - \tfrac{1}{2}$:
    \begin{itemize}
        \item If
              %\begin{equation*}
              $\delta < \frac{s_0 + s_1}{2n} (1-|\Sigma|\delta)$,
              %\end{equation*}
              then each non-zero~$X_k$ is more likely to result from a direct observation of a source without distortion, than to be a consequence of noise.
              Hence, each non-zero~$X_k$ is relatively informative,
              %, since the probability that it corresponds to an observation of a source is large.
              and we can show that both protocols
              %\begin{equation*}
              $p -\frac{1}{2} \geq \frac{s}{4 (s_0+s_1)}$.
              %\end{equation*}

        \item Conversely, when
              %\begin{equation*}
              $\delta \geq \frac{s_0 + s_1}{2n} (1-|\Sigma|\delta)$,
              %\end{equation*}
              each non-zero~$X_k$ is quite likely to be the consequence of a corrupted message, hence~$p-\frac{1}{2}$ cannot be too large; in fact, we show that for both protocols, we only have
              %\begin{equation*}
              $p-\frac{1}{2} \geq \frac{s}{n}\frac{(1-|\Sigma|\delta)}{8\,\delta}$.
              %\end{equation*}
              However, this drop in quality is offset by the greater number of non-zero~$X_k$s, enabling agents to accumulate enough information.
    \end{itemize}
    The above computations of~$p$ and~$\Pr \pa{X_k \neq 0}$ are presented in
    \Cref{lemma:first_phase_SF}
    for \SF, and in
    \Cref{lemma:correctness_weak}
    for \SSF.
    From these quantities, we derive a sufficient condition on~$m$ that applies to both protocols, in
    \Cref{lemma:first_phase}.
    Specifically, with the above bounds, in the regime where $\delta \geq \frac{s_0 + s_1}{2n} (1-|\Sigma|\delta)$, it suffices to choose $m$ such that
    %\begin{align*}
$m\geq \frac{\delta n \log n}{s^2 (1-|\Sigma|\delta)^2}$,
    %\end{align*}
    to satisfy \Cref{eq:sufficient_cond_for_bias}.
    However, to ensure the correctness of the weak opinion across a broader range of parameters, additional additive terms for $m$ are necessary.

\section{Analysis} \label{sec:upper_bound}

We now present the analysis of our two protocols, \SF and \SSF. We defer all proofs to \Cref{sec:upper_bound_proofs} due to space constraints. Throughout this section, the noise matrix~$N$ is assumed to be~$\delta$-uniform, and, by symmetry of the protocols, the correct preference is assumed to be~$1$, i.e., $s_1>s_0$.

\subsection{Notation} \label{sec:prelim}

Recall that~$Y_t^{(i)}$ and~$\Y_t^{(i)}$ denote respectively the opinion and the weak opinion of Agent~$i$ in round~$t$. We use the following random variables in the analysis.
\begin{itemize}
	\item For $\ell \in \{1,\ldots,h\}$, $S_{t,\ell}^{(i)}$ is the index of the $\ell^{\text{th}}$ agent sampled by Agent~$i$ in round~$t$. Thus $S_{t,\ell}^{(i)}$ is uniform on~$I=\{1,\ldots,n\}$.

	\item For $\ell \in \{1,\ldots,h\}$, $N_{t,\ell}^{(i)}$ is the noise affecting the $\ell^{\text{th}}$ message received by Agent~$i$ in round~$t$. Formally, it is a random map from~$\Sigma$ to~$\Sigma$.

	\item $C_t^{(i)} \in \{0,1\}$ is the fair coin used by Agent~$i$ to break ties in majority operations in round~$t$.
\end{itemize}
We write
\[
    S_t^{(i)} := (S_{t,\ell}^{(i)})_{\ell \in \{1,\ldots,h\}},
    \qquad
    N_t^{(i)} := (N_{t,\ell}^{(i)})_{\ell \in \{1,\ldots,h\}}.
\]
The variables $\{S_t^{(i)},N_t^{(i)},C_t^{(i)}\}_{i\in I,t\in\bbN}$ are mutually independent. For a time interval, we also write
\[
    \R_{t_1:t_2}^{(i)} := \left\{ (N_t^{(i)},S_t^{(i)},C_t^{(i)}), \quad t \in \{ t_1+1,\ldots,t_2 \} \right\}.
\]

\subsubsection{Correctness of Weak-Opinions}

The following master lemma is the common probabilistic black box used to prove that the weak opinions produced by \SF and \SSF have a small but sufficient bias toward the correct preference.

\begin{lemma}[Technical master lemma]
\label{lemma:first_phase}
    Let $m\in\bbN$ and $\delta<1/|\Sigma|$, and consider i.i.d. random variables $\{X_k\}_{k=1,\ldots,m}$ taking values in $\{-1,0,1\}$. Let $X=\sum_{k=1}^m X_k$, and
    \begin{equation*}
        p := \Pr\pa{ X_k = 1 \mid X_k \neq 0}.
    \end{equation*}
    There exists a constant~$c_1>0$, s.t. if the following conditions hold:
    \begin{equation} \label{eq:zero_th_claim}
        m \geq c_1 \pa { \frac{n \delta \log n}{\min\{s^2,n\}  (1-|\Sigma|\delta)^2} + \frac{\sqrt{n}\log n}{s} + \frac{(s_0+s_1)\log n}{s^2} + \log n },
    \end{equation}
    \begin{equation} \label{eq:first_claim}
         \Pr\pa{ X_k\neq 0 } \geq (1-|\Sigma|\delta)^2 \cdot \frac{s_0 + s_1}{2n} + \delta,
    \end{equation}
    \begin{align}
        \delta \geq \frac{s_0 + s_1}{2n}(1-|\Sigma|\delta) \quad &\implies \quad p \geq \frac{1}{2} + \frac{s}{n}\frac{(1-|\Sigma|\delta)}{8\,\delta}, \label{eq:second_claim} \\
        \delta < \frac{s_0 + s_1}{2n}(1-|\Sigma|\delta) \quad &\implies \quad p \geq \frac{1}{2} + \frac{s}{4 (s_0+s_1)}, \label{eq:third_claim}
    \end{align}
    then for~$n$ large enough,
    \begin{equation*}
        \Pr\pa{ X > 0  } - \Pr\pa{ X < 0  } \geq 8\sqrt{\frac{\log n}{n}}.
    \end{equation*}
\end{lemma}

Let $c_1$ be the constant from \Cref{lemma:first_phase}. In the analysis of \SF, we choose the memory parameter~$m$ such that
\begin{equation} \label{eq:assumption_m_SF}
    m \geq c_1 \pa{\frac{n \delta \log n}{\min\{s^2,n\}  (1-|\Sigma|\delta)^2} + \frac{\sqrt{n}\log n}{s} + \frac{(s_0+s_1)\log n}{s^2} + h \log n}.
\end{equation}
In the analysis of \SSF, we choose~$m'$ such that
\begin{equation} \label{eq:assumption_m_SSF}
    m' \geq c_\varepsilon \cdot \pa{ \frac{\delta n \log n}{(1-4\delta)^2} + n}.
\end{equation}
Here~$c_\varepsilon$ denotes a sufficiently large constant that may depend on the fixed constant~$\varepsilon$ in \Cref{thm:self}. For \SSF, let~$I_{\mathrm{ns}}$ denote the set of non-source agents and let~$n_{\mathrm{ns}}:=|I_{\mathrm{ns}}|$; under the assumption of \Cref{thm:self}, $n_{\mathrm{ns}}\geq \varepsilon n$.

The additional terms in the bounds on~$m$ (compared to the one prescribed by \Cref{lemma:first_phase}) are necessary for the amplification phase.

\begin{lemma}[Correctness of weak opinions in \SF] \label{lemma:first_phase_SF}
    The weak opinion of each agent~$i$, computed after Phase 1 in Algorithm \SF, satisfies
    \begin{equation*}
        \Pr\pa{ \tilde{Y}^{(i)} = 1  } \geq \frac{1}{2} + 4\sqrt{\frac{\log n}{n}}.
    \end{equation*}
    Moreover, all $\{\Y^{(i)}\}_{i \in I}$ are mutually independent.
\end{lemma}

\begin{lemma}[Probability of correctness of weak opinions in \SSF] \label{lemma:correctness_weak}
    For every~$t \geq 2 \lceil m'/h \rceil$ and~$i \in I$,
    \begin{enumerate}
        \item[\bf{(i)}] $\Y_t^{(i)}$ is a deterministic function of~$\R_{t_1:t_2}^{(i)}$, where $t_2 \leq t$ is the last update round of Agent~$i$ up to round~$t$, and~$t_1 = t_2 - \lceil m'/h \rceil$. As a consequence, all $\{\Y_t^{(i)}\}_{i \in I}$ are mutually independent.

        \item[\bf{(ii)}] We have
        \begin{equation*}
        \Pr\pa{ \tilde{Y}_t^{(i)} = 1  } \geq \frac{1}{2} + \frac{8}{\varepsilon} \sqrt{\frac{\log n}{n}}.
        \end{equation*}
    \end{enumerate}
\end{lemma}

\begin{lemma}[Correctness of non-source weak opinions in \SSF] \label{lemma:weak_opinion_strong_correctness}
    For every~$t \geq 2 \lceil m'/h \rceil$,
    \begin{equation*}
        \Pr\pa{ \bigcap_{s = t+1}^{t+\lceil m'/h \rceil} \left\{ \sum_{j \in I_{\mathrm{ns}}} \tilde{Y}_s^{(j)} > \frac{n_{\mathrm{ns}}}{2} + 2 \sqrt{n\log n} \right\} } \geq 1-\frac{1}{n^7}.
    \end{equation*}
\end{lemma}

\subsubsection{Amplification and Convergence}

Once the weak opinions have the above bias and independence, simple majority operations amplify this bias into high-probability correctness.
For the analysis of the Majority Boosting phase of \SF, let~$L:= 10 \, \log n$ be the number of sub-phases before the last one. For~$\ell \in \{0,\ldots,L\}$, define
\begin{equation*}
    \Delta_\ell := \left| \left\{ i \in I, Y_t^{(i)} = 1, \text{where $t$ is the $\ell^\text{th}$ update round} \right\} \right| - \frac{n}{2}.
\end{equation*}
Thus, $\Delta_\ell$ measures the excess number of agents holding the correct preference during the $\ell^\text{th}$ sub-phase.

\begin{lemma} \label{claim:duration_boosting}
    The Majority Boosting phase lasts at most~$2\lceil m/h \rceil$ rounds.
\end{lemma}

\begin{lemma} \label{lemma:initialization_boosting}
    We have that
    \[
        \Pr\pa{\Delta_0 \geq 2\, \sqrt{n \log n}} \geq 1-\frac{1}{n^8}.
    \]
\end{lemma}

\begin{lemma} \label{lemma:boosting_main}
    There exists a constant~$c_2 > 2$ s.t.
    \begin{equation*}
        \Pr \pa{ \Delta_{\ell+1} \geq \min \pa{ 1.2 \, \Delta_\ell, \frac{n}{\sqrt{8 \pi e}} } ~\Big|~ \Delta_\ell \geq 2 \sqrt{n \log n}} \geq 1-\frac{1}{n^8}.
    \end{equation*}
\end{lemma}

\begin{lemma} \label{lemma:correctness_opinion}
    We have that
    \begin{equation*}
        \Pr \pa{\Delta_L \geq \frac{n}{\sqrt{8 \pi e}}} \geq 1 - \frac{1}{n^7}.
    \end{equation*}
\end{lemma}
Before the last sub-phase, a constant fraction of the agents hold the correct preference, and thus, the last sub-phase is sufficient to ensure that a large majority of agents hold the correct preference.

\begin{lemma}[Correctness of \SF after amplification]
    \label{lemma:w.h.p._agent_correct_SF}
    At the end of the Majority Boosting phase, all agents' opinion is the correct preference w.h.p.
\end{lemma}

The corresponding amplification statements for \SSF are the following.

\begin{lemma} \label{lemma:ssf_second_phase}
    Let~$t \geq 4 \lceil m'/h \rceil$. For every~$i \in I$,
    \begin{equation*}
        \Pr \pa{Y_t^{(i)} = 1} \geq 1-\frac{1}{n^6}.
    \end{equation*}
\end{lemma}

\begin{lemma} \label{lemma:end_of_SSF}
    For every~$t_0 \geq 4\lceil m'/h\rceil$, and any polynomial~$P=\mathcal{O}(n^3)$, all opinions are correct between rounds~$t_0$ and~$t_0+P(n)$ w.h.p., i.e., for~$n$ large enough:
    \begin{equation*}
        \Pr \pa{ \bigcap_{t = t_0}^{t_0+P(n)} \bigcap_{i \in I} Y_t^{(i)} = 1 } \geq 1 - \frac{1}{n^2}.
    \end{equation*}
\end{lemma}

    \section{Discussion and Conclusions}
    \label{sec:conclusions}

We investigated the problem of disseminating a binary opinion in noisy, stochastic multi-agent systems, where communication is constrained and initial coordination may be absent. In the noisy $\pull(h)$ model, we show that even under severe noise and minimal message sizes (1–2 bits), fast and robust convergence to the sources’ plurality opinion can be achieved. Our results establish tight upper bounds for convergence time across all values of the sampling parameter $h$, demonstrating that increasing $h$ linearly accelerates dissemination and mitigates the effects of noise.

A useful way to interpret our protocols, in particular Algorithm \SF, is as a two-stage process. In the first stage, agents independently sample the system without relaying information, thereby avoiding interference and forming reliable local estimates. In the second stage, agents exchange these estimates and amplify the emerging majority. This separation between estimation and amplification appears to be a robust design principle for information dissemination under noise.

To model the lack of synchronization in the starting time, we adopt the framework of self-stabilization, which may be stronger than what is required in natural systems. Nevertheless, we show that a simple protocol achieves fast (quasi) self-stabilizing convergence. This comes at the cost of increasing the message size from 1 to 2 bits, highlighting a tradeoff between bit complexity and self-stabilization. It remains open whether a simple and efficient self-stabilizing algorithm using 1-bit messages exists.

Our work also emphasizes simplicity in communication rather than in internal computation. This choice is motivated by systems such as social insects, where individuals may perform sophisticated local computations but face severe communication constraints \cite{feinerman2017individual,dornhaus2008individual}. Accordingly, our protocols rely on minimal message sizes, while allowing agents to carry out non-trivial operations, such as counting. Exploring the tradeoff between communication complexity and internal computational simplicity remains an interesting direction for future work.

Our results also provide a theoretical perspective on collective behavior observed in biological systems (see also Section \ref{house-hunting}). In cooperative transport by longhorn “crazy ants” \cite{gelblum2015ant,gelblum2016emergent,fonio2016locally}, directional information is conveyed indirectly through the motion of a shared load. Prior work \cite{gelblum2015ant} showed that such information transfer is possible under certain population sizes, although scalability to large groups remained unresolved.
The noisy $\pull(h)$ model captures different regimes of such indirect communication. When the noise in estimating the aggregate force $F$ scales as $\sqrt{F}$, corresponding to large $h$, our results imply that even a single informed agent can quickly influence the entire population. In contrast, higher-variance estimates of $F$ can be associated with more localized sampling (constant $h$), for which fast dissemination from a small number of sources is impossible, but can be recovered by increasing the number of sources.

Overall, our results suggest that sensing aggregate tendencies, as observed in cooperative transport and related systems \cite{gelblum2015ant,gorbonos2016long,gelblum2016emergent,sarfati2021self}, can enable efficient and scalable information propagation even under strong noise and minimal communication.

\section*{Acknowledgements}
The authors thank Ofer Feinerman for the helpful and insightful discussions regarding ants. AK was supported by the Israel Science Foundation (\url{https://www.isf.org.il}) grant 1574/24. ND was supported by the AID INRIA-DGA project no.~2023000872 ``BioSwarm.'' EN was supported by the French National Research Agency through UCA JEDI (ANR-15-IDEX-01), EUR DS4H (ANR-17-EURE-004), and the 3IA Côte d'Azur Investments project ANR-23-IACL-0001.

\section{Handling Non-Uniform Noise}  \label{sec:non-uniform}
In this section, we use results from \cite{draft_symm_noise} to demonstrate that, given a~$\delta$-upper bounded noise matrix~$N$, agents can apply an artificial noise~$P$ to incoming messages such that the combined effect of~$N$ and~$P$ is~$\delta'$-uniform, with~$\delta'$ slightly exceeding~$\delta$.
This will allow us to focus on uniform noise when proving the correctness of our algorithms in \Cref{sec:upper_bound}.
In what follows, we formalize this statement.

%In this section, we explain how to extend the results of \Cref{sec:upper_bound} to the case in which the noise matrix is not $\delta$-uniform, following the intuition provided in Section \ref{ssec:nonunif_noise_intuition}.

\begin{definition} [Simulation with artificial noise] \label{def:simulation}
    Given a protocol~$\mathcal{A}$ and a (stochastic) noise  matrix~$P$, we say that agents {\em simulate~$\mathcal{A}$ with artificial noise~$P$} when they perform the following operations.
    In every round, each agent replaces every sampled message~$i$ by a random~$\sigma \in \Sigma$, such that~$\Pr(\sigma = j) = P_{ij}$. Each agent then applies protocol~$\mathcal{A}$ to the resulting set of modified messages.
\end{definition}

\begin{definition} \label{def:delta}
    Let~$f : [ 0, \frac{1}{d} )  \rightarrow \bbR$ s.t. $f(0) = 0$, and for every~$\delta \in \pa{ 0, \frac{1}{d} }$,
    \begin{equation*}
        f(\delta) := \pa{d + \frac{1}{2} \cdot \frac{1}{(d-1)^2} \cdot \frac{1-d\, \delta}{\delta}}^{-1}.
    \end{equation*}
\end{definition}
See \Cref{fig:function_f} for an illustration.

\begin{theorem}
\label{thm:reduction}
    Let~$N$ be a~$\delta$-upper bounded noise matrix on an alphabet~$\Sigma$ of size $d$, and let~$\mathcal{A}$ be a protocol.
    Let~$\delta' = f(\delta)$ as given by \Cref{def:delta}.
    There exists a stochastic matrix~$P = P(N,\delta)$, such that the simulation of~$\mathcal{A}$ with artificial noise~$P$ (\Cref{def:simulation}) under noise~$N$ has the same distribution as Protocol~$\mathcal{A}$ under a~$\delta'$-uniform noise matrix.
\end{theorem}
The proof of \Cref{thm:reduction} can be found in \cite{draft_symm_noise}.

\begin{figure}[htbp]
    \centering
    \includegraphics[width=1\linewidth]{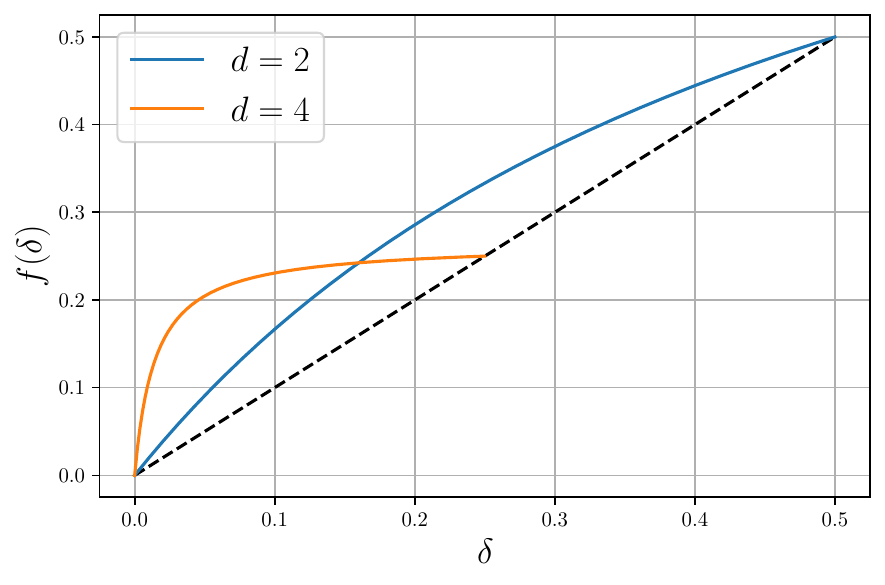}
    \caption{Plot of function~$f$ for 2 different values of~$d$.}
    \label{fig:function_f}
\end{figure}

\section{Deferred Proofs} \label{sec:upper_bound_proofs}

In this section we prove the statements from \Cref{sec:upper_bound}, in the case that the noise matrix~$N$ is~$\delta$-uniform.
\Cref{sec:useful_lemmas} establishes probabilistic results and \Cref{sec:prelim} introduces notations that are used in the analysis of both algorithms.
Finally, \Cref{sec:SF_analysis} and \Cref{sec:SSF_analysis} focus respectively on the analysis of Algorithms \SF and \SSF.
We recommend referring to \Cref{sec:intuition} for an informal overview of the proofs presented here.

\subsection{Preliminary Results in Probability} \label{sec:useful_lemmas}

\begin{definition}[Rademacher distribution]
    We say that a random variable~$X$ follows a Rademacher distribution with parameter $p$, and write~$X\sim \Rad(p)$, if it takes value in $\{-1,1\}$ and $\Pr\pa{X=1}=p$.
\end{definition}

The following result is a direct observation concerning binomial random variables with a small success probability parameter~$p$.

\begin{claim} \label{claim:binomial_0_average}
Let $X\sim\binomial\pa{n,p}$. If $n \, p \leq 1$, then
    \begin{equation*}
        \Pr\pa{X=1} \geq \frac{1}{e} n \, p .
    \end{equation*}
\end{claim}
\begin{proof}
    Let~$u_1 = 1$, and~$u_n = (1-\frac{1}{n})^{n-1}$ for $n \geq 2$. $(u_n)_{n \geq 1}$ is decreasing and~$\lim_{n \rightarrow +\infty} u_n = \frac{1}{e}$, so $u_n \geq \frac{1}{e}$ for every~$n \geq 1$.
    Therefore,
    \begin{align*}
        \Pr\pa{X=1} =  n \, p \, (1-p)^{n-1} &\geq n \,p \, u_n & \text{(since~$p \leq \tfrac{1}{n}$)} \\
        &\geq \frac{1}{e} n \, p,
    \end{align*}
    which concludes the proof of \Cref{claim:binomial_0_average}.
\end{proof}

As outlined informally in \Cref{sec:intuition}, our analysis will later introduce random variables~$X_k$ that take values in~$\{-1,0,1\}$.
We will focus specifically on the non-zero variables, treating them as Rademacher variables.
The following result formalizes this approach and addresses potential dependency issues.

\begin{lemma} \label{lemma:for_SF_and_SSF}
    Let~$\{X_i\}_{i=1,\dots,m}$ be i.i.d. random variables taking values in $\{-1,0,1\}$. Let~$Y$ be the number of~$X_i$ with non-zero value, i.e.,
    \begin{equation*}
        	Y = \Big| \{ i=1,\dots,m \, :\, X_i \neq 0 \} \Big|.
    \end{equation*}
    Let~$X=\sum_{i=1}^m X_i$. Conditioning on $\{Y=r\}$, $X$ is distributed as a sum of~$r$ Rademacher random variables with parameter $p= \Pr\pa{ X_i = 1 \mid X_i \neq 0}$.
\end{lemma}
\begin{proof}
    Let $a=(a_1,\dots, a_m)\in\{0,1\}^m$. Let~$\mathcal{E}(a)$ be the event that non-zero elements of~$\{ X_i \}_{i=1,\dots, m}$ corresponds to non-zero elements of~$a$, i.e.,
    \begin{equation*}
        \mathcal{E}(a) := \bigcap_{i=1}^m \left\{ |X_i| = a_i \right\}.
    \end{equation*}
    Similarly, we write~$\mathcal{E}_k(a)$ to denote the same event, without the condition on the element at index~$k$:
    \begin{equation*}
        \mathcal{E}_k(a):= \bigcap_{\substack{i=1 \\ i\neq k}}^m \left\{ |X_i| = a_i \right\}.
    \end{equation*}
    Note that all $\{\mathcal{E}(a)\}_{a \in \{0,1\}^m}$ are disjoint. Moreover, writing~$|a| = \sum_{i=1}^m a_i$, we have
    \begin{equation*}
        \{Y=r\} = \bigcup_{\substack{a\in \{0,1\}^m \\ |a| = r }} \mathcal{E}(a).
    \end{equation*}
    Since all~$\{ X_i \}_{i=1,\dots, m}$ are mutually independent, $X_k$ and $\mathcal{E}_k$ are also independent.
    We have
    \begin{align*}
        &\Pr\pa{ X_k = 1 \mid X_k\neq 0, Y=r} \\
        &= \frac{\Pr\pa{ X_k = 1, Y=r }}{\Pr\pa{X_k\neq 0, Y=r}}\\
        &=\frac{\sum_{|a|= r} \Pr\pa{ X_k = 1, \mathcal{E}(a)   }}{\sum_{|a|= r} \Pr\pa{ X_k \neq 0, \mathcal{E}(a)   }}\\
        &=\frac{\sum_{\substack{|a|= r \\ a_k=1}} \Pr\pa{ X_k = 1, \mathcal{E}(a)   }}{\sum_{\substack{ |a|= r \\ a_k=1}} \Pr\pa{ X_k \neq 0, \mathcal{E}(a) }} & \text{(when~$a_k = 0$, $\Pr\pa{ X_k \neq 0, \mathcal{E}(a) } = 0$)} \\
        &= \frac{\sum_{\substack{|a|= r \\ a_k=1 }} \Pr\pa{ X_k = 1, \mathcal{E}_k(a) }}{\sum_{\substack{|a|= r \\ a_k=1 }} \Pr\pa{ X_k \neq 0, \mathcal{E}_k(a) }}\\
        &= \frac{\sum_{\substack{|a|= r \\ a_k=1 }} \Pr\pa{ X_k = 1} \cdot \Pr\pa{ \mathcal{E}_k(a) }}{\sum_{\substack{|a|= r \\ a_k=1 }} \Pr\pa{ X_k \neq 0} \cdot \Pr\pa{ \mathcal{E}_k(a) }} & \text{($X_k$ and $\mathcal{E}_k(a)$ being independent)} \\
        &= \frac{\Pr\pa{ X_k = 1}}{\Pr\pa{ X_k \neq 0}} = \Pr\pa{ X_k = 1 \mid X_k\neq 0}.
    \end{align*}
    %This establishes the first statement in \Cref{lemma:for_SF_and_SSF}.
    Moreover, for any $x\in\mathbb{Z}$, we have
    \begin{align*}
        \Pr\pa{ X = x \mid Y=r} &= \Pr\pa{ \sum_i X_i = x \mid Y=r} \\
        &= \sum_{|a|= r} \Pr\pa{ \mathcal{E}(a) \mid Y=r}  \cdot \Pr\pa{ \sum_i X_i = x \mid \mathcal{E}(a) } \\
        &= \sum_{|a|= r} \Pr\pa{ \mathcal{E}(a) \mid Y=r }  \cdot \Pr\pa{ \sum_{i: a_i=1} X_i + \sum_{i: a_i = 0} X_i = x \mid \mathcal{E}(a) } \\
        &= \sum_{|a|= r} \Pr\pa{ \mathcal{E}(a)  \mid Y=r}  \cdot \Pr\pa{ \sum_{i: a_i=1} X_i = x \mid \mathcal{E}(a) } \\
        &= \sum_{|a|= r} \Pr\pa{ \mathcal{E}(a)  \mid Y=r}  \cdot \Pr\pa{ \sum_{i: a_i=1} X_i = x \mid \bigcap_{i: a_i=1} \{X_i \neq 0\} } \\
        &\underset{X_i \text{ i.i.d.}}{=} \sum_{|a|= r} \Pr\pa{ \mathcal{E}(a)  \mid Y=r}  \cdot \Pr\pa{ \sum_{i= 1}^r X_i = x \mid \bigcap_{i=1}^r \{ X_i \neq 0 \}}\\
        &= \Pr\pa{ \sum_{i= 1}^r X_i = x \mid \bigcap_{i=1}^r \{ X_i \neq 0 \}},
    \end{align*}
    which concludes the proof of \Cref{lemma:for_SF_and_SSF}.
\end{proof}

The next lemma recalls an existing lower bound on the probability to obtain more ``heads'' than ``tails'', when throwing a biased coin~$m$ times.

\begin{lemma}[From Lemma 9 in \cite{DBLP:journals/dc/FraigniaudN19}]
\label{lem:from_lemma_9}
    Let~$\theta > 0$ and~$B\sim \binomial(m,1/2+\theta)$.
    Then,
    \begin{align*}
        \Pr\pa{ B > \frac{m}{2}} - \Pr\pa{ B < \frac{m}{2}} \geq \min\left\{\sqrt{\frac{2m}{\pi}} \cdot  g(\theta, m),1\right\},
    \end{align*}
    where
    \begin{equation*}
        g(\theta, m)=
        \begin{cases}
            \theta \, \pa{1-\theta^2}^{\frac{m-1}{2}} & \text{if } \theta < \frac{1}{\sqrt{m}}, \\
            \frac{1}{\sqrt{m}} \, \pa{1-\frac{1}{m}}^{\frac{m-1}{2}} & \text{if } \theta \geq \frac{1}{\sqrt{m}}.
        \end{cases}
    \end{equation*}
\end{lemma}
\begin{proof}
    We note that
    \begin{align*}
        \Pr\pa{ B > \frac{m}{2}} - \Pr\pa{ B < \frac{m}{2}} &= \Pr\pa{ B \geq \frac{m}{2}} - \Pr\pa{ B \leq \frac{m}{2}} \\
        &= \sum_{k = \lceil m/2 \rceil}^{m} \binom{m}{k} \, p^{k} (1-p)^{m-k} - \sum_{k = \lceil m/2 \rceil}^{m} \binom{m}{k} \, p^{m-k} (1-p)^{k},
    \end{align*}
    and then we proceed as in the proof of~\cite[Lemma 9]{DBLP:journals/dc/FraigniaudN19}. Please note that in this paper, the definition of~$g$ contains a typographical error; the correct definition is the one provided here.
\end{proof}

Next, we rewrite the previous statement in a more convenient way for our analysis.

\begin{lemma} \label{lemma:majority_boosting}
    Let $0 \leq \theta < 1/2$ and $m>0$, and consider i.i.d.~random variables $\{ X_i \}_{i=1,\dots, m} \sim \Rad\pa{ \frac{1}{2}+\theta}$.
    Then, $X=\sum_{i=1}^m X_i$ satisfies
    \begin{equation*}
        \Pr\pa{ X>0 } - \Pr\pa{ X<0 } \geq \sqrt{\frac{2}{\pi e}} \cdot \min\{\sqrt{m} \,  \theta,1\}.
    \end{equation*}
\end{lemma}
\begin{proof}
    %This proof consists only in the employ of \Cref{lem:from_lemma_9} and the relationship between Rademacher and Bernoulli random variables.
    For~$i \in \{1, \dots, m\}$, let~$B_i = (X_i + 1)/2$, and let $B=\sum_{i=1}^m B_i$. By definition, $B\sim \binomial(m, \frac{1}{2}+\theta )$ and $\{X>0\}= \{B>\frac{m}{2}\}$.
    %and $\{X<0\}= \{B<\frac{m}{2}\}$. Now consider the two following cases:
    \begin{itemize}
        \item First, consider the case that~$\theta \geq \frac{1}{\sqrt{m}}$. By \Cref{lem:from_lemma_9}, we have
        \begin{equation*}
            \Pr\pa{X > 0 } - \Pr\pa{X < 0} \geq \min\left\{ \sqrt{\frac{2m}{\pi}} \cdot \frac{(1-1/m)^{\frac{m-1}{2}}}{\sqrt{m}} ,1\right\} \geq \min\left\{ \sqrt{\frac{2m}{\pi e}},1\right\},
        \end{equation*}
        where we used the fact $x \mapsto (1-1/x)^{\frac{x-1}{2}}$ is decreasing and that $\lim_{x\to\infty} (1-1/x)^{\frac{x-1}{2}} = e^{-1/2}$.
        \item Now, consider the case that~$\theta < \frac{1}{\sqrt{m}}$.
        In that case,
        \begin{equation*}
            \pa{1-\theta^2}^{\frac{m-1}{2}} \geq  \pa{1-\frac{1}{m}}^{\frac{m-1}{2}} \geq e^{-1/2},
        \end{equation*}
        and by \Cref{lem:from_lemma_9}, we have
        \[
            \Pr\pa{X > 0 } - \Pr\pa{X < 0} \geq  \sqrt{\frac{2m}{\pi e}} \theta ,
        \]
    which concludes the proof of \Cref{lemma:majority_boosting}.
\end{itemize}
\end{proof}

We now prove the technical master lemma stated as \Cref{lemma:first_phase}.

\begin{proof}
    We begin the proof with a technical computation that will be used multiple times.
    \begin{claim}
    \label{claim:sigma_delta}
        If $\delta < \frac{s_0 + s_1}{2n}(1-|\Sigma|\delta)$, then $\pa{ 1-|\Sigma|\delta }^2 \geq \frac{1}{4}$ for~$n$ large enough.
    \end{claim}
    \begin{proof}[Proof of \Cref{claim:sigma_delta}]
        We have
        \begin{equation*}
            \delta < \frac{s_0 + s_1}{2n}(1-|\Sigma|\delta) \leq \frac{1}{2}(1-|\Sigma|\delta),
        \end{equation*}
        where in the last inequality we used that $s_0+s_1 \leq n$. This inequality can be rewritten as
        \[
            \delta \leq \frac{1}{2+ |\Sigma|}.
        \]
        This implies that
        \begin{equation*} %\label{eq:sigma_delta}
            \pa{ 1-|\Sigma|\delta }^2 \geq \pa{ 1-\frac{|\Sigma|}{2+ |\Sigma|} }^2 \geq \frac{1}{4},
        \end{equation*}
        where in the last inequality we used that, since $|\Sigma| \geq 2$, we have $\frac{|\Sigma|}{2+ |\Sigma|} \leq \frac{1}{2}$.  This concludes the proof of \Cref{claim:sigma_delta}.
    \end{proof}
    Now, let~$Y$ denote the number of $k$ for which $X_k \neq 0$. Formally,
    \begin{equation*}
        Y = \Big| \{ k=1,\dots,m \, :\, X_k \neq 0 \} \Big|.
    \end{equation*}
    $Y$ follows a binomial distribution with parameters~$\Pr\pa{ X_k\neq 0}$, and~$m$.
    %The idea is to apply \Cref{lemma:for_SF_and_SSF}, to obtain that, conditioning on $\{Y=r\}$, $X$ is a sum of $r$ Rademacher random variables with parameter $p= \Pr\pa{ X_k = 1 \mid X_k \neq 0}$. Then, we basically distinguish two types of regimes: in the first (i.e., $\bbE(Y) < 1$), for almost all $k$, $X_k=0$, but whenever $X_k\neq 0$ there is a good chance that $X_k = 1$. In the second regime (i.e. $\bbE(Y) \geq 1$), for many $k$, $X_k\neq 0$, but on the other hand the probability of getting a 1 or a -1 are very similar. In both cases, these two aspects will balance each other out.
    The next two claims consider the case that~$Y$ is relatively large, namely, that~$\bbE(Y) \geq 1$.
    In this regime, the sum~$X$ contains enough information about the bias, because many~$X_k$s are different from~$0$.
    In \Cref{lemma:r_is_big}, we translate the lower-bound on the expectation of~$Y$, into a lower-bound that holds with constant probability.
    %is large with constant probability, while
    Then, in \Cref{claim:big_expectations}, we prove the statement in the case that~$Y$ is sufficiently large.
    \begin{claim}
    \label{lemma:r_is_big}
        If $ \bbE(Y) \geq 1$, there is a constant $c_2>0$, such that
        \begin{equation*}
            \Pr \pa{ Y > \frac{\bbE(Y)}{2}} \geq c_2.
        \end{equation*}
    \end{claim}
    \begin{proof}[Proof of \Cref{lemma:r_is_big}]
        Recall that~$Y$ follows a binomial distribution with parameter~$m,\Pr\pa{ X_k\neq 0}$.
        We have
        \begin{align*}
            \Pr \pa{ Y \leq \frac{1}{2} \bbE(Y) } &\leq  \Pr \pa{ Y \leq (1-\tfrac{1}{2}) \, \bbE(Y) } \\
            &\leq \exp\pa{-\frac{1}{8} \bbE(Y) } & \text{(by the \nameref{thm:multiplicative_chernoff_bound})} \\
            &\leq \exp\pa{ -\frac{1}{8} }, & \pa{\bbE(Y) \geq 1}.
        \end{align*}
        We conclude the proof of \Cref{lemma:r_is_big} by taking~$c_2 = 1- \exp\pa{ -\frac{1}{8} }$.
    \end{proof}

    \begin{claim} \label{claim:big_expectations}
        If $\bbE(Y)\geq 1$, and if~$c_1$ from \Cref{eq:zero_th_claim} is large enough, then for every~$r \geq \bbE(Y)/2$,
        \[
        \Pr\pa{X > 0 \mid Y=r } - \Pr\pa{X < 0 \mid Y=r } \geq \frac{8}{c_2}\,\sqrt{\frac{\log n}{n}}.
        \]
    \end{claim}
        \begin{proof}[Proof of \Cref{claim:big_expectations}]
            Let~$r \geq \bbE(Y)/2$.
            By \Cref{lemma:for_SF_and_SSF}, conditioning on $\{Y=r\}$, $X$ is a sum of $r$ Rademacher random variables of parameter~$p := \Pr\pa{ X_k=1 \mid X_k\neq 0}$. Then, we can apply \Cref{lemma:majority_boosting} to obtain
            \begin{equation} \label{eq:majority_boosting}
                \Pr\pa{X > 0 \mid Y=r } - \Pr\pa{X < 0 \mid Y=r } \geq \sqrt{\frac{2}{\pi \, e}} \min \Bigg\{ \sqrt{r} \pa{p-\frac{1}{2}}, 1 \Bigg\}.
            \end{equation}
            We distinguish between two cases, depending on the value of~$\delta$.
            \begin{itemize}
            \item {\bf Case 1.} Assume that $\delta < \frac{s_0 + s_1}{2n}(1-|\Sigma|\delta)$. Then,
            \begin{equation*}
                \Pr \pa{X_k \neq 0} \underset{\text{\Cref{eq:first_claim}}}{\geq} (1-|\Sigma|\delta)^2 \cdot \frac{s_0 + s_1}{2n} + \delta \underset{\text{\Cref{claim:sigma_delta}}}{\geq} \frac{s_0 + s_1}{8n}.
            \end{equation*}
            Therefore, by \Cref{eq:zero_th_claim},
            \begin{equation*}
                \bbE(Y) = m \cdot \Pr \pa{X_k \neq 0} \geq c_1 \cdot \frac{(s_0+s_1) \log n}{s^2} \cdot\frac{s_0 + s_1}{8n},
            \end{equation*}
            and thus,
            \begin{equation} \label{eq:r_bound}
                r \geq \frac{\bbE(Y)}{2} \geq c_1 \cdot \frac{(s_0+s_1)^2 \log n}{16\,s^2 n}.
            \end{equation}
            %From our assumption on~$\delta$, we can use \Cref{eq:third_claim} to obtain that $p \geq \frac{1}{2}+\frac{s}{4(s_0+s_1)}$.
            Finally,
            \begin{align*}
                &\Pr\pa{X > 0 \mid Y=r } - \Pr\pa{X < 0 \mid Y=r } \\
                &\geq \sqrt{\frac{2}{\pi \, e}} \min \Bigg\{ \sqrt{r} \pa{p-\frac{1}{2}}, 1 \Bigg\} & \text{(by \Cref{eq:majority_boosting})} \\
                &\geq \sqrt{\frac{2}{\pi \, e}} \min \Bigg\{ \sqrt{c_1 \cdot \frac{(s_0+s_1)^2 \log n}{16\,s^2 n}} \cdot \frac{s}{4(s_0+s_1)}, 1 \Bigg\} & \text{(by \Cref{eq:third_claim,eq:r_bound})} \\
                &= \sqrt{\frac{2}{\pi \, e}} \min \Bigg\{ \frac{\sqrt{c_1}}{16} \sqrt{\frac{\log n}{ n}}, 1 \Bigg\}\\
                &\geq \frac{8}{c_2} \sqrt{\frac{\log n}{n}}, & \text{(as long as~$c_1$ is large enough)}%& \text{(taking~$c_3 := \frac{1}{8\sqrt{2 \pi e}}$)}
            \end{align*}
            which concludes the proof of Case 1.

            \item {\bf Case 2.}
            Now, assume that $\delta \geq \frac{s_0 + s_1}{2n}(1-|\Sigma|\delta)$. We have
            \begin{equation} \label{eq:r_bound2}
                r \geq \frac{E(Y)}{2} = \frac{1}{2} \cdot m \, \Pr\pa{X_k \neq 0} \underset{\text{\Cref{eq:first_claim}}}{\geq} \frac{1}{2} \cdot m \, \delta \underset{\text{\Cref{eq:zero_th_claim}}}{\geq} c_1 \cdot \frac{n \, \delta^2 \log n}{2 \min\{s^2,n\}  (1-|\Sigma|\delta)^2}.
            \end{equation}
            Finally,
            \begin{align*}
                &\Pr\pa{X > 0 \mid Y=r } - \Pr\pa{X < 0 \mid Y=r } \\
                &\geq \sqrt{\frac{2}{\pi \, e}} \min \Bigg\{ \sqrt{r} \pa{p-\frac{1}{2}}, 1 \Bigg\} & \text{(by \Cref{eq:majority_boosting})} \\
                &\geq \sqrt{\frac{2}{\pi \, e}} \min \Bigg\{ \sqrt{c_1} \cdot \frac{\delta \sqrt{n \log n} }{\min\{s,\sqrt{n}\}  \, (1-|\Sigma|\delta)} \cdot \frac{s}{n}\frac{(1-|\Sigma|\delta)}{8\,\delta}, 1 \Bigg\} & \text{(by \Cref{eq:r_bound2,eq:second_claim})} \\
                &\geq \sqrt{\frac{2}{\pi \, e}} \min \Bigg\{ \frac{\sqrt{c_1}}{8} \sqrt{\frac{\log n}{n}} , 1 \Bigg\} \geq \frac{8}{c_2} \sqrt{\frac{\log n}{n}},
            \end{align*}
            where as above, $c_1$ is assumed to be large enough.
            \end{itemize}
            This concludes the proof of \Cref{claim:big_expectations}.
        \end{proof}

        Now, we consider the case that~$Y$ is relatively small, that is, $\bbE(Y) \leq 1$.
        In this regime, few~$X_k$s are different from~$0$, but each non-zero~$X_k$ is equal to~$1$ with a sufficiently large probability.
        The following result formalizes this claim, by focusing on the event $\{Y=1\}$.
        %We calculate the probability of such event and the corresponding bias.
    \begin{claim} \label{claim:low_expectation}
        If $\bbE(Y) \leq 1$, we have
        \[
            \Pr\pa{ Y=1 } \cdot \big[ \Pr\pa{ X>0 \mid  Y=1 } - \Pr\pa{ X<0 \mid  Y=1 } \big] \geq 8 \sqrt{\frac{\log n}{n}}.
        \]
    \end{claim}
    \begin{proof}[Proof of \Cref{claim:low_expectation}]
        We have
        \begin{align*}
            \Pr\pa{Y=1 } &\geq \frac{1}{e} m \cdot \Pr \pa{X_k \neq 0} & \text{(by \Cref{claim:binomial_0_average})} \\
            &\geq \frac{1}{e} m \cdot \pa{(1-|\Sigma|\delta)^2 \cdot \frac{s_0 + s_1}{2n} + \delta} & \text{(by \Cref{eq:first_claim})} \\
            &\geq \frac{\sqrt{n}\log n}{e \, s} \pa{(1-|\Sigma|\delta)^2 \cdot \frac{s_0 + s_1}{2n} + \delta}. & \text{(by \Cref{eq:zero_th_claim})}
        \end{align*}
        %By \Cref{eq:zero_th_claim}, this implies
        %\begin{equation} \label{eq:probability_A1}
        %    \Pr\pa{Y=1 } \geq \frac{\sqrt{n}\log n}{e \, s} \pa{(1-|\Sigma|\delta)^2 \cdot \frac{s_0 + s_1}{2n} + \delta}.
        %\end{equation}
        Conditioning on~$\{Y=1\}$, $X$ is a single Rademacher random variable with parameter~$p$.
        Therefore,
        \begin{equation*} %\label{eq:simple_rademacher}
            \Pr\pa{X > 0 \mid Y=1 } - \Pr\pa{X < 0 \mid Y=1 } = 2\pa{p-\frac{1}{2}},
        \end{equation*}
        and together with the last equation, we obtain
        \begin{multline} \label{eq:single_equation}
            \Pr\pa{ Y=1 } \cdot \big[\Pr\pa{ X>0 \mid  Y=1 } - \Pr\pa{ X<0 \mid  Y=1 } \big] \\ \geq \frac{\sqrt{n}\log n}{e \, s} \pa{(1-|\Sigma|\delta)^2 \cdot \frac{s_0 + s_1}{2n} + \delta} \cdot 2\pa{p-\frac{1}{2}}.
        \end{multline}
        As in the proof of \Cref{claim:big_expectations}, we distinguish between two cases, depending on the value of~$\delta$.
        \begin{itemize}
            \item {\bf Case 1.} Consider the case that $\delta < \frac{s_0 + s_1}{2n}(1-|\Sigma|\delta)$. We can rewrite \Cref{eq:single_equation} as
            \begin{align*}
                &\Pr\pa{ Y=1 } \cdot \big[\Pr\pa{ X>0 \mid  Y=1 } - \Pr\pa{ X<0 \mid  Y=1 } \big] \\
                &\geq \frac{\sqrt{n}\log n}{e \, s} \pa{(1-|\Sigma|\delta)^2 \cdot \frac{s_0 + s_1}{2n} + \delta} \cdot 2\pa{p-\frac{1}{2}} \\
                &\geq \frac{\sqrt{n}\log n}{e \, s} \cdot \frac{s_0 + s_1}{8n} \cdot 2\pa{p-\frac{1}{2}} & \text{(by \Cref{claim:sigma_delta})} \\
                &\geq \frac{\sqrt{n}\log n}{e \, s} \cdot \frac{s_0 + s_1}{8n} \cdot 2 \cdot \frac{s}{4(s_0+s_1)} & \text{(by \Cref{eq:third_claim})} \\
                &= \frac{\log n}{16\,e \, \sqrt{n}} \geq 8 \sqrt{\frac{\log n}{n}}.
            \end{align*}

            \item {\bf Case 2.} Now, consider the case that~$\delta \geq \frac{s_0 + s_1}{2n}(1-|\Sigma|\delta)$. Again, we rewrite \Cref{eq:single_equation} as
            \begin{align} \nonumber
                &\Pr\pa{ Y=1 } \cdot \big[\Pr\pa{ X>0 \mid  Y=1 } - \Pr\pa{ X<0 \mid  Y=1 } \big] \\ \nonumber
                &\geq \frac{\sqrt{n}\log n}{e \, s} \pa{(1-|\Sigma|\delta)^2 \cdot \frac{s_0 + s_1}{2n} + \delta} \cdot 2\pa{p-\frac{1}{2}} \\ \nonumber
                &\geq \frac{\sqrt{n}\log n}{e \, s} \cdot \delta \cdot 2\pa{p-\frac{1}{2}} \\ \nonumber
                &\geq \frac{\sqrt{n}\log n}{e \, s} \cdot \delta \cdot 2 \cdot \frac{s}{n}\frac{(1-|\Sigma|\delta)}{8\,\delta} & \text{(by \Cref{eq:second_claim})} \\ \label{eq:Y=1}
                &= \frac{\log n}{4e\sqrt{n}}  \cdot (1-|\Sigma|\delta).
            \end{align}
            By \Cref{eq:first_claim}, $\bbE(Y)= m \Pr \pa{X_k \neq 0}\geq m \delta$ , therefore $\frac{1}{m} > \frac{\bbE(Y)}{m} \geq \delta$. For $n\rightarrow \infty$, we have $m \geq \log n \rightarrow \infty$ and therefore $\delta \rightarrow 0$. Then, for $n$ large enough, $(1-|\Sigma|\delta) \geq \frac{1}{2}$ and by \Cref{eq:Y=1} we finally obtain
            \[
                \Pr\pa{ Y=1 } \cdot \big[\Pr\pa{ X>0 \mid Y=1 } - \Pr\pa{ X<0 \mid Y=1 }\big] \geq \frac{\log n}{8e\sqrt{n}} \geq 8 \sqrt{\frac{\log n}{n}}.
            \]
        \end{itemize}
        This concludes the proof of \Cref{claim:low_expectation}.
    \end{proof}
    Now we have the tools to conclude the proof of \Cref{lemma:first_phase}. If $\bbE(Y)\geq 1$,
    \begin{align*}
        &\Pr\pa{ X > 0  } - \Pr\pa{ X < 0  } \\
        &\geq \sum_{r=\bbE(Y)/2}^m \Pr\pa{ Y=r } \cdot \big[ \Pr\pa{ X>0 \mid  Y=r } - \Pr\pa{ X<0 \mid  Y=r } \big] \\
        &\geq \frac{8}{c_2} \sqrt{\frac{\log n}{n}} \cdot \Pr\pa{ Y \geq \frac{\bbE(Y)}{2} } &\pa{\text{by \Cref{claim:big_expectations}}}\\
        &\geq 8\sqrt{\frac{\log n}{n}}. &\pa{\text{by \Cref{lemma:r_is_big}}}
    \end{align*}
    If, instead, $\bbE(Y) < 1$,
    \begin{align*}
        \Pr\pa{ X > 0  } - \Pr\pa{ X < 0  }
        &\geq \Pr\pa{ Y=1 } \cdot \big[ \Pr\pa{ X>0 \mid  Y=1 } - \Pr\pa{ X<0 \mid  Y=1 } \big] \\
        &\geq 8\sqrt{\frac{\log n}{n}}. &\pa{\text{by \Cref{claim:low_expectation}}}
    \end{align*}
    This concludes the proof of \Cref{lemma:first_phase}.
    %, by taking~$c_1 = \min\{ 1,c_2 \, c_3 \}$.
\end{proof}

\subsection{\texorpdfstring{Analysis of Algorithm \SF}{Analysis of Algorithm SF}} \label{sec:SF_analysis}

In this section, we present the proof of \Cref{thm:upper}. Consequently, we will always assume that the parameter~$m$ used in Algorithm \SF satisfies \Cref{eq:assumption_m_SF}.
Recall that we assume the noise matrix~$N$ to be $\delta$-uniform. This assumption will be relaxed in \Cref{sec:final_proof_SF} when proving \Cref{thm:upper}, by using the results of \Cref{sec:non-uniform}.

\subsubsection{\texorpdfstring{Correctness of weak-opinions in Algorithm \SF}{Correctness of weak-opinions in Algorithm SF}}

We now prove \Cref{lemma:first_phase_SF}.
\begin{proof}
    Let~$i \in I$.
    Recall that in each of the first two phases of Algorithm \SF, Agent~$i$ receives at least~$m$ messages.
    Let $A^{(i)}, B^{(i)} \in \{0,1\}^m$ be a random vector containing the first~$m$ messages received in Phase 0 and Phase 1, respectively.
    %We recall that, since in the first phase all non-source agents display always 0s, in $A^{(i)}$ we have mostly all 0s, except for the messages that comes from the 1-sources and for the messages corrupted by noise. Similarly, in $B^{(i)}$ there are mostly all 1s.
    %Consider the couples $(A_k^{(i)}, B_k^{(i)})$ for all $k=1,\dots, m$. More precisely,
    By construction, Agent~$i$ observes a 0 in Phase~$0$ if and only if it received a corrupted message from a 1-source, or an uncorrupted message from either a 0-source or a non-source agent.
    By applying a similar reasoning to Phase 1, we compute for every~$k \in \{1,\dots, m\}$,
    \begin{align*}
        \Pr\pa{ A_k^{(i)} = 0 } &= \frac{s_1}{n} \delta + \pa{ 1- \frac{s_1}{n}}(1-\delta)
        & \Pr\pa{ B_k^{(i)} = 0 } &= \frac{s_0}{n}(1-\delta) + \pa{ 1- \frac{s_0}{n} } \delta\\
        \Pr\pa{ A_k^{(i)} = 1 } &= \frac{s_1}{n}(1-\delta) + \pa{ 1- \frac{s_1}{n} } \delta
        &\Pr\pa{ B_k^{(i)} = 1 } &= \frac{s_0}{n} \delta + \pa{ 1- \frac{s_0}{n}}(1-\delta).
    \end{align*}
    %By definition of Algorithm \SF, the value of~the message received (and thus~$(A_k^{(i)}, B_k^{(i)})$) only depends on the sample $S_t^{(i)}$, on the realization of the noise affecting the corresponding message (and on individual coin tosses in case of ties). As a consequence, each pair~$(A_k^{(i)}, B_k^{(i)})$ is a deterministic function of a different~$N_{t,\ell}^{(i)}$ and $S_t^{(i)}$; therefore, they are all mutually independent.
    Let us define the random variables $\{X_k^{(i)}\}_{k\in \{1,\dots, m\}}$ s.t.
    \begin{align*}
        X_k^{(i)} =\begin{cases}
            1 \quad \text{if } (A_k^{(i)},B_k^{(i)})=(1,1) \\
            0 \quad \text{if } (A_k^{(i)},B_k^{(i)})=(0,1), \text{ or } (A_k^{(i)},B_k^{(i)})=(1,0)\\
            -1 \quad \text{if } (A_k^{(i)},B_k^{(i)})=(0,0).
        \end{cases}
    \end{align*}
    Let $X^{(i)}:= \sum_{k=1}^m X_k^{(i)}$.
    Informally, each $X_k^{(i)}$ is assigned value~$+1$ or~$-1$ if the corresponding pair of messages tends to indicate that the correct preference is respectively~$1$ or~$0$, and it is assigned value~$0$ otherwise.
    %compares the two phases, and it is 0 in the cases there isn't a step forward detecting which is the correct opinion, i.e. both counters raise or stay still, while it is 1 (-1) when there is a step toward the correct (incorrect) opinion, i.e. only one counter raises.
    Recall that~$C_t^{(i)}$ denotes the value of the coin toss performed by Agent~$i$ in round~$t$.
    By definition of Algorithm \SF, we have
    \begin{equation} \label{eq:weak_opn_equiv_SF_first_phases}
        \Y^{(i)} = 1 \iff \begin{cases}
            X^{(i)} > 0, &\text{or} \\
            X^{(i)} = 0 \text{ and } C_t^{(i)} = 1.
        \end{cases}
    \end{equation}
    Note that the value of the messages received in Phases 0 and 1 only depends on the identity of the sampled agents, and on the realization of the noise. As a consequence, each~$A_k^{(i)}$ and each~$B_k^{(i)}$ is a deterministic function of a corresponding pair~$(N_{t,\ell}^{(i)},S_{t,\ell}^{(i)})$, and so is each~$X_k^{(i)}$. Therefore, %all~$\{X_k^{(i)}\}_{k \in \{1,\ldots,m\}}$ are mutually independent.
    since the variables~$\left\{ S_t^{(i)}, N_t^{(i)}, C_t^{(i)} \right\}_{i \in I, t \in \bbN}$ are mutually independent,
	it follows that all $\{\Y^{(i)}\}_{i \in I}$ are also mutually independent.
     %Since $A_k^{(i)},B_k^{(i)}$ are mutually independent,
    Moreover, we can write
    \begin{align*}
        \Pr\pa{ X_k^{(i)} = 1 } &= \Pr\pa{ A_k^{(i)} = 1 } \cdot \Pr\pa{ B_k^{(i)} = 1 }, \\
        \Pr\pa{ X_k^{(i)} = 0 } &= \Pr\pa{ A_k^{(i)} = 0 } \cdot \Pr\pa{ B_k^{(i)} = 1 } + \Pr\pa{ A_k^{(i)} = 1 } \cdot \Pr\pa{ B_k^{(i)} = 0 }, \\
        \Pr\pa{ X_k^{(i)} = -1 } &= \Pr\pa{ A_k^{(i)} = 0 } \cdot \Pr\pa{ B_k^{(i)} = 0 }.
    \end{align*}
    The proof of the following result is straightforward. %, and is therefore deferred to the Appendix.
    \begin{claim} \label{claim:simple_ineq}
        For every~$s_0,s_1 \in \bbN$ such that~$s_0 + s_1 \leq n$, we have that~$\frac{s_0 s_1}{n^2}\leq\frac{s_0+s_1}{4n}$.
    \end{claim}
    \begin{proof}
        Let~$x = s_0/n$ and~$S = (s_0+s_1)/n \in [0,1]$. We have
        \begin{equation*}
            \frac{s_0 s_1}{n^2}\leq\frac{s_0+s_1}{4n} \iff x(S-x) \leq \frac{S}{4} \iff P(x) := x^2 - Sx + \frac{S}{4} \geq 0.
        \end{equation*}
        If~$S \in [0,1)$, $P$ is always positive, and if~$S=1$, $P$ has a unique root at~$x=\frac{1}{2}$. In both cases, $P$ is non-negative, which concludes the proof of \Cref{claim:simple_ineq}.
    \end{proof}
    In order to apply \Cref{lemma:first_phase} to $X^{(i)}$, we first verify that all its conditions are satisfied.
    \begin{claim}
    \label{claim:p_SF}
        We have
        \begin{equation} \label{eq:first_claim_OLD}
             \Pr\pa{ X_k^{(i)}\neq 0 } \geq (1-2\delta)^2 \cdot \frac{s_0 + s_1}{2n} + \delta
        \end{equation}
        Moreover, if $\delta \geq \frac{s_0 + s_1}{2n}(1-2\delta)$,
         \begin{equation} \label{eq:second_claim_OLD}
            \Pr\pa{ X_k^{(i)} = 1 \mid X_k^{(i)} \neq 0} \geq \frac{1}{2} + \frac{s}{n}\frac{(1-2\delta)}{8\,\delta},
        \end{equation}
        while if $\delta < \frac{s_0 + s_1}{2n}(1-2\delta)$,
        \begin{equation} \label{eq:third_claim_OLD}
            \Pr\pa{ X_k^{(i)} = 1 \mid X_k^{(i)} \neq 0} \geq \frac{1}{2} + \frac{s}{4 (s_0+s_1)}
        \end{equation}
    \end{claim}
    \begin{proof}[Proof of \Cref{claim:p_SF}]
        We first prove \Cref{eq:first_claim_OLD}.
        \begin{align}
            \nonumber
            &\Pr\pa{ X_k^{(i)}\neq 0} \\
            \nonumber &=\Pr\pa{ A_k^{(i)} = 1 } \cdot \Pr\pa{ B_k^{(i)} = 1 } + \Pr\pa{ A_k^{(i)} = 0 } \cdot \Pr\pa{ B_k^{(i)} = 0 }\\ \nonumber
            &=\pa{\frac{s_1}{n}(1-\delta) + \pa{ 1- \frac{s_1}{n} } \delta} \cdot
            \pa{ \frac{s_0}{n} \delta + \pa{ 1- \frac{s_0}{n}}(1-\delta)} \\ \nonumber
            & + \pa{\frac{s_1}{n} \delta + \pa{ 1- \frac{s_1}{n}}(1-\delta)  } \cdot
            \pa{ \frac{s_0}{n}(1-\delta) + \pa{ 1- \frac{s_0}{n} } \delta }\\ \nonumber
            &= \frac{s_0+s_1}{n} - 2 \frac{s_0 s_1}{n^2} + 2 \, \delta \left[ 1 - 2 \frac{s_0 + s_1}{n} + 4\frac{s_0 s_1}{n^2} \right] + 2 \, \delta^2 \left[ -1 + 2 \frac{s_0+s_1}{n} - 4 \frac{s_0 s_1}{n^2}  \right] \\
            &= 2\delta (1-\delta) + (1-2\delta)^2 \pa{ \frac{s_0+s_1}{n}-2 \frac{s_0 s_1}{n^2} }.\label{eq:sum_events}
        \end{align}
        Since $2\delta (1-\delta) \geq \delta$ whenever $\delta\in \left[ 0,1/2 \right]$, and since $\frac{s_0 s_1}{n^2}\leq\frac{s_0+s_1}{4n}$ by \Cref{claim:simple_ineq}, we obtain
        \begin{align*}
            \Pr\pa{ X_k^{(i)}\neq 0}
            %&\geq \delta + (1-2\delta)^2 \pa{ \frac{s_0}{n} \pa{1-2\frac{s_1}{n}} + \frac{s_1}{n} \pa{1-2\frac{s_0}{n}} }\\
            &\geq \delta + (1-2\delta)^2 \frac{s_0+s_1}{2n}.
        \end{align*}
        This establishes \Cref{eq:first_claim_OLD}.
        From \Cref{eq:sum_events}, since $1-2\delta \leq 1-\delta \leq 1$, we have
        \begin{equation} \label{eq:X_neq0_upper}
            \Pr\pa{ X_k^{(i)}\neq 0} \leq  2 \, \delta + (1-2\delta)\frac{s_0+s_1}{n}.
        \end{equation}
        Moreover,
        \begin{align}
            \nonumber
            &\Pr\pa{ X_k^{(i)} = 1} - \Pr\pa{ X_k^{(i)} = -1} \\
            \nonumber &=\Pr\pa{ A_k^{(i)} = 1 } \cdot \Pr\pa{ B_k^{(i)} = 1 } - \Pr\pa{ A_k^{(i)} = 0 } \cdot \Pr\pa{ B_k^{(i)} = 0 }\\ \nonumber
            &=\pa{\frac{s_1}{n}(1-\delta) + \pa{ 1- \frac{s_1}{n} } \delta} \cdot
            \pa{ \frac{s_0}{n} \delta + \pa{ 1- \frac{s_0}{n}}(1-\delta)} \\ \nonumber
            & - \pa{\frac{s_1}{n} \delta + \pa{ 1- \frac{s_1}{n}}(1-\delta)  } \cdot
            \pa{ \frac{s_0}{n}(1-\delta) + \pa{ 1- \frac{s_0}{n} } \delta }\\ \label{eq:X=1-X=-1}
            &=  \pa{ 1 - 2\delta} \frac{s_1-s_0}{n}
            = \pa{ 1 - 2\delta} \frac{s}{n}.
        \end{align}
        We have
        \begin{align*}
            \Pr\pa{ X_k^{(i)} = 1 \mid X_k^{(i)} \neq 0} - \Pr\pa{ X_k^{(i)} = -1 \mid X_k^{(i)} \neq 0}
            &= \frac{\Pr\pa{ X_k^{(i)} = 1} - \Pr\pa{ X_k^{(i)} = -1}}{\Pr\pa{ X_k^{(i)}\neq 0}},
        \end{align*}
        therefore, if $\delta \geq \frac{s_0 + s_1}{2n}(1-2\delta)$, by \Cref{eq:X_neq0_upper,eq:X=1-X=-1}, we obtain
        \[
            \Pr\pa{ X_k^{(i)} = 1 \mid X_k^{(i)} \neq 0} - \Pr\pa{ X_k^{(i)} = -1 \mid X_k^{(i)}\neq 0} \geq \frac{\pa{ 1 - 2\delta} \frac{s}{n}}{2 \, \delta + (1-2\delta)\frac{s_0+s_1}{n}} \geq \frac{\pa{ 1 - 2\delta} \frac{s}{n}}{4 \, \delta}.
        \]
        Instead, if $\delta < \frac{s_0 + s_1}{2n}(1-2\delta)$,  we have
        \begin{multline*}
            \Pr\pa{ X_k^{(i)} = 1 \mid X_k^{(i)} \neq 0} - \Pr\pa{ X_k^{(i)} = -1 \mid X_k^{(i)}\neq 0} \\
            \geq \frac{\pa{ 1 - 2\delta} \frac{s}{n}}{2 \, \delta + (1-2\delta)\frac{s_0+s_1}{n}} \geq \frac{\pa{ 1 - 2\delta} \frac{s}{n}}{2 (1-2\delta)\frac{s_0+s_1}{n}} = \frac{s}{2 (s_0+s_1)}.
        \end{multline*}
        We have proven that, depending on the value of $\delta$, for $x=\frac{s}{n}\frac{(1-2\delta)}{4\,\delta}$ or $x=\frac{s}{2 (s_0+s_1)}$, we have
        \begin{equation*}
            \Pr\pa{ X_k^{(i)} = 1 \mid X_k^{(i)} \neq 0} - \Pr\pa{ X_k^{(i)} = -1 \mid X_k^{(i)}\neq 0} \geq x.
        \end{equation*}
        Since $\Pr\pa{ X_k^{(i)} = -1 \mid X_k^{(i)}\neq 0} = 1 - \Pr\pa{ X_k^{(i)} = 1 \mid X_k^{(i)} \neq 0}$, we obtain
        \[
             2 \, \Pr\pa{ X_k^{(i)} = 1 \mid X_k^{(i)} \neq 0} - 1 \geq x.
        \] \Cref{eq:second_claim_OLD,eq:third_claim_OLD} follow from rearranging the inequality, which concludes the proof of \Cref{claim:p_SF}.
    \end{proof}
    Finally, we have
    \begin{align*}
        \Pr\pa{ \Y^{(i)} = 1  }
        &= \Pr\pa{X^{(i)}>0} + \frac{1}{2} \Pr\pa{X^{(i)}=0} & \text{(by \Cref{eq:weak_opn_equiv_SF_first_phases})} \\
        &= \Pr\pa{X^{(i)}>0} + \frac{1}{2} \pa{1- \Pr\pa{X^{(i)}>0} - \Pr\pa{X^{(i)}<0} } \\
        &= \frac{1}{2} + \frac{\Pr\pa{X^{(i)}>0} - \Pr\pa{X^{(i)}<0}}{2}. %\\
        %& \geq  \frac{1}{2} +\frac{1}{2} \sqrt{\frac{\log n}{n}}. %&\pa{\text{by \Cref{claim:p_SF,eq:assumption_m_SF,lemma:first_phase}}}
    \end{align*}
    \Cref{eq:assumption_m_SF,claim:p_SF} together imply that we can apply \Cref{lemma:first_phase}, in order to obtain
    \begin{equation*}
        \Pr\pa{ \Y^{(i)} = 1  } \geq  \frac{1}{2} + 4\sqrt{\frac{\log n}{n}},
    \end{equation*}
    which concludes the proof of \Cref{lemma:first_phase_SF}.

\end{proof}

\subsubsection{\texorpdfstring{Correctness of the majority boosting phase in Algorithm \SF}{Correctness of the majority boosting phase in Algorithm SF}}

Let~$L:= 10 \, \log n$ be the number of sub-phases lasting for~$\lceil w/h\rceil$ rounds, within the Majority Boosting phase (Phase~2) of Algorithm \SF. For~$\ell \in \{0,\ldots,L\}$, we write:
\begin{equation*}
    \Delta_\ell := \left| \left\{ i \in I, Y_t^{(i)} = 1, \text{where $t$ is the $\ell^\text{th}$ update round} \right\} \right| - \frac{n}{2}.
\end{equation*}
Informally, $\Delta_\ell$ denotes the number of agents holding as opinion the correct preference ``beyond half of the population'', during the~$\ell^\text{th}$ sub-phase. In particular, $\Delta_0$ is its value before the first update round.
\begin{claim} \label{claim:easy_bound_m}
    We have that
    \begin{equation*}
        m \geq c_1 \, \frac{7}{16} \cdot \frac{\log n}{(1-2\delta)^2}.
    \end{equation*}
\end{claim}
\begin{proof}
    The minimum value attained by the function~$x \rightarrow x + (1-2x)^2$ on the interval~$[0,1/2]$ is equal to~$7/16$. By \Cref{eq:assumption_m_SF},
    \[
        m\geq c_1\pa{ \frac{\delta \log n}{(1-2\delta)^2}+ \log n} = c_1 \frac{(\delta + (1-2\delta)^2) \log n}{(1-2\delta)^2} \geq c_1 \, \frac{7}{16} \cdot \frac{\log n}{(1-2\delta)^2}.
    \]
    %which concludes the proof of \Cref{claim:easy_bound_m}.
\end{proof}

We now prove \Cref{claim:duration_boosting}.
\begin{proof}
    By the definition of Algorithm \SF, the Majority Boosting phase is made up of~$L$ short sub-phases which last for~$\lceil w/h \rceil$ rounds each, where~$w := 100/(1-2\delta)^2$, and a longer sub-phase which lasts for~$\lceil m/h\rceil$ rounds. To prove the statement, we therefore need to verify that~$L \cdot \lceil w/h \rceil \leq \lceil m/h\rceil$.

    If~$h \geq w$, then $L \cdot \lceil w/h \rceil = L = 10 \, \log n$, which is less than~$m/h$ by \Cref{eq:assumption_m_SF}.
    Now, consider the case that~$h < w$. Then,
    \begin{equation} \label{eq:ceiling_ineq}
        \left\lceil \frac{w}{h} \right\rceil < \frac{w}{h}+1 < \frac{2w}{h}.
    \end{equation}
    Moreover,
    \begin{align*}
        &m \geq c_1 \, \frac{7}{16} \cdot \frac{\log n}{(1-2\delta)^2} & \text{(by \Cref{claim:easy_bound_m})} \\
        &\implies m \geq \frac{2000}{(1-2\delta)^2} \log n & \text{(as long as $c_1 \geq \frac{16}{7} \cdot 2000$)} \\
        &\implies \frac{m}{h} \geq \frac{2w}{h}\cdot 10 \, \log n.
    \end{align*}
    By \Cref{eq:ceiling_ineq}, this implies that~$\lceil m/h\rceil \geq \lceil w/h \rceil \cdot 10 \, \log n$, which concludes the proof of \Cref{claim:duration_boosting}.
    %Therefore, our goal is to prove that $\frac{2w}{h}\cdot 10 \, \log n\leq \frac{m}{h}\iff \frac{2000}{(1-2\delta)^2} \log n \leq m$.
    %By \Cref{claim:easy_bound_m}, $m \geq c_1 \, \frac{7}{16} \cdot \frac{\log n}{(1-2\delta)^2}$.
    %\Cref{claim:duration_boosting} follows if $c_1 \geq \frac{16}{7} \cdot 2000$.
\end{proof}

We now prove \Cref{lemma:initialization_boosting}.
\begin{proof}
    Recall that the opinion of each agent before the first update round is equal to its weak-opinion.
    Therefore, by \Cref{lemma:first_phase_SF} and by the fact the number of agents is $n$, $\bbE(\Delta_0) = 4\sqrt{n \log n}$. The result follows from applying the \nameref{thm:additive_chernoff_bound}.
\end{proof}

Next, in \Cref{lemma:boosting_main,lemma:correctness_opinion}, we show that during the Majority Boosting phase, $\Delta_\ell$ grows exponentially for rounds $\ell=1,\ldots,L$, reaching a constant fraction of $n$.

We now prove \Cref{lemma:boosting_main}.
\begin{proof}
    Let~$M$ be the multi-set of all messages received by Agent~$i$ between the~$\ell^{\text{th}}$ and~$(\ell+1)^{\text{th}}$ update rounds, and let~$t$ be the round number of the latter. By construction of Algorithm \SF, $|M| \geq \frac{100}{(1-2\delta)^2}$.
    Moreover, the opinion of Agent~$i$ in round~$t$, $Y_{t}^{(i)}$, corresponds to the majority opinion in~$M$.
    Formally, let~$X_k = +1$ (resp. $-1$) if the~$k^{\text{th}}$ message in~$M$ is~$1$ (resp. $0$), and let~$X = \sum_k X_k$.
    Recall that~$C_{t}^{(i)}$ denotes the coin toss performed by Agent~$i$ to break ties in round~$t$.
    We have that
     \begin{equation} \label{eq:weak_opn_comput_SF}
        Y_{t}^{(i)} = 1 \iff \begin{cases}
            X > 0, &\text{or} \\
            X = 0 \text{ and } C_{t}^{(i)} = 1.
        \end{cases}
    \end{equation}
    %Note that for every~$t \in \{t_1,\ldots,t_2-1\}$, agents do not update their opinion, and hence $\Delta_{t} = \Delta_{t_1}$.
    Conditioning on~$\Delta_{\ell} = x$, all~$\{X_k\}$ are mutually independent and identically distributed. Specifically, they follow a Rademacher distribution with parameter~$p$, where
    \begin{equation*}
        p := \pa{\frac{1}{2} + \frac{x}{n}}(1-\delta) + \pa{\frac{1}{2} - \frac{x}{n}}\delta = \frac{1}{2} + (1-2\delta)\cdot \frac{x}{n}.
    \end{equation*}
    Thus, by \Cref{lemma:majority_boosting},
    \begin{align*}
        &\Pr\pa{X>0\mid \Delta_\ell = x} - \Pr\pa{X<0\mid \Delta_\ell = x}\\
        &\geq \sqrt{\frac{2}{\pi e}} \cdot \min\{ \sqrt{|M|} (1-2\delta) \cdot \frac{x}{n}, 1 \} \\
        &\geq \sqrt{\frac{2}{\pi e}} \cdot \min\{ 10 \cdot \frac{x}{n}, 1 \}. & \text{(since~$|M| \geq \tfrac{100}{(1-2\delta)^2}$)}
    \end{align*}
    Using the fact that~$\Pr\pa{ X = 0} = 1 - \Pr\pa{ X >0 } -  \Pr\pa{ X <0}$, we obtain
    \begin{align*}
        \Pr\pa{Y^{(i)}_{t} = 1 \mid \Delta_{\ell} = x}
        &= \Pr\pa{ X >0 \mid \Delta_\ell = x} + \frac{1}{2} \cdot \Pr\pa{ X = 0 \mid \Delta_\ell = x} &\text{(by \Cref{eq:weak_opn_comput_SF})} \\
        & = \frac{1}{2} + \frac{ \Pr\pa{ X >0 \mid \Delta_\ell = x} -  \Pr\pa{ X <0\mid \Delta_\ell = x }}{2} \\
        &\geq  \frac{1}{2} + \sqrt{\frac{1}{2\pi e}} \cdot \min\{ 10\cdot \frac{x}{n}, 1 \} := \frac{1}{2}+A.
    \end{align*}
    Therefore, the expected number of agents with opinion~$1$ in round~$t$ is greater than~$n/2 + A\, n$, and by the \nameref{thm:additive_chernoff_bound},
    \begin{align*}
        &\Pr\pa{\Delta_{\ell+1} \leq \frac{1}{\sqrt{8\pi e}} \cdot \min\{ 10\, x, n \} ~\Big|~ \Delta_\ell = x} \\
        &= \Pr\pa{\Delta_{\ell+1} \leq \frac{A\, n}{2} ~\Big|~ \Delta_\ell = x}  \\
        &\leq \exp \pa{-2 \cdot \frac{A^2 \, n}{4}} \\
        &= \exp \pa{ - \frac{1}{4\pi e} \cdot \min\{ 100 \frac{x^2}{n},n \}} \\
        &\leq \exp \pa{ - \frac{1}{4\pi e} \cdot \min\{ 400 \log n, n \}} & \text{(since~$x \geq 2 \sqrt{n \log n}$)} \\
        &= \exp \pa{ - \frac{400}{4\pi e} \cdot \log n}.
    \end{align*}
    This concludes the proof of \Cref{lemma:boosting_main} by noting that~$\frac{10}{\sqrt{8\pi e}} > 1.2$, and $\frac{400}{4 \pi e}>8$.
\end{proof}

We now prove \Cref{lemma:correctness_opinion}.
\begin{proof}
    Let~$\mathcal{E}_\ell$ be the event
    \begin{equation*}
        \mathcal{E}_\ell = \left\{ \Delta_{\ell+1} \geq \min \pa{ 1.2 \, \Delta_\ell, \frac{n}{\sqrt{8 \pi e}} } \right\}.
    \end{equation*}
    We have
    \begin{align*}
        \Pr\pa{ \bigcap_{\ell = 0}^{L-1} \mathcal{E}_\ell ~\Big|~ \Delta_0 \geq 2 \sqrt{n \log n} } &\geq \prod_{\ell = 0}^{L-1} \pa{1 - \frac{1}{n^8}} & \text{(by \Cref{lemma:boosting_main})} \\
        &\geq 1 - L \cdot \frac{1}{n^8} & \text{(union bound)} \\
        &= 1 - 10 \, \log n \cdot \frac{1}{n^8} \\
        &\geq 1 - \frac{1}{n^7}. & \text{(for~$n$ large enough)}
    \end{align*}
    By \Cref{lemma:initialization_boosting}, $\Delta_0 \geq 2 \sqrt{n \log n }$ w.p. $1-\tfrac{1}{n^8}$, and so~$\cap_{\ell = 0}^{L-1} \mathcal{E}_\ell$ happens w.p. $1 - \tfrac{1}{n^7}$
    This implies that
    \begin{equation*}
        \Delta_L \geq \min \{ 2 \cdot 1.2^L \sqrt{n \log n} , \frac{n}{\sqrt{8\pi e}} \}. % = \frac{n}{\sqrt{8\pi e}},
    \end{equation*}
    Since~$1.2^L = 1.2^{10 \, \log n} = n^{\log 12} > n$, we have $\Delta_L \geq \frac{n}{\sqrt{8\pi e}}$,
    which concludes the proof of \Cref{lemma:correctness_opinion}.
\end{proof}

We now prove \Cref{lemma:w.h.p._agent_correct_SF}.
\begin{proof}
    Let~$T$ be the last round of the execution of Algorithm \SF, which follows the last sub-phase of length~$\lceil m /h\rceil$. Let~$i \in I$, and let~$M$ be the memory of Agent~$i$ in this round, which contains~$m$ messages.
    Let $X_k$ be the $k$-th message stored in $M$, and $X=\sum_k X_k$.
    Conditioning on~$\Delta_L \geq \frac{n}{\sqrt{8\pi e}}$, each message is equal to 1 with probability at least~$1/2+\frac{1}{\sqrt{8\pi e}}(1-2\delta)$. Therefore, for every agent~$i$, the expectation of the number of messages equal to 1 is equal to
    \begin{equation*}
        \mu := \bbE(X) = m \cdot \pa{\frac{1}{2}+\frac{1-2\delta}{\sqrt{8\pi e}}}.
    \end{equation*}
    Note that
    \begin{equation} \label{eq:tmp}
        m \cdot \frac{(1-2\delta)^2}{8\pi e} \underset{\text{\Cref{claim:easy_bound_m}}}{\geq} c_1 \, \frac{7}{16} \cdot \frac{\log n}{(1-2\delta)^2} \cdot \frac{(1-2\delta)^2}{8\pi e} = c_1 \cdot \frac{7}{16} \cdot \frac{\log n}{8\pi e}.
    \end{equation}
    Moreover,
    \begin{align*}
        &\Pr\pa{ X \leq \frac{m}{2} ~\Big|~ \Delta_L \geq \frac{n}{\sqrt{8\pi e}} } \\
        &= \Pr\pa{ X \leq \mu + \frac{m}{2} -\mu ~\Big|~ \Delta_L \geq \frac{n}{\sqrt{8\pi e}} }\\
        &\leq \Pr\pa{ X \leq \mu - m \frac{1-2\delta}{\sqrt{8 \pi e}} ~\Big|~ \Delta_L \geq \frac{n}{\sqrt{8\pi e}} }\\
        &\leq \exp\pa{- 2 m \pa{\frac{1-2\delta}{\sqrt{8 \pi e}}}^2} & \text{(by \nameref{thm:additive_chernoff_bound})} \\
        &\leq \exp\pa{- c_1 \cdot \frac{7}{8} \cdot \frac{\log n}{8\pi e} } & \text{(by \Cref{eq:tmp})} \\
        &\leq \frac{1}{n^{3}},
    \end{align*}
    where we used that $c_1$ is big enough.
    We obtain
    \begin{align*}
        \Pr\pa{ X > \frac{m}{2}}
        &\geq  \Pr\pa{ \Delta_L \geq \frac{n}{\sqrt{8\pi e}} } \cdot \Pr\pa{ X \leq \frac{m}{2} ~\Big|~ \Delta_L \geq \frac{n}{\sqrt{8\pi e}} }\\
        &\geq \pa{1-\frac{1}{n^7}} \cdot  \pa{1-\frac{1}{n^{3}}} \geq 1-\frac{1}{n^3}.& \text{(by \Cref{lemma:correctness_opinion})}
    \end{align*}
    Since $\{X> m/2\}\implies Y_T^{(i)}=1$, this implies that~$\Pr\pa{ Y_T^{(i)}=1} \geq 1-\frac{1}{n^3}$, i.e., the opinion of Agent~$i$ is correct w.h.p.
    We conclude the proof of \Cref{lemma:w.h.p._agent_correct_SF} by a union bound over all agents.
    %By \Cref{lemma:correctness_opinion}, $\Delta_L \geq \frac{n}{\sqrt{8\pi e}}$ w.h.p.
\end{proof}
%We can end the proof of correctness w.h.p. of Algorithm \SF applying a union bound argument.

% Write the next section (ROBIN)
% Comments of Amos in SSF
% explain the value of m in the intuition part (NIC)

\subsubsection{\texorpdfstring{Proof of \Cref{thm:upper}}{Proof of Theorem 4}} \label{sec:final_proof_SF}

Let~$N$ be an arbitrary~$\delta$-upper bounded noise matrix of size~$d=2$, and let
\begin{equation*}
    m \geq c_1 \pa{6 \cdot \frac{n \delta \log n}{\min\{s^2,n\}  (1-|\Sigma|\delta)^2} + \frac{\sqrt{n}\log n}{s} + \frac{(s_0+s_1)\log n}{s^2} + h \log n},
\end{equation*}
where~$c_1$ is the constant from \Cref{eq:assumption_m_SF}.
By \Cref{thm:reduction}, there exists a matrix~$P$ s.t. the simulation of Algorithm \SF with artificial noise~$P$ under~$N$ is equivalent to Algorithm \SF under a $\delta'$-uniform noise, where
\begin{equation*}
    \delta' = \pa{2 + \frac{1}{2} \cdot \frac{(1-2 \, \delta)}{\delta}}^{-1}.
\end{equation*}
Writing~$A = {\delta'}^{-1}$, we have
\begin{equation*}
    \frac{(1-2\delta')^2}{\delta'} = A \cdot \pa{1-\frac{2}{A}}^2 = \frac{(A-2)^2}{A} = \frac{ \left[ \tfrac{1}{2}\tfrac{1-2\delta}{\delta} \right]^2}{2 + \tfrac{1}{2}\tfrac{1-2\delta}{\delta}} = \frac{1}{4 \delta} \cdot \frac{(1-2\delta)^2}{2\delta + \tfrac{1}{2}(1-2\delta)} \geq \frac{1}{6} \cdot \frac{(1-2\delta)^2}{\delta},
\end{equation*}
where in the last inequality we have used the fact that~$\delta \in [0,1/2)$ and so~$2\delta + \tfrac{1}{2}(1-2\delta) \leq 1$.
Hence,
\begin{equation*}
    \frac{n \delta' \log n}{s^2 (1-2\delta')^2} \leq 6 \cdot \frac{n \delta \log n}{s^2 (1-2\delta)^2}.
\end{equation*}
Therefore, $m$ satisfies \Cref{eq:assumption_m_SF} with~$\delta'$, and \Cref{thm:upper} follows from our analysis under uniform noise (see \Cref{lemma:w.h.p._agent_correct_SF}).

\subsection{\texorpdfstring{Analysis of Algorithm \SSF}{Analysis of Algorithm SSF}} \label{sec:SSF_analysis}

In this section, we prove \Cref{thm:self}. Accordingly, we work under its source-count assumption and assume that the parameter~$m'$ used by the protocol satisfies \Cref{eq:assumption_m_SSF}. Recall that we assume the noise matrix~$N$ to be $\delta$-uniform. This assumption will be relaxed later in \Cref{sec:final_proof_SSF} when concluding the proof of \Cref{thm:self}, by using the results of \Cref{sec:non-uniform}.

Agents and rounds are indexed solely for analysis purposes; agents themselves do not require identifiers or knowledge of the current time.
In addition to the notations already introduced in \Cref{sec:prelim}, we will write~$\R_{t_1:t_2}^{(i)}$ to denote all the randomness (noise, samples, and coin tosses) linked to Agent~$i$, between rounds~$t_1$ and~$t_2$. Formally,
\begin{equation*}
    \R_{t_1:t_2}^{(i)} := \left\{ (N_t^{(i)},S_t^{(i)},C_t^{(i)}), \quad \text{for every } t \in \{ t_1+1,\ldots,t_2 \} \right\}.
\end{equation*}

\subsubsection{\texorpdfstring{Correctness of weak-opinions in Algorithm \SSF}{Correctness of weak-opinions in Algorithm SSF}}

In this section, we prove that at any time in the execution after at least~$2\lceil m'/h \rceil$ rounds, the weak-opinion of any non-source agent is correct with a probability slightly larger than~$1/2$. Moreover, we show that all weak-opinions are mutually independent.

We now prove \Cref{lemma:correctness_weak}.
\begin{proof}
We prove each statement in \Cref{lemma:correctness_weak} separately.

\paragraph*{Proof of~(i)} Fix a round~$T \geq 2 \lceil m'/h \rceil$, let $t_2$ be the last update round of Agent~$i$ before round~$T$, and let~$t_1 = t_2 - \lceil m'/h \rceil$.
    Since there are $\lceil m'/h \rceil$ rounds between two consecutive updates, $t_2 \geq T - \lceil m'/h \rceil$, and hence~$t_1 \geq 0$. Moreover, the fact that the memory of Agent~$i$ has been emptied in round~$t_1$ implies that in round~$t_2$ (when~$\Y_T^{(i)}$ is computed), it does not contain any message determined by the arbitrary initial configuration.
    Consider the multi-set~$\big\{\sigma_k^{(i)}, k \in \{1,\ldots,m'\} \big\}$ of all~$m'$ messages received by Agent~$i$ between round~$t_1$ and~$t_2$.
    Define
    \begin{equation*}
        X_k^{(i)} := \begin{cases}
         1 & \text{if } \sigma_k^{(i)} = (1,1), \\
         -1 & \text{if } \sigma_k^{(i)} = (1,0),  \\
         0 & \text{otherwise. }
         \end{cases} \\
    \end{equation*}
    %By definition of \Cref{alg:SSF}, the value of~$X_k^{(i)}$ only depends on the sample $S_t^{(i)}$ and the noise~$\{N_{t,\ell}^{(i)}\}_\ell$.
    Recall that formally, each~$N_{t,\ell}^{(i)}$ is a random function from~$\Sigma$ to~$\Sigma$, with the convention that the~$\ell^{\text{th}}$ message sampled by Agent~$i$ in round~$t$ is modified by noise according to~$\sigma \rightarrow N_{t,\ell}^{(i)}(\sigma)$.
    We can assume, without loss of generality, that~$N_{t,\ell}^{(i)}$ is distributed explicitly as follows:
	\begin{equation} \label{eq:noise_def}
		N_{t,\ell}^{(i)} := \begin{dcases} \func{ (0,0) \rightarrow (0,0) \\ (0,1) \rightarrow (0,1) \\ (1,1) \rightarrow (1,1) \\ (1,0) \rightarrow (1,0)} \qquad \text{\em (uncorrupted message)} & \text{w.p. } 1-3\delta, \\
        \func{ (0,0) \rightarrow (0,1) \\ (0,1) \rightarrow (0,0) \\ (1,1) \rightarrow (0,0) \\ (1,0) \rightarrow (0,0) } \text{ or } \func{ (0,0) \rightarrow (1,0) \\ (0,1) \rightarrow (1,0) \\ (1,1) \rightarrow (0,1) \\ (1,0) \rightarrow (0,1) } \text{ or } \func{ (0,0) \rightarrow (1,1) \\ (0,1) \rightarrow (1,1) \\ (1,1) \rightarrow (1,0) \\ (1,0) \rightarrow (1,1) } & \text{w.p. } \delta.
		 \end{dcases}
	\end{equation}
	Indeed, with this definition, one can check that the noise is~$\delta$-uniform: for every~$\sigma,\sigma' \in \Sigma$, we have
	\begin{equation*}
		\Pr\pa{ N_{t,\ell}^{(i)}(\sigma) = \sigma'} = \begin{cases} (1-3\delta) & \text{if } \sigma = \sigma', \\ \delta & \text{otherwise.} \end{cases}
	\end{equation*}
    It is clear from \Cref{eq:noise_def} that whenever the original message has a first bit equal to~$0$, and this bit is flipped to~$1$ under the effect of noise, then the second bit of the modified message is independent of the second bit of the original message.
    Consequently, each~$k \in \{1,\ldots,m'\}$ corresponds to a different pair~$(t,\ell) \in \{t_1+1,\ldots,t_2\} \times \{1,\ldots,h\}$, such that~$X_k^{(i)}$ is a deterministic function of~$S_{t,\ell}^{(i)}$ and~$N_{t,\ell}^{(i)}$. Specifically, $S_{t,\ell}^{(i)}$ determines whether the sampled agent is a source or not, while~$N_{t,\ell}^{(i)}$ determines how the message is affected by noise, and no other information is needed to compute the value of~$X_k^{(i)}$.
    Note that this argument relies heavily on the fact that the noise matrix~$N$ is $\delta$-uniform.
    Now, let
    \begin{equation*}
        X^{(i)} = \sum_{k=1}^{m'} X_k^{(i)}.
    \end{equation*}
    Informally, $X^{(i)}$ measures the bias towards the correct preference, corresponding to the messages observed by Agent~$i$.
    Recall that~$C_{t_2}^{(i)}$ represents the coin used by Agent~$i$ in round~$t_2$ to break ties when computing~$\Y_T^{(i)}$.
    By definition of Algorithm \SSF,
    \begin{equation} \label{eq:weak_opn_equiv}
        \Y_T^{(i)} = 1 \iff \begin{cases}
            X^{(i)} > 0, &\text{or} \\
            X^{(i)} = 0 \text{ and } C_{t_2}^{(i)} = 1. %, &\text{or} \\
            %i=1,\dots, s_1. & \text{(Agent~$i$ is a source supporting~$1$)}
        \end{cases}
    \end{equation}
    Together with the reasoning above, this implies that $\Y_T^{(i)}$ is a deterministic function of~$\mathcal{R}_{t_1:t_2}^{(i)}$.
    Since all~$\big\{ \mathcal{R}_{t_1:t_2}^{(i)} \big\}_{i \in I}$ are mutually independent,
	it follows that all $\{\Y_T^{(i)}\}_{i \in I}$ are also mutually independent, which establishes \Cref{lemma:correctness_weak}.(i).

    \paragraph*{Proof of~(ii)}
    There are two ways in which~$X_k^{(i)}$ can be equal to~$1$ (resp.~$-1$): either Agent~$i$ samples a source with preference~$1$ (resp.~$0$) and the noise does not modify the message, or they sample any other agent and the noise changes the message into~$(1,1)$ (resp.~$(1,0)$). Therefore,
    \begin{equation} \label{eq:proba_Xk}
        \begin{split}
            \Pr\pa{ X_k^{(i)} = 1 } &= \frac{s_1}{n}(1-3\delta) + \pa{1-\frac{s_1}{n}}\delta, \quad \text{and} \\
            \Pr\pa{ X_k^{(i)} = -1 } &= \frac{s_0}{n}(1-3\delta) + \pa{1-\frac{s_0}{n}}\delta.
        \end{split}
    \end{equation}

    \begin{claim} \label{claim:p_SSF}
        %Recall that~$s := |s_1 - s_0|$ denotes the difference between the number of sources supporting opinions~$0$ and~$1$.
        For every~$k \in \{1,\ldots,m'\}$, we have
        \begin{equation} \label{eq:first_claim_OLD_ssf}
             \Pr\pa{ X_k^{(i)}\neq 0 } \geq (1-4\delta)^2 \cdot \frac{s_0 + s_1}{2n} + \delta
        \end{equation}
        Moreover, if $\delta \geq \frac{s_0 + s_1}{2n}(1-4\delta)$,
         \begin{equation} \label{eq:second_claim_OLD_ssf}
            \Pr\pa{ X_k^{(i)} = 1 \mid X_k^{(i)} \neq 0} \geq \frac{1}{2} + \frac{s}{n}\frac{(1-4\delta)}{8\,\delta},
        \end{equation}
        while if $\delta < \frac{s_0 + s_1}{2n}(1-4\delta)$,
        \begin{equation} \label{eq:third_claim_OLD_ssf}
            \Pr\pa{ X_k^{(i)} = 1 \mid X_k^{(i)} \neq 0} \geq \frac{1}{2} + \frac{s}{4 (s_0+s_1)}.
        \end{equation}
    \end{claim}
    \begin{proof}[Proof of \Cref{claim:p_SSF}]
        By \Cref{eq:proba_Xk}, we have
        \begin{align}
            \Pr\pa{ X_k^{(i)}\neq 0 }  \nonumber
            &= \frac{s_1}{n}(1-3\delta) + \pa{1-\frac{s_1}{n}}\delta + \frac{s_0}{n}(1-3\delta) + \pa{1-\frac{s_0}{n}}\delta \\ \label{eq:Xneq0_ssf}
            &= 2 \delta + (1-4\delta) \pa{ \frac{s_0+s_1}{n} }
        \end{align}
        \Cref{eq:first_claim_OLD_ssf} follows from \Cref{eq:Xneq0_ssf} and the fact that~$0 < 1-4\delta \leq 1$.
        Moreover, we have
        \begin{align} \nonumber
            &\Pr\pa{ X_k^{(i)} = 1 } - \Pr\pa{ X_k^{(i)} = -1}\\ \nonumber
            &= \frac{s_1}{n}(1-3\delta) + \pa{1-\frac{s_1}{n}}\delta - \frac{s_0}{n}(1-3\delta) - \pa{1-\frac{s_0}{n}}\delta\\ \label{eq:X=1-X=-1_ssf}
            &= \frac{s_1-s_0}{n} (1-4\delta).
        \end{align}
        Then, if $\delta \geq \frac{s_0 + s_1}{2n}(1-4\delta)$, we have
        \begin{align*}
            &\Pr\pa{ X_k^{(i)} = 1 \mid X_k^{(i)} \neq 0} - \Pr\pa{ X_k^{(i)} = -1 \mid X_k^{(i)} \neq 0}\\
            &= \frac{\Pr\pa{ X_k^{(i)} = 1} - \Pr\pa{ X_k^{(i)} = -1}}{  \Pr\pa{ X_k^{(i)}\neq 0 } }\\
            &= \frac{ \frac{s_1-s_0}{n} (1-4\delta) }{ 2 \delta + (1-4\delta) \pa{ \frac{s_0+s_1}{n} } } &\pa{ \text{by \Cref{eq:Xneq0_ssf,eq:X=1-X=-1_ssf}} }\\
            &\geq \frac{ \frac{s_1-s_0}{n} (1-4\delta) }{ 4 \delta } = \frac{s}{n} \cdot \frac{1-4\delta}{4\delta}.
        \end{align*}
        Instead, if $\delta < \frac{s_0 + s_1}{2n}(1-4\delta)$, we have
        \begin{align*}
            &\Pr\pa{ X_k^{(i)} = 1 \mid X_k^{(i)} \neq 0} - \Pr\pa{ X_k^{(i)} = -1 \mid X_k^{(i)} \neq 0}\\
            &= \frac{ \frac{s_1-s_0}{n} (1-4\delta) }{ 2 \delta + (1-4\delta) \pa{ \frac{s_0+s_1}{n} } } \\
            &\geq \frac{s_1 - s_0}{2(s_0+s_1)} = \frac{s}{2(s_0 + s_1)}.
        \end{align*}
        In both cases, since $\Pr\pa{ X_k^{(i)} = -1 \mid X_k^{(i)}\neq 0} = 1 - \Pr\pa{ X_k^{(i)} = 1 \mid X_k^{(i)} \neq 0}$,
        \begin{align*}
            &\Pr\pa{ X_k^{(i)} = 1 \mid X_k^{(i)} \neq 0} - \Pr\pa{ X_k^{(i)} = -1 \mid X_k^{(i)} \neq 0} \geq A \\
            &\implies 2 \, \Pr\pa{ X_k^{(i)} = 1 \mid X_k^{(i)} \neq 0} - 1 \geq A \\
            &\implies \Pr\pa{ X_k^{(i)} = 1 \mid X_k^{(i)} \neq 0} \geq \frac{1}{2} + \frac{A}{2},
        \end{align*}
        and \Cref{eq:second_claim_OLD_ssf,eq:third_claim_OLD_ssf} follow when replacing~$A$ by~$\frac{s}{n} \cdot \frac{1-4\delta}{4\delta}$ and~$\frac{s}{2(s_0 + s_1)}$, respectively.
    \end{proof}
    Because of \Cref{claim:p_SSF}, the proof of \Cref{lemma:first_phase} applies to the variables above. Choosing~$c_\varepsilon$ in \Cref{eq:assumption_m_SSF} sufficiently large, the same argument gives the stronger bias needed for the non-source concentration below. More precisely, and using the fact that $\Pr(X^{(i)} = 0) = 1 - \Pr( X^{(i)} >0 ) -  \Pr( X^{(i)} <0 )$,
    \begin{align*}
        \Pr\pa{\Y^{(i)}_T = 1}
        &= \Pr\pa{ X^{(i)} >0 } + \frac{1}{2} \cdot \Pr\pa{ X^{(i)} = 0 } &\text{(by \Cref{eq:weak_opn_equiv})} \\
        & = \frac{1}{2} + \frac{ \Pr\pa{ X^{(i)} >0 } -  \Pr\pa{ X^{(i)} <0 }}{2} \\
        &\geq  \frac{1}{2} + \frac{8}{\varepsilon} \sqrt{\frac{\log n}{n}},
    \end{align*}
    where the last inequality follows from \Cref{claim:p_SSF}, \Cref{eq:assumption_m_SSF} and the proof of \Cref{lemma:first_phase}.
    This concludes the proof of \Cref{lemma:correctness_weak}.
\end{proof}

\subsubsection{\texorpdfstring{Correctness of opinions in Algorithm \SSF}{Correctness of opinions in Algorithm SSF}}

In this section, we show that the small advantage established in \Cref{lemma:correctness_weak} implies that all agents' opinion is the correct preference with high probability. First, we prove that in every round, the number of non-source agents with a correct weak-opinion is sufficiently large with high probability.

We now prove \Cref{lemma:weak_opinion_strong_correctness}.
\begin{proof}
    Let~$t_1 \geq 2 \lceil m'/h \rceil$, and $t_2 = t_1+\lceil m'/h \rceil$.
    By construction of Algorithm \SSF, each non-source agent updates its weak-opinion exactly once between rounds~$t_1+1$ and~$t_2$ (included).
    For every~$j \in I_{\mathrm{ns}}$, let~$Z_j$ (resp. $Z_j'$) be the weak-opinion of Agent~$j$ before (resp. after) its update.

    Without loss of generality, enumerate the non-source agents as~$j_1,\ldots,j_{n_{\mathrm{ns}}}$ according to the order of their updates.
    With this assumption, for every~$s \in \{t_1+1,\ldots,t_2\}$, there exists~$k \in \{0,\ldots,n_{\mathrm{ns}}\}$ such that
    \begin{equation} \label{eq:Z_k-def}
        \sum_{j \in I_{\mathrm{ns}}} \tilde{Y}_s^{(j)}  = \sum_{r=1}^k Z_{j_r}' + \sum_{r=k+1}^{n_{\mathrm{ns}}} Z_{j_r} =: \mathcal{Z}_k.
    \end{equation}
    By \Cref{lemma:correctness_weak}.(ii) and $n_{\mathrm{ns}}\geq \varepsilon n$, for every~$k \in \{0,\ldots,n_{\mathrm{ns}}\}$,
    \begin{equation*}
        \bbE \pa{ \mathcal{Z}_k } \geq \frac{n_{\mathrm{ns}}}{2} + \frac{8n_{\mathrm{ns}}}{\varepsilon}\sqrt{\frac{\log n}{n}} \geq \frac{n_{\mathrm{ns}}}{2} + 8 \sqrt{n \log n}.
    \end{equation*}
    By \Cref{lemma:correctness_weak}.(i), the weak opinions of non-source agents are mutually independent, so we can apply \nameref{thm:additive_chernoff_bound} to obtain
    \begin{equation} \label{eq:Z_k-hoeffding_bound}
        \Pr \pa{ \mathcal{Z}_k \leq \frac{n_{\mathrm{ns}}}{2} + 2 \sqrt{n\log n }} \leq
        \exp\pa{ - 72 \log n }.
    \end{equation}
    Finally, we have
    \begin{align*}
        &\Pr\pa{ \bigcup_{s = t_1+1}^{t_2} \left\{ \sum_{j \in I_{\mathrm{ns}}} \tilde{Y}_s^{(j)} \leq \frac{n_{\mathrm{ns}}}{2} + 2 \sqrt{n \log n} \right\} } \\
        &= \Pr\pa{ \bigcup_{k=0}^{n_{\mathrm{ns}}} \left\{ \mathcal{Z}_k \leq \frac{n_{\mathrm{ns}}}{2} + 2 \sqrt{n\log n} \right\} } &\text{(by \Cref{eq:Z_k-def})} \\
        &\leq \sum_{k=0}^{n_{\mathrm{ns}}} \Pr \pa{\mathcal{Z}_k \leq \frac{n_{\mathrm{ns}}}{2} + 2 \sqrt{n\log n}} & \text{(union bound)} \\
        &\leq n \, \exp\pa{ - 72 \log n }. &\text{(by \Cref{eq:Z_k-hoeffding_bound})}
    \end{align*}
    This concludes the proof of \Cref{lemma:weak_opinion_strong_correctness}.
\end{proof}

Next, we demonstrate that combining the previous result with the majority operation ensures, with very high probability, the correctness of all opinions.

We now prove \Cref{lemma:ssf_second_phase}.
\begin{proof}
    Let~$t \geq 4 \lceil m'/h \rceil$ and $i \in I$.
    Let~$t_2 < t$ be the index of the last update round of Agent~$i$ before round~$t$, and let~$t_1 = t_2 - \lceil m'/h \rceil$.
    By definition, $Y_t^{(i)}$, was computed based on the second bit of all messages observed by Agent~$i$ between rounds~$t_1+1$ and~$t_2$.
    Since sources always send their fixed preference and non-sources send their weak opinion, $Y_t^{(i)}$ is a deterministic function of $\R_{t_1:t_2}^{(i)}$, the fixed source preferences, and the set of non-source weak opinions between rounds~$t_1$ and~$t_2$, namely,
    \begin{equation*}
        \left\{ \tilde{Y}_s^{(j)}, \quad \text{for } j \in I_{\mathrm{ns}}, s \in \{t_1+1,\ldots,t_2\} \right\}.
    \end{equation*}
    Let~$A$ be the following event:
    \begin{equation*}
        A := \bigcap_{s = t_1+1}^{t_2} \left\{ \sum_{j \in I_{\mathrm{ns}}} \tilde{Y}_s^{(j)} > \frac{n_{\mathrm{ns}}}{2} + 2 \sqrt{n\log n} \right\},
    \end{equation*}
    which is the high-probability event from \Cref{lemma:weak_opinion_strong_correctness}.
    Note that:
    \begin{itemize}
        \item[(i)] Let~$s \in \{t_1+1,\ldots,t_2\}$ and~$t_0 = t_1 - \lceil m'/h \rceil$.
        By \Cref{lemma:correctness_weak}.(i), $\tilde{Y}_{s}^{(i)}$ is a deterministic function of~$\R_{t_0:t_1}^{(i)}$.
        Therefore, since~$\{t_0+1,\ldots,t_1\} \cap \{t_1+1,\ldots,t_2\} = \emptyset$, $\R_{t_1:t_2}^{(i)}$ is independent from $\tilde{Y}_s^{(i)}$.

        \item[(ii)] Let~$j \in I_{\mathrm{ns}}$, $j \neq i$, and~$s \in \{t_1+1,\ldots,t_2\}$. By \Cref{lemma:correctness_weak}.(i), there are two rounds~$t_1',t_2'$ s.t. $\tilde{Y}_{s}^{(j)}$ is a deterministic function of~$\R_{t_1':t_2'}^{(j)}$. Therefore, since~$i \neq j$, $\R_{t_1:t_2}^{(i)}$ is independent from $\tilde{Y}_s^{(j)}$.
    \end{itemize}
    Together, (i) and (ii) imply that~$\R_{t_1:t_2}^{(i)}$ is independent from~$A$.
    Conditioning on~$A$, before noise at least~$s_1 + n_{\mathrm{ns}}/2 + 2\sqrt{n\log n}$ agents display second bit~$1$. Since~$s_1>s_0$ and~$n=n_{\mathrm{ns}}+s_0+s_1$, this is larger than~$n/2+2\sqrt{n\log n}$; the source bias can only help. Therefore, the probability that a given message received by Agent~$i$ between rounds~$t_1$ and~$t_2$ has a second bit equal to~$1$, is at least (after noise is applied)
    \begin{equation*}
        p := (1-2\delta) \pa{\frac{1}{2} + 2 \sqrt{\frac{\log n}{n}}} + 2\delta \pa{\frac{1}{2} - 2 \sqrt{\frac{\log n}{n}}} = \frac{1}{2} + 2 \sqrt{\frac{\log n}{n}} (1-4\delta).
    \end{equation*}
    Denoting by~$B_{m',p}$ a binomial random variable with parameters~$m'$ and~$p$, this implies that
    \begin{align*}
        &\Pr \pa{Y_t^{(i)} = 1 \mid A} \geq \Pr \pa{ B_{m',p} > \frac{m'}{2}} \\
        &=  \Pr \pa{ B_{m',p} > \bbE(B_{m',p}) - m' \cdot 2 \sqrt{\frac{\log n}{n}} \,(1-4\delta)} \\
        &\geq 1-\exp\pa{-\frac{2}{m'} \pa{m' \cdot 2 \sqrt{\frac{\log n}{n}} \, (1-4\delta)}^2} & \text{(by \nameref{thm:additive_chernoff_bound})}\\
        &= 1-\exp\pa{-8m'\frac{\log n}{n}(1-4\delta)^2} \\
        &\geq 1-\exp\pa{-8 \, \log n }. & \text{(by \Cref{eq:assumption_m_SSF}, for $n$ large enough)}
        %&\geq 1-\exp\pa{-2 \, c_4 \log n}, % & \text{(since $\delta < \tfrac{1}{4}$)}
    \end{align*}
    Finally, since updates occur every~$\lceil m'/h \rceil$ rounds, $t_2 \geq t-\lceil m'/h \rceil$, and therefore~$t_1 \geq 2 \lceil m'/h \rceil$. Thus, by \Cref{lemma:weak_opinion_strong_correctness},
    \begin{equation*}
        \Pr \pa{Y_t^{(i)} = 1} \geq \Pr(A) \cdot \Pr \pa{Y_t^{(i)} = 1 \mid A} \geq \pa{1-\frac{1}{n^7}} \cdot \pa{1-\frac{1}{n^8}} \geq 1-\frac{1}{n^6},
    \end{equation*}
    for $n$ large enough, which concludes the proof of \Cref{lemma:ssf_second_phase}.
\end{proof}

Finally, we can check that a consensus on the correct preference happens for polynomially many rounds with high probability.
We now prove \Cref{lemma:end_of_SSF}.
\begin{proof}
    The statement follows directly from \Cref{lemma:ssf_second_phase} and a union bound over all agents and all rounds between~$t_0$ and~$t_0+P(n)$, given that for~$n$ large enough,
    \begin{equation*}
        1 - n \cdot P(n) \cdot \frac{1}{n^6} \geq 1 - \frac{1}{n^2}.
    \end{equation*}
\end{proof}

\subsubsection{\texorpdfstring{Proof of \Cref{thm:self}}{Proof of Theorem 5}} \label{sec:final_proof_SSF}

%\begin{proof}[Proof of \Cref{thm:self}]

Let~$N$ be an arbitrary~$\delta$-upper bounded noise matrix of size~$d=4$, and let
\begin{equation*}
    m' = 2916 \, c_\varepsilon \cdot \pa{\frac {\delta n\log n}{ (1-4\delta)^2}+n},
\end{equation*}
where~$c_\varepsilon$ is the constant in \Cref{eq:assumption_m_SSF}.
Note that no attempt was made to optimize the additional constant factor of~$2916$; our focus is primarily on the asymptotic behavior of the dynamics.
By \Cref{thm:reduction}, there exists a matrix~$P$ s.t. the simulation of Algorithm \SSF with artificial noise~$P$ under~$N$ is equivalent to Algorithm \SSF under a $\delta'$-uniform noise, where
\begin{equation*}
    \delta' = \pa{4 + \frac{1}{54} \cdot \frac{1 - 4 \, \delta}{\delta}}^{-1}.
\end{equation*}
Writing~$A = {\delta'}^{-1}$, we have
\begin{equation*}
    \frac{(1-4\delta')^2}{\delta'} = A \cdot \pa{1-\frac{4}{A}}^2 = \frac{(A-4)^2}{A} = \frac{ \left[ \tfrac{1}{54}\tfrac{1-4\delta}{\delta} \right]^2}{4 + \tfrac{1}{54}\tfrac{1-4\delta}{\delta}} = \frac{1}{54^2 \, \delta} \cdot \frac{(1-4\delta)^2}{4\delta + \tfrac{1}{54} (1-4\delta)}.
\end{equation*}
The fact that~$\delta \in [0,1/4)$ implies that~$4\delta + \tfrac{1}{54}(1-4\delta) \leq 1$, and plugging it in the last equation gives
\begin{equation*}
    \frac{(1-4\delta')^2}{\delta'} \geq \frac{1}{2916} \cdot \frac{(1-4\delta)^2}{\delta}.
\end{equation*}
Hence,
\begin{equation*}
    m' \geq c_\varepsilon \cdot \pa{\frac {\delta' n\log n}{ (1-4\delta')^2}+n}.
\end{equation*}
Therefore, $m'$ satisfies \Cref{eq:assumption_m_SSF} with~$\delta'$, and \Cref{thm:self} follows from our analysis under uniform noise (see \Cref{lemma:end_of_SSF}).

\section{\texorpdfstring{Reducing the size of the alphabet $\Sigma$}{Reducing the size of the alphabet}}
\label{sec:reduction alphabet}
     In this section we show how a protocol designed for a communication over a specific alphabet $\Sigma'$ can be adapted to work with an arbitrary alphabet $\Sigma$ (with fixed size), without changing the asymptotic of the convergence time of that protocol. We only require that $|\Sigma| \geq |\Sigma'| \geq 2$. The idea is quite simple. Assume for example that we have a protocol designed to work with a binary alphabet $\Sigma'=\{0,1\}$. Then, if the given arbitrary alphabet $\Sigma$ has even cardinality, we divide $\Sigma$ into 2 groups of equal size, one associated to ``0'' and the other to ``1''. Then, whenever an agent receives $s\in \Sigma$ which belongs to the ``0'' group (resp., ``1'' group), it runs the protocol as if it had received a 0-message (resp., 1-message). In the case $|\Sigma|$ is odd, the first $|\Sigma|-1$ symbols are treated as in the even case, and each time the last symbol in $|\Sigma|$ is received, it is treated as either ``0'' or as a ``1'', each with probability $1/2$.

More precisely, let $\Sigma'$ be the alphabet of a given protocol. We are mainly interested in the case $|\Sigma'|=2$ (in \SF and \USF) and $|\Sigma'|=4$ (in \SSF).
Let $\Sigma$ be an arbitrary alphabet and $N$ be the corresponding arbitrary $\delta$-bounded noise matrix. Assume $ |\Sigma| \geq |\Sigma'| $, and $|\Sigma| = q |\Sigma'| + r$, for some $q,r \geq 0 \in \mathbf{Z}$, with $r < |\Sigma'|$. W.l.o.g. assume $\Sigma = \{0,1,\ldots, |\Sigma|-1\}$ and $\Sigma' = \{0,1,\ldots, |\Sigma'|-1\}$.
We define a filter $F$ that probabilistically maps $\Sigma$ to $\Sigma'$, so that the first $q|\Sigma'|$ elements in $\Sigma$ are deterministically distributed in $|\Sigma'|$ groups of size $q$, each group being associated to a different symbol in $\Sigma'$, while the remaining $r$ elements are distributed uniformly at random among all the elements of $\Sigma'$.
Formally, $F: \Sigma \mapsto \Sigma'$ s.t., if $i< q|\Sigma'|$, then $F(i) = \lfloor i/q \rfloor$  w.p. 1, while, if $i \geq q|\Sigma'|$, $F(i) = j$ w.p. $1/|\Sigma'|$ for any $j\in \Sigma'$. It is important to note that we assume that each instance of $F(i)$ is independent of the others.

Whenever an agent receives a message $i \in \Sigma$, it executes the protocol (designed for $\Sigma'$) as if it would have received $F(i)\in \Sigma'$.
\begin{lemma}
    \label{lemma:reducing_alphabet}
    Let $P$ be the noise matrix associated to the communication with alphabet $\Sigma'$ after the filter $F$ has been applied to the communication with alphabet $\Sigma$ with $\delta$-bounded noise matrix $N$. We have that $P$ is $\delta'$-bounded where $\delta' = \frac{|\Sigma|}{|\Sigma'|} \delta$.
\end{lemma}
\begin{proof}
    Since $N$ is $\delta$-bounded, $N_{i,j} \leq \delta$ for any pair $i\neq j \in \Sigma$. Therefore, by definition of the filter $F$, for any pair $i\neq j \in \Sigma'$, we have
    \[
        P_{i,j} \leq q \delta + \frac{r \delta}{|\Sigma'|} = q \delta + \frac{\pa{|\Sigma| - q |\Sigma'|} \delta }{|\Sigma'|} = \frac{|\Sigma|}{|\Sigma'|} \delta,
    \]
    and we conclude that $P$ is $\delta'$-bounded where $\delta' = \frac{|\Sigma|}{|\Sigma'|} \delta$.
\end{proof}
By \Cref{lemma:reducing_alphabet}, we obtain that $1-|\Sigma'|\delta' = 1-|\Sigma|\delta $. Given also the fact $|\Sigma|<C$ for some constant $C>0$, we obtain that  $\frac{\delta_p}{(1-|\Sigma'|\delta_p)} \leq C \frac{\delta}{(1-|\Sigma'|\delta)}$, and therefore the asymptotic convergence times in \Cref{thm:upper,thm:self} hold for any given alphabet $\Sigma$ of cardinality larger than the one deployed.

\section{Uniform Source Filter (USF)}
\label{sec:USF}
Algorithm \SF assumes that the system designer knows many of the system's parameters, namely: the noise matrix $N$ and constant approximations of $n$, the number of sources $s_0+s_1$, and the bias $s$. However, it is possible to weaken this assumption by repeatedly guessing estimates for the parameter $m$, which governs the duration of the phases, and hence controls the amount of information gathered in the phases.

Here, we introduce the Uniform Source Filter (\USF) that only requires agents to know the noise matrix $N$ and a constant approximation of $\log n$. The pseudo-code can be found in \Cref{alg:RSF}. In \USF, each agent runs multiple instances of \SF. Initially, each agent executes \SF for $\mathcal{O}(\log n)$ times using an optimistic guess for $m$, specifically $m^{(0)} = h$. Each agent then registers all output opinions, for each execution of \SF. If all these outputs of an agent are the same, the agent adopts that opinion. Then, each agent repeats the process, doubling the guess for $m$ from the last attempt, i.e., $m^{(i+1)} = 2 m^{(i)}$, and so on. The agents terminate when $m^{(i)}$ reaches the most pessimistic guess, i.e., whenever $m^{(i)} \geq m(n,s_0+s_1,s)$ from \Cref{eq:m for SF}, by assuming $s=1$, $s_0+s_1 = n$.

It is important to note that the output opinion of one instance of \SF is not used by other instances. This ensures that the outputs are independent of one another. Note also that, for any guessed value for $m$, the probability of an incorrect output opinion from a single execution of \SF is at most $1/2$. As a matter of fact, in the design of \SF, if the agents don't gather enough messages during phases 0 and 1, it is likely that none of those messages come directly from a source, and therefore most agents start the majority boosting phase with a random opinion.  Hence, if the guessed value for $m$ is too small, we do not expect the $\mathcal{O}(\log n)$ outputs of \SF to be consistent, and the agents will retry with a larger value of $m$. In contrast, if the estimate $m^{(i)}$ is at least as large as $m(n,s_0+s_1,s)$ (see \Cref{eq:m for SF}), then the outputs from the $\mathcal{O}(\log n)$ repetitions of \SF will be consistently correct w.h.p.
Since the running time of Algorithm \SF is $\mathcal{O}(m(n,s_0+s_1,s)/h)$ rounds, Algorithm \USF ensures convergence after $\mathcal{O}(m(n,s_0+s_1,s)/h) \log n$ rounds.

As mentioned, convergence of opinions does not imply that agents stop their actions.
Since agents do not have precise knowledge of the system's parameters, it is not clear how they can ensure that they terminate soon after they converge.
Instead, using their constant approximation of $\log n$, the agents can terminate after a polynomial time in $n$. Specifically, the protocol could safely terminate after the agents have attempted the most pessimistic guess of $m^{(i)} \geq m(n,s_0+s_1 = n, s = 1) = \Theta\pa{ \frac{n \log n}{(1-|\Sigma|\delta)^2} }$ from \Cref{eq:m for SF}. Using their constant approximation of $\log n$, the agents could know that they reach this state only after $\mathcal{O}(\frac{ poly(n)}{h (1-|\Sigma|\delta)^{2}}+ \log n)$ rounds.

\section{Pseudo-codes}
%\newcommand{\counter}{\texttt{Counter}}
%Pseudo-code of Algorithm \SF.
\begin{algorithm}[H]
    Communication alphabet~$\Sigma := \{ 0,1 \}$ \;
    %\alginput{Any set of $m$ messages received in the last $2m/h$ rounds.}
    %
    \alginput{Sample size~$h$, Parameter~$m$, phase duration $T := \lceil m/h \rceil$, sub-phase duration $w = \frac{100}{(1-2\delta)^2}$}
    %threshold number of opinion to collect in the third phase before an update  }
    %any deterministic choice of $m$ different messages received within the last $2m/h$ rounds \;
    %
    $\counter_0$, $\counter_1 \leftarrow 0$ \;
    \BlankLine

    \textbf{Phase 0:}
    \For{ $t=1,\ldots, T$ }{
        \lIf{Agent~$i$ is a source}{Display preference {\bf else} Display 0}
        S $\leftarrow$ observe the messages of $h$ random agents \;
        $\counter_1$ $\leftarrow$ $\counter_1$ + Number of $1$-messages in $S$ \;
    }
    \BlankLine
    \textbf{Phase 1:}
    \For{ $t=T+1,\ldots, 2T$ }{
        \lIf{Agent~$i$ is a source}{Display preference {\bf else} Display 1}
        S $\leftarrow$ observe the messages of $h$ random agents \;
        $\counter_0$ $\leftarrow$ $\counter_0$ + Number of $0$-messages in $S$ \;
    }
    \BlankLine
    $\Y^{(i)} \leftarrow \mathds{1}\{ \counter_1 > \counter_0 \}$ (breaking ties randomly) \tcp*{Weak opinion}
    $Y^{(i)} \leftarrow \Y^{(i)}$ \tcp*{Opinion}
    $M \leftarrow \emptyset$ \tcp*{Memory}
    $\ell \leftarrow 0$ \tcp*{Index of the sub-phase}
    \BlankLine
    \textbf{Majority Boosting phase:}
    \While{ $\ell \leq 10 \, \log n + 1$ }{
        Display opinion $Y^{(i)}$ \;
        S $\leftarrow$ observe the messages of $h$ random agents \;
        Add all messages of~$S$ into~$M$ \;
        \If{$\ell \leq 10 \, \log n$ and $|M| \geq w$ {\bf or} $\ell = 10 \, \log n+1$ and $|M| \geq m$ }{
            \tcp{All sub-phases (except the final one) last $\lceil w/h \rceil$ rounds}
            $Y^{(i)}\leftarrow$ majority opinion in~$M$ (breaking ties randomly) \;
            $M \leftarrow \emptyset$ \;
            $\ell \leftarrow \ell +1$ \;
        }
    }
    %$Y^{(i)} \leftarrow$ majority opinion in~$M$ (breaking ties randomly) \tcp*{Final opinion is computed at the end of the last sub-phase, which lasts for about~$T/2$ rounds}
    \caption{Source Filter (\SF) for all Agents~$i=1,\ldots,n$, including the sources}
    \label{alg:SF}
\end{algorithm}

\begin{algorithm}[H]
    Communication alphabet~$\Sigma := \{ 0,1 \}$ \;
    %\alginput{Any set of $m$ messages received in the last $2m/h$ rounds.}
    %
    \alginput{Sample size~$h$, sub-phase duration $w = \frac{100}{(1-2\delta)^2}$, repetitions of SF $r=\log n$, worst guess m*}
    %threshold number of opinion to collect in the third phase before an update  }
    %any deterministic choice of $m$ different messages received within the last $2m/h$ rounds \;
    %
    $m \leftarrow h$ \\
    \SetKwRepeat{Do}{do}{while}
    \Do{ $m < m$*  }
    {
        $m \leftarrow 2*m$\\
        $T \leftarrow \lceil m/h \rceil$\\
        opinion\_register $\leftarrow$ empty array of size $\log n$\\
        \For{$i=1,\ldots,r$}
        {
            opinion\_register[i] $\leftarrow$ SF$( h, m, T , w)$
        }
        \If{opinion\_register $= (0,\ldots, 0)$ \textbf{or} opinion\_register $= (1,\ldots, 1)$}
        {
        $Y^{(i)}\leftarrow$ opinion\_register[1];
        }
    }
    %$Y^{(i)} \leftarrow$ majority opinion in~$M$ (breaking ties randomly) \tcp*{Final opinion is computed at the end of the last sub-phase, which lasts for about~$T/2$ rounds}
    \caption{Uniform Source Filter (\USF) for all Agents~$i=1,\ldots,n$, including the sources}
    \label{alg:RSF}
\end{algorithm}

%\clearpage
%Pseudo-code of Algorithm \SSF.

\section{Simulations}
\label{sec:simulations}
In this section, we present and compare the performance of our protocols \SF and \SSF across several scenarios. Specifically, we compare their convergence times while sweeping over the sample size (\Cref{fig:convergence_vs_h}), the noise level (\Cref{fig:convergence_vs_delta}), and the number of sources (\Cref{fig:convergence_vs_sources}), in order to experimentally validate our theoretical asymptotic results.
Interestingly, despite having a more favorable asymptotic upper bound, \SF turns out to perform considerably worse than \SSF in many scenarios.

In \Cref{fig:convergence_vs_h,fig:convergence_vs_delta,fig:convergence_vs_sources}, we compare the theoretical bounds derived in \Cref{thm:upper,thm:self} with simulation results for both protocols.
For each combination of parameters, we average the convergence time\footnote{The convergence time is the first time at which all agents agree on the correct opinion. We use this measure rather than the time required to complete all protocol phases, as convergence may occur earlier. For both \SF and \SSF, the convergence time is a stochastic quantity, since agents can adopt the correct opinion before all phases are completed.} over~100 executions of \SF and \SSF.
Due to the stochastic nature of the protocols, we exclude runs in which the agents do not converge to the correct opinion; the fraction of such failed experiments is always below 5\% of the total.
We fit the theoretical bounds from \Cref{thm:upper,thm:self}, of the form $C\cdot m/h$, to the experimental data by choosing $C$ so as to minimize the squared error between the observed average convergence times and $C \cdot m/h$.

For \SF, we optimized the constants governing the number of subphases in the third phase and the duration of each subphase using the hyperparameter optimization framework \textsc{Optuna}, with the objective of minimizing the convergence time and imposing a penalty when the agents did not converge by the end of the execution. The search was performed over predefined ranges for the number of subphases and the duration of each subphase, and the same optimized parameters were then used across all experimental scenarios. The optimization was carried out with the remaining parameters fixed, namely $n=100$, $\delta=0.15$, and $h=10$.
Moreover, to ensure a fair comparison between the two protocols, we run \SF with an alphabet of size~4 and apply the reduction described in \Cref{sec:reduction alphabet}.

For \SSF, we initialize the memory of each agent to be empty, effectively synchronizing the updates of all agents, in order to allow for a fair comparison with \SF, which is not quasi-self-stabilizing. The impact of different initializations is shown instead in \Cref{fig:selfstab_plots}.

\begin{figure}[H]
    \centering
    \includegraphics[width=0.5\textwidth]{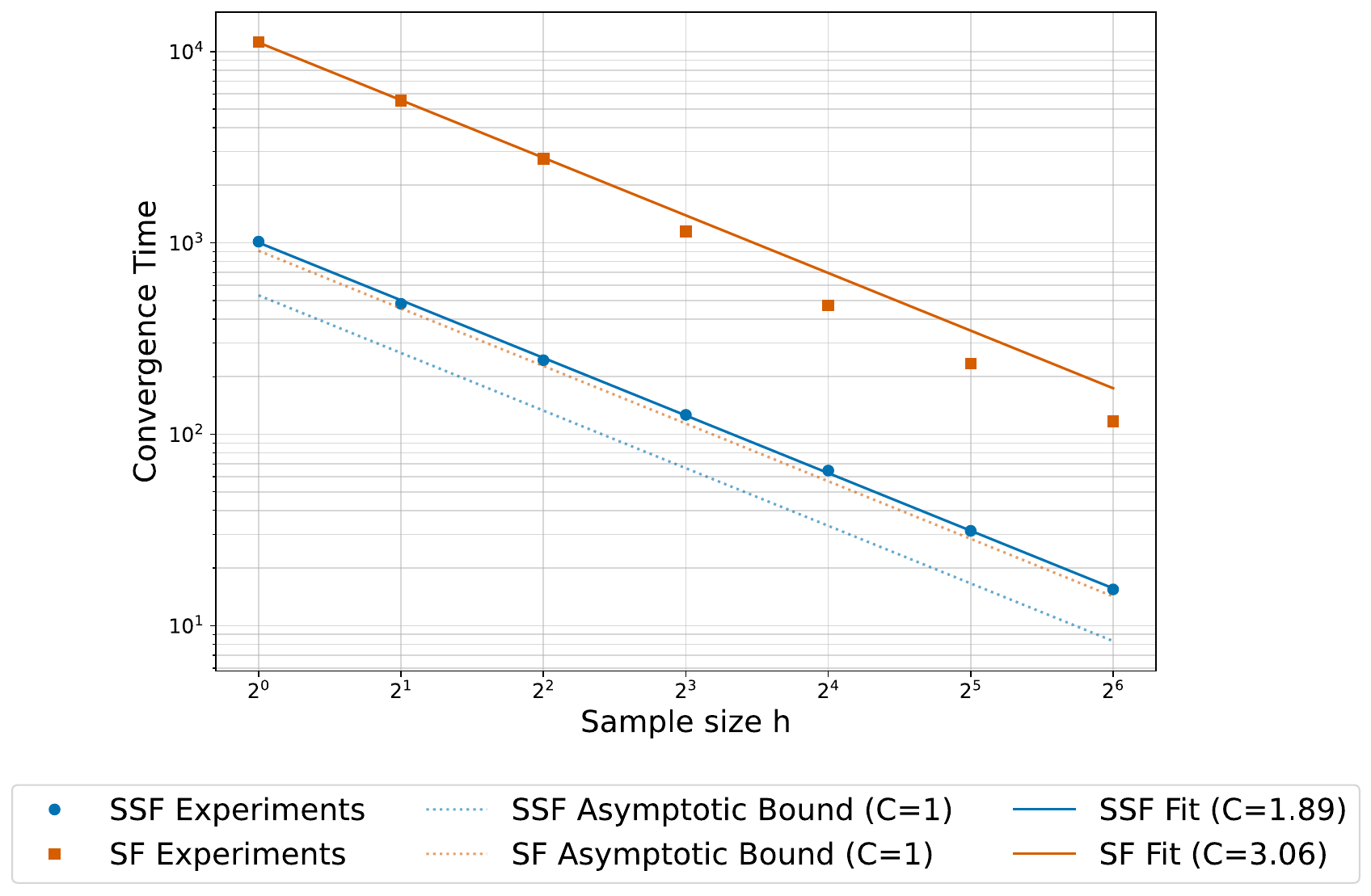}
    \caption{Convergence times of SF and SSF as a function of sample size: comparison of simulation results and theoretical bounds. The simulation parameters are $n=100$, $s_1=s=1$, and $\delta=0.15$.}
    \label{fig:convergence_vs_h}
\end{figure}

\begin{figure}[H]
    \centering
    \includegraphics[width=0.5\textwidth]{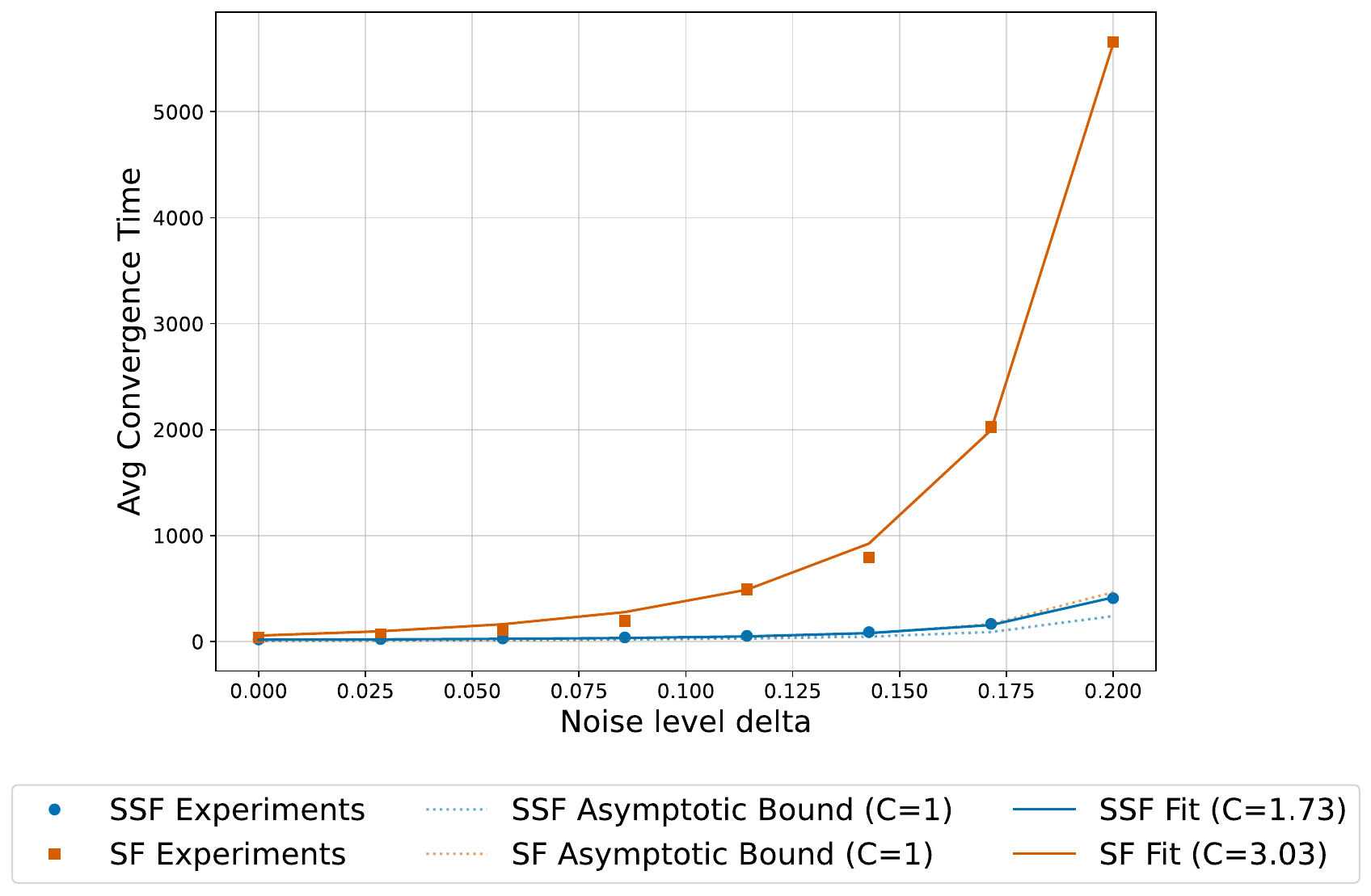}
    \caption{Convergence times of SF and SSF as a function of noise level: comparison of simulation results and theoretical bounds. The simulation parameters are $n=100$, $s_1=s=1$, and $h=10$. }
    \label{fig:convergence_vs_delta}
\end{figure}

\begin{figure}[H]
    \centering
    \includegraphics[width=0.5\textwidth]{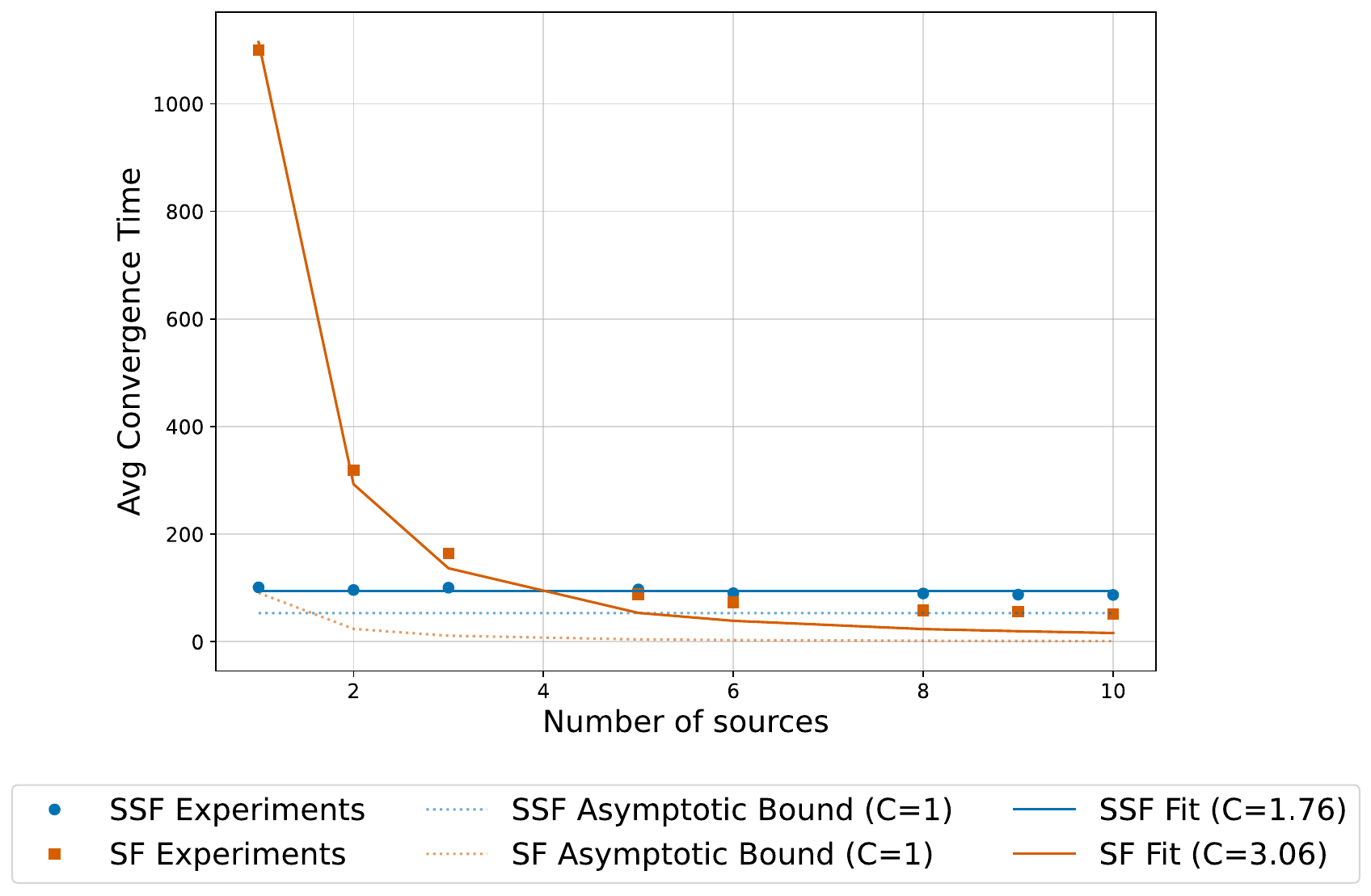}
    \caption{Convergence times of SF and SSF as a function of the number of sources: comparison of simulation results and theoretical bounds. The simulation parameters are $n=100$, $\delta=0.15$, and $h=10$.}
    \label{fig:convergence_vs_sources}
\end{figure}

The observed performance gap between \SF and \SSF can be explained by the different structure of the two protocols.
Protocol \SSF performs a single majority computation over the $m$ collected weak opinions, and thus requires collecting $m$ messages to form the weak opinions and an additional batch of $m$ messages to compute the final opinion, for a total of $2m$ samples.
In contrast, \SF requires $m$ samples to form the weak opinions, additional messages to perform the majority-boosting phase, and a final batch of $m$ samples for the concluding majority step, which is necessary to guarantee correctness with high probability. Overall, \SF therefore requires strictly more than $2m$ samples.
When the noise level $\delta$ and the number of sources are fixed, the values of $m$ prescribed by the theoretical bounds for \SF and \SSF are comparable, so the additional majority-boosting phase effectively slows down \SF in practice. This explains why \SSF consistently outperforms \SF in these regimes.
On the other hand, as the number of sources increases, the theoretical bounds predict that the required sample size for \SF decreases more rapidly, making the overhead of the boosting phase less dominant. This is consistent with the trend observed in \Cref{fig:convergence_vs_sources}, where \SF becomes comparatively more competitive, confirming our initial expectation.

\section{Probabilistic Tools} \label{sec:probabilistic_tools}

\begin{theorem} [Multiplicative Chernoff's bound] \label{thm:multiplicative_chernoff_bound}
    Let~$X_1,\ldots,X_n$ be i.i.d. random variables taking values in $\{0,1\}$, let~$X = \sum_{i=1}^n X_i$ and~$\mu = \bbE(X) = n \Pr(X_1 = 1)$. Then it holds for all~$\delta \in (0,1)$ that
    \begin{equation*}
        \Pr \pa{ X \leq (1 - \delta) \mu } \leq \exp\pa{-\frac{\delta^2 \, \mu}{2}}.
    \end{equation*}
\end{theorem}

The proof of the following result can be found in, e.g., \cite{dubhashi1996balls}.
\begin{theorem} [Chernoff-Hoeffding's bound] \label{thm:additive_chernoff_bound}
    Let $X_1,\dots,X_n$ be independent random variables such that $a_i\leq X_i\leq b_i$ almost surely, let~$X = \sum_{i=1}^n X_i$ and~$\mu = \bbE(X)$.
    Then, for all $\delta>0$,
    \begin{equation*}
        \Pr\pa{X\leq \mu - \delta}, \Pr\pa{X\geq \mu + \delta} \leq \exp\pa{-\frac{2\delta^2}{\sum_{i=1}^n (b_i-a_i)^2}}
    \end{equation*}
    In particular, if~$X_1,\ldots,X_n$ take values in $\{0,1\}$, then
    \begin{equation*}
        \Pr \pa{ X \leq \mu - \delta } , \Pr \pa{ X \geq \mu + \delta } \leq \exp \pa{ -\frac{ 2 \delta^2}{n} }.
    \end{equation*}
\end{theorem}
\section{Implications to house-hunting by ants}
\label{house-hunting}
%We consider the task of spreading a binary value from a set of (possibly conflicting) {\em source} agents to the rest of the population, under the noisy $\pull(h)$ model of communication, where the goal is to converge to the plurality opinion of the sources.

%When the sample size $h$ is small, the lower bound proved in \cite{boczkowski2018limits} becomes linear in the group size. This result allowed the authors to predict that recruiting a few desert ants among a population, a process which could be viewed as an information spreading process under the noisy $\pull(1)$ model, would be slower when more ants are involved. This counter-intuitive prediction was indeed verified empirically in \cite{boczkowski2018limits}. As a general message, the authors in \cite{boczkowski2018limits} used the linear lower bound to argue that in the case of pairwise-meeting patterns, the loss of communication structure can severely weaken the system's ability to effectively counteract noise. While a small sample size is evident in some biological scenarios \cite{razin}, it is inconsistent with various other natural scenarios where agents extract information from numerous other agents simultaneously (see more details in \Cref{sec:relatedworks}).  This paper provides upper bounds for the information spreading time for every $h$, which match the bounds proved in \cite{boczkowski2018limits} up to lower order terms, thus demonstrating how a larger sample size can linearly accelerate the information spreading time.

In this section we interpret the process of {\em house-hunting} in ants through the lens of our findings.
House-hunting is observed in ant species such as {\em Temnothorax}. When their nest is damaged, these ants embark on a house-hunting process to find and select a new site for their colony \cite{DBLP:conf/podc/GhaffariMRL15,pratt2002quorum}. Suppose the possible nest locations have already been identified, and the remaining task is to reach a consensus on the best site. This process can be divided into two phases. In the first phase, scout ants are sent to evaluate the quality of different sites --- a process that can be expected to be noisy. Upon completing an evaluation, a scout ant returns to the nest and begins a {\em tandem run}, leading another ant from the nest to the site. The leader moves slowly enough to keep the follower engaged, with regular pauses to ensure the follower remains close. Communication occurs through frequent physical contact, such as antennal touching. This allows the leader to adjust pace and direction, maintaining coordination with the follower, which can in turn, learn the path to the new site. When reaching the site, the newly recruited ant assesses it herself and, upon returning to the nest, repeats the tandem run process. The rate of tandem runs for a particular site depends on its quality, but each tandem run is time-consuming delaying the assessment phase. While relaying quality estimates between ants could have sped up the process, it might have also introduced errors due to the difficulty of accurately communicating those measurements. Instead, tandem runs, though slow, allow ants to gather first-hand information by evaluating the sites themselves.
In the language of this paper, one may interpret the ants' strategy as investing time (performing tandem runs) to increase the number of sources (and thereby the bias) instead of attempting to relay estimates. The second phase involves a quorum-sensing mechanism, where if a sufficient number of ants are present at a site, that site is chosen by the colony. This phase bears similarities to the majority-consensus phase.

\bibliography{references}

\end{document}